\documentclass[onecolumn, 11pt]{IEEEtran}

\usepackage{macro}

\title{Tight List-Sizes for Oblivious AVCs under Constraints}

\author{
\IEEEauthorblockN{
Yihan Zhang\IEEEauthorrefmark{1}, 
Sidharth Jaggi\IEEEauthorrefmark{2}\IEEEauthorrefmark{1},
Amitalok J. Budkuley\IEEEauthorrefmark{3}
}\\
\IEEEauthorblockA{
\IEEEauthorrefmark{1}Dept.\ of Information Engineering, The Chinese University of Hong Kong \\
\IEEEauthorrefmark{2}School of Mathematics, University of Bristol \\
\IEEEauthorrefmark{3}Dept.\ of Electronics and Electrical Communication Engineering, Indian Institute of Technology Kharagpur
}
}

\begin{document}
\maketitle


\begin{abstract}
We study list-decoding over adversarial channels governed by oblivious adversaries (a.k.a. oblivious Arbitrarily Varying Channels (AVCs)). This type of adversaries aims to maliciously corrupt the communication without knowing the actual transmission from the sender. For any oblivious AVCs potentially with constraints on the sender's transmitted sequence and the adversary's noise sequence, we determine the exact value of the minimum list-size that can support a reliable communication at positive rate. This generalizes a classical result by Hughes (IEEE Transactions on Information Theory, 1997) and answers an open question posed by Sarwate and Gastpar (IEEE Transactions on Information Theory, 2012). A lower bound on the list-decoding capacity (whenever positive) is presented. Under a certain combinatorial conjecture, we also prove a matching upper bound. 
En route to a tight characterization of the list-decoding capacity, we propose a method for subcode construction towards the resolution of the combinatorial conjecture.
\end{abstract}

\section{Introduction}
Coding against adversaries is one of the central subjects in coding theory and information theory. 
The canonical model of interest is as follows. 
Suppose a transmitter Alice would like to send a message to a receiver Bob through a noisy channel. 
The channel is governed by an adversary (also called a jammer) James whose aim is to \emph{maliciously} corrupt the data transmission from Alice to Bob. 
Coding is a technique to protect data from corruption and thus ensures a reliable communication.
Alice, instead of sending the message per se, introduces some redundancy to the message and sends an \emph{encoded} sequence (called \emph{codeword}). 
James  carefully designs a \emph{jamming sequence} (historically a.k.a. \emph{state sequence}) and transmits it to the channel. 
The hope is that even James adversarially distorted Alice's transmitted codeword via the channel, Bob receiving the noisy channel output is still able to reliably \emph{decode}. 
The goal is to find the maximum information \emph{throughput} (a.k.a. \emph{rate}) of a given channel, i.e., the largest number of bits that could be delivered with a low probability of reconstruction error. 
This paper is concerned with characterizing information-theoretic {fundamental limits} (called \emph{capacity}) of this kind assuming there is \emph{no} restriction on the computational power of Alice/Bob/James. 
See \Cref{sec:prelim} for formal definitions.

The aforementioned adversarial communication channel model is often referred to as the \emph{Arbitrarily Varying Channel} (AVC) model, first introduced by Blackwell, Breiman and Thomasian \cite{blackwell-1960-avc}.
It contains a large family of channels of interests and has received a lot of study since it was proposed. 
It turns out that in an adversarial communication problem, it is important to clarify the power of James. 
The knowledge that James possesses when designing the jamming vector is an important component of the problem setup. 
There are two natural and popular models that are well-studied in the literature.
The \emph{omniscient} model assumes that James knows precisely which particular codeword was transmitted by Alice.
Hence he could look at the transmission and tailor his jamming strategy for a specific instantiation of transmission. 
This is a very strong type of adversaries and the capacity of such AVCs is widely open. 
Another model called the \emph{oblivious} model instead assumes that the adversary does not know at all which codeword was transmitted by Alice. 
Said differently, one can think of James as choosing a jamming sequence \emph{before} Alice's transmission is instantiated. 
Much more capacity results are known in this setting. 
The oblivious model is the main focus of this paper. 
See \cite{lapidoth-narayan-it1998} for a 
survey on AVCs and see \Cref{sec:prior} for prior works. 

We now shift our attention from James to Bob.
The reliability of communication can be measured in various different ways. 
The metric that was alluded to in the first paragraph of this section is called \emph{unique-decodability}. 
As the name suggests, Bob is required to, based on his received vector, decode to a single message which is hopefully the correct one that Alice meant to deliver.
In some scenarios, this is too stringent a requirement to satisfy or it is much trickier to directly prove without intermediate steps. 
One possible notion of reliability that  relaxes this condition is called \emph{list-decodability}. 
This allows Bob to output a \emph{list} of messages which should contain the correct one. 
Of course, to avoid triviality, we would like the list to be as small as possible.
The notion of list-decoding was first proposed by Elias \cite{elias-list-dec} and Wozencraft \cite{wozencraft-list-dec}. 
It has since received a lot of study from various aspects and has become an important subject within and outside the scope of Coding \& Information Theory. 
In Computer Science, list-decoding is most well-studied under the omniscient bitflip channel model.
See \cite{guruswami-lncs,guruswami-fnt} for a comprehensive survey on combinatorial and algorithmic results on list-decoding of this flavour. 
List-decoding also serves as a primitive in complexity theory and cryptography, e.g., \cite{doron-2019-pseudorand-hard-listdec}. 
Recently, the idea of relaxing the problem by allowing the solver to return a list of answers rather than a unique answer also goes into the development of robust statistics \cite{diakonikolas-2018-listdec-estimation,diakonikolas-2020-listdec-estimation,mohanty-2020-listdec-estimation,kothari-2019-listdec-regression,bakshi-kothari-2020-listdec-recovery,raghavendra-yau-2020-listdec-recovery,raghavendra-yau-2020-listdec-learning}. 
In our context, given an adversarial channel, one of the fundamental questions is to pin down the smallest list-size $L^* $ (a.k.a. \emph{list-symmetrizability}) that supports a reliable communication at positive rate. 
Understanding capacity positivity is the first step towards understanding the exact value of the capacity of an AVC under list-decoding. 
In this paper, for any oblivious AVC, we give a tight characterization on the exact value of the minimum list-size such that the list-decoding capacity w.r.t. such a list-size is strictly positive. 
See \Cref{sec:overview_results_techniques} for an overview of our results. 

Finally, as a technical note, we emphasize that in this paper we use the \emph{average  probability of error}\footnote{The average refers to averaging over messages that Alice can potentially transmit.
They are assumed to be uniformly distributed.} criterion, i.e., a $ 1-o(1) $ (but not exactly 1) fraction of messages are required to be correctly decoded.\footnote{There exist other notions of error criteria which are also interesting and have been studied in the literature. 
One other important one is the \emph{maximum} probability of error criterion. 
Again, the maximization is over messages. 
Under  deterministic encoders, the problem is equivalent to the omniscient case since we can assume that the adversary knows the transmitted message and hence the transmitted codeword. 
Under stochastic encoders (which is possible if Alice has access to private randomness), it can be shown that the capacity remains the same as that under the average probability of error criterion.}
Also, we do not assume common randomness available to Alice and Bob (but secret to James). 
AVC problems under other probability of error criteria and/or in the presence of common randomness are also interesting. 
See \Cref{sec:prior} for related works.

\section{Prior work}
\label{sec:prior}
\subsection{Oblivious AVCs}
\label{sec:prior_obli}
For AVCs, one of the challenges of characterizing the capacity is perhaps to first determine whether the capacity is zero or strictly positive. 
For general AVCs, finding a sufficient and necessary condition for achieving positive rate turns out to be highly nontrivial. 
Such conditions are usually called \emph{symmetrizability} in the literature. 
We explain the underlying intuition using the oblivious \emph{bitflip} channels which are  perhaps the simplest nontrivial example of oblivious adversarial channels. 
\begin{figure}[htbp]
	\centering
	\includegraphics[width=0.7\textwidth]{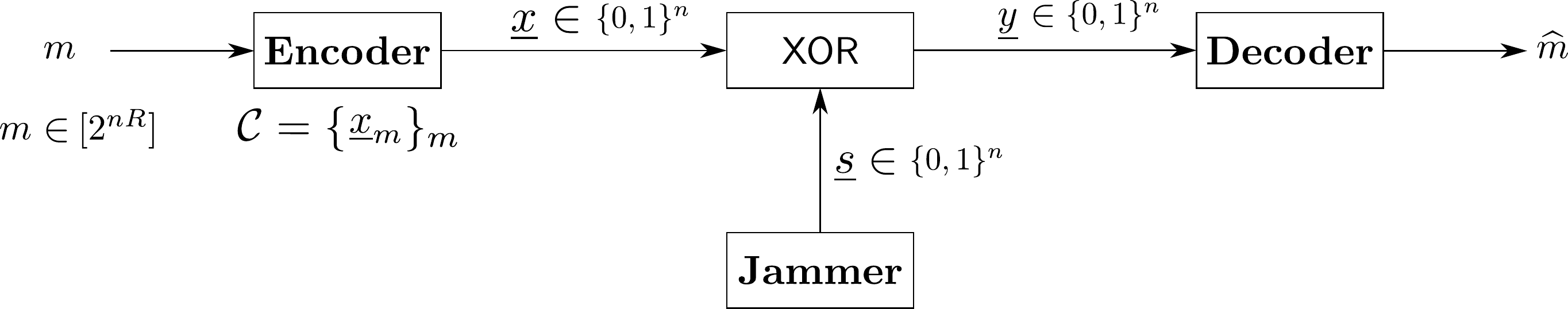}
	\caption{Block diagram of an oblivious bitflip channel.}
	\label{fig:diag_obli_bitflip}
\end{figure}
As shown in \Cref{fig:diag_obli_bitflip},
both the transmitted codeword $ \vbfx $ and the jamming sequence $ \vbfs $ are $ \zon $-valued. 
The channel in this case is simply a deterministic function that adds up $ \vbfx $ and $ \vbfs $ component-wise modulo 2. 
The channel output is therefore $ \vbfy = \vbfx\oplus\vbfs $. 
Though James does not know the value of $ \vbfx $, he could transmit a uniformly random codeword $ \vbfx'\sim\cC $ from the codebook (which is assumed to be known to everyone). 
If $ \vbfx'\ne\vbfx $ (which does happen with high probability as long as the code is large), Bob receives $ \vbfy = \vbfx\oplus\vbfx' $ and could not distinguish whether Alice transmitted $ \vbfx $ or $ \vbfx' $. 
In other words, the scenario where Alice transmits $ \vbfx $ and James transmits $ \vbfx' $ and the scenario where Alice transmits $ \vbfx' $ and James transmits $ \vbfx $ result in the same statistics at the output end of the channel, i.e., $ P_{\vbfy|\vbfx,\vbfx'} = P_{\vbfy|\vbfx',\vbfx} $\footnote{We use boldface lowercase letters to denote (scalar) random variables, plain lowercase letters to denote their realizations and underlines to denote vectors of length $n$, where $n$ is the blocklength of the code unless otherwise specified. See \Cref{sec:notation} for the notational convention  followed in this paper.}.
We then say that in this case James \emph{symmetrizes} the channel. 
This forces Bob to fail to decode  the correct message with constant probability and hence forces the capacity to be zero.  
Symmetrizability is a fundamental characteristic associated to an AVC.
Given a generic AVC, 
it is not obvious how to capture this kind of phenomena using a precise notion of symmetrizability.
The seminal work by \csiszar and Narayan \cite{csiszar-narayan-it1988-obliviousavc} provided the right notion of symmetrizability of a general oblivious AVC and used it to give a characterization of the capacity. 
\cite{csiszar-narayan-it1988-obliviousavc}'s notion of symmetrizability  is a certain linear algebraic condition that can be easily verified for any given oblivious AVC. 
Remarkably, such a condition was not available previously even for AVCs \emph{without} input or state constraints \cite{ahlswede-1978-unconstr-oblivious-avc-common-rand}, i.e., Alice and James are allowed to transmit \emph{any} length-$n$ sequence over their alphabets. 

The quadratically constrained version of the oblivious AVC problem is also well-studied. 
In this variant,  vectors in the communication system are $ \bR^n $-valued  subject to $ \ell^2 $-norm constraints. 
The capacity of such channels (a.k.a. \emph{Gaussian AVCs}) was obtained in \cite{csiszar-narayan-it1991} via a geometric approach.

For list-decoding in the oblivious setting, without input and state constraints, the notion of list-symmetrizability (formally defined in \Cref{def:list_symm}) was given in a paper by Hughes \cite{hughes-1997-list-avc} and the $L$-list-decoding capacity (formally defined in \Cref{def:ach_rate_listdec_cap}) for any list-size $ L\in\bZ_{\ge1} $ was characterized therein. 
As we shall see in more details in \Cref{sec:overview_results_techniques}, incorporating constraints (especially state constraints) into the definition of symmetrizability (and list-symmetrizability) is nontrivial. 
Intuitively, this is because symmetrizability is a symbol-wise notion.
For simplicity, we explain the effect of constraints using unique-decoding as an example.
If a pair of input symbols $ (x,x')\in\cX^2 $ are ``confusable'' in the symmetrizability sense that $ P_{\bfy|x,x'} = P_{\bfy|x',x} $, then the encoder should avoid using them,
since $ x $ and $ x' $ will cause confusion to Bob and receiving symbol $\bfy$ in certain locations, he not could tell if the value of the original transmission was $x$ or $x'$ in these locations.
If all pairs of input symbols are confusable, then Alice has no effective symbols that can be used for communication without causing confusion at the decoder end. 
Hence the symbol-wise notion of symmetrizability translates to an operational notion of confusability at the level of \emph{vectors} ($ P_{\vbfy|\vbfx,\vbfx'} = P_{\vbfy|\vbfx',\vbfx} $). 
However, in the presence of state constraints,  even if all pairs of symbols are confusable, it does not immediately mean that James can always confuse Bob w.r.t two different codewords.
In fact,  this is known to be not true \cite{csiszar-narayan-it1988-obliviousavc}.
In the oblivious bitflip example, if we further impose constraints on $ \vx $ and $ \vs $ such that $ \wth{\vx}\le w $ and $ \wth{\vs}\le p $ for some $ w,p\in[0,1/2] $, where $ \wth{\cdot} $  denotes the Hamming weight of a vector, then the capacity is positive if (and only if) $ w>p $.
However,  the only two input symbols 0 and 1 are apparently confusable. 
This phenomenon is precisely due to the effect of state constraints.
If James is only allowed to jam using \emph{some} rather than \emph{arbitrary} sequences, then symbol-wise symmetrizability does not rule out the possibility that Alice could still communicate to Bob at a positive rate while being robust to all James's \emph{feasible} jamming sequences. 
Indeed, it was first observed by Ericson \cite{ericson-1985-unconstrained-impractical} that unconstrained AVCs are impractical models since many channels of interests are symmetric (e.g., the bitflip channels discussed above) and hence by definition is symmetrizable and has zero capacity. 
Constrained oblivious AVCs under list-decoding were studied by Sarwate and Gastpar in \cite{sarwate-gastpar-2012-list-dec-avc-state-constr}. 
They gave upper and lower  bounds on $ L^* $ -- the smallest list-size $L$ such that the $L$-list-decoding capacity is positive. 
In particular, they defined the notions of \emph{weak} and \emph{strong} list-symmetrizability such that $ L_\strong^*\le L^*\le L_\weak^* $. 
Using these notions, they prove natural lower  and upper bounds on the $L$-list-decoding capacity. 

The quadratically constrained version of the oblivious AVC problem under list-decoding turns out to be easier due to the Euclidean nature.
Hosseinigoki and Kosut recently \cite{hosseinigoki-kosut-2018-oblivious-gaussian-avc-ld} characterized the $L$-list-decoding capacity of such channels. 
The achievability uses typicality methods for real-valued vectors. 
A natural list-symmetrization strategy was analyzed using a certain novel bounding technique.

\subsection{Omniscient AVCs}
\label{sec:prior_omni}
Much less is known in the omniscient setting.
Let us first look at capacity positivity. 
For omniscient AVCs without input/state constraints, it is well-known that the capacity is zero if and only if the channel is symmetrizable in the sense of Kiefer and Wolfowitz \cite{kiefer-wolfowitz-1962-omniscient_confusability_unconstrained}. 
Again, incorporating constraints is nontrivial in the omniscient case as well.
Until very recently, Wang, Budkuley, Bogdanov and Jaggi \cite{wbbj-2019-gen_plotkin} managed to get the right notion of symmetrizability (a.k.a. \emph{confusability} in their paper) for general omniscient AVCs with input \& state constraints whose channel transition  distribution is a 0-1 law, i.e., the channel output is a deterministic function of Alice's and James's inputs. 
They showed that such channels have zero capacity if and only if they are symmetrizable in \cite{wbbj-2019-gen_plotkin}'s sense. 
Their results are a significant generalization of the Plotkin bound in classical coding theory. 
Remarkably, the connection between the \emph{complete positivity} of joint distributions and the structure of codes was first introduced in \cite{wbbj-2019-gen_plotkin}. 
With additional work, their techniques can potentially carry over to the case where $ W_{\bfy|\bfx,\bfs} $ is an arbitrary distribution. 
However, a formal proof has not been presented yet. 

We now turn to bounds on the capacity. 
As we shall see below,  the situation in this direction is rather sad and 
essentially no capacity result is known even for very simple channels.
For general omniscient AVCs without input/state constraints, \csiszar and \korner \cite{csiszar-korner-1981} proved a lower bound on the capacity using techniques similar to those used in \cite{csiszar-narayan-it1988-obliviousavc} for  oblivious channels. 
For channels with input \& state constraints and with 0-1 transition distributions, a Gilbert--Varshamov-type lower bound and its natural generalization using time-sharing were presented in \cite{wbbj-2019-gen_plotkin}; 
an Elias--Bassalygo-type upper bound will be presented in the full version of \cite{wbbj-2019-gen_plotkin} which is not yet available when the present paper is written.
In classical coding theory, the most commonly considered model is the omniscient bitflip channel.
The capacity of such a channel is equivalent to the largest sphere packing density in Hamming space.
The best lower and upper bounds so far are the Gilbert--Varshamov bound \cite{gilbert-gv,varshamov-gv} and the second MRRW (Linear Programming-type) bound \cite{mrrw2}, respectively.
They match nowhere except when the channel is completely noiseless or noisy, which are trivial cases. 
The exact value of the capacity is perhaps \emph{the} most central open question in coding theory. 

The quadratically constrained version of the channel coding problem for omniscient AVCs is equivalent to packing caps on a sphere. 
In the latter problem, one would like to pack as many as possible mutually disjoint spherical caps of radii $ \sqrt{nN} $ on a sphere of radius $ \sqrt{nP} $. 
It is well-known that one can achieve positive packing density whenever $ N/P<1/2 $. 
The current best lower and upper bounds on the packing density (or the capacity of a quadratically constrained omniscient AVC) are the Gilbert--Varshamov-type lower bound due to Blachman \cite{blachman-1962} and the Linear Programming-type upper bound due to Kabatiansky and Levenshtein \cite{kabatiansky-1978}, respectively. 
Narrowing the gap from either side is regarded as a big breakthrough.

For list-decoding against omniscient adversaries, the situation is similar. 
For general channels with input \& state constraints and with 0-1 transition distributions, Zhang, Budkuley and Jaggi \cite{zbj-2019-generalized-ld} recently extended the results of \cite{wbbj-2019-gen_plotkin} to list-decoding and provided a sufficient and necessary condition for positivity of the $L$-list-decoding capacity for any given list-size $ L\in\bZ_{\ge1} $. 
In the same paper, the authors also gave lower bounds on the $L$-list-decoding capacity using random coding with expurgation. 
An Elias--Bassalygo-/Johnson-type upper bound will also be presented in the full version of \cite{zbj-2019-generalized-ld} (which is not yet available when the present paper is written). 
These lower and upper bounds for general omniscient AVCs are generalizations of similar bounds for the bitflip channels due to a sequence of beautiful works by Blinovsky \cite{blinovsky-1986-ls-lb-binary,blinovsky-2005-ls-lb-qary,blinovsky-2008-ls-lb-qary-supplementary}. 
Blinovsky's upper bounds were improved by Polyanskiy \cite{polyanskiy-2016-list_dec_ub} in the high-rate regime for odd list-sizes at least 3. 

The quadratically constrained version of list-decoding is equivalent to  packing caps on a sphere with bounded multiplicity of overlap (a.k.a. \emph{multiple packing}).
To the best of the authors' knowledge, the current best lower and upper bounds are due to Blinovsky \cite{blinovsky-1997-multiple_packing_lb} and Blinovsky--Litsyn \cite{blinovsky-litsyn-2009-multiple_packing_ub}, respectively. 
The largest multiple packing density also remains elusive.

\subsection{Zero-error information theory}
\label{sec:prior_zero_error}
The omniscient model is intimately related to ``noiseless'' channels under zero-error criterion.
It can be shown that for omniscient channels, without loss of rate, all error criteria are equivalent to the zero-error criterion.
That is, the capacity under average/maximum error criterion is the same as that with the requirement that the decoder makes an error with probability \emph{precisely} (rather than asymptotically) zero.
The problem of characterizing the capacity is hence of combinatorial nature.
If the channel is ``noiseless'' in the sense that the adversary is absent, then the problem falls into the realm of zero-error information theory \cite{shannon-1956-zero_error}.
Though without noise, the zero-error capacity of a  channel $ W_{\bfy|\bfx} $ is still widely open.
It is well known to be equal to $ \lim_{n\to\infty}\alpha\paren{\cG(W_{\bfy|\bfx})^{\boxtimes n}}^{1/n} $, where $ \cG(W_{\bfy|\bfx}) $ is the confusability graph of $ W_{\bfy|\bfx} $, $ \cG^{\boxtimes n} $ denotes the $n$-fold strong product of $\cG$ and $ \alpha(\cG) $ denotes the independence number of $\cG$. 
This characterization is not single-letter and hence is not computable (since it involves a limit as the blocklength $n$ of the code grows). 
This formula is only successfully evaluated for sporadic nontrivial channels, e.g., the noisy typewriter channel with alphabet size five \cite{lovasz-1979-shannon_capacity_graph}.
Determination of the zero-error capacity of even the noisy typewriter channel with alphabet size seven remains formidable, let alone general channels.
Nevertheless, we emphasize that the capacity positivity of this problem \emph{is} understood.
Indeed, it is well-known and not hard to see that the capacity is positive if and only if the confusability graph of the channel is not a complete graph.
See \cite{korner-orlitsky-1998-zeroerror} for a survey on zero-error information theory.

\subsection{Other exotic models}
\label{sec:prior_other}
Other types of adversaries such as myopic adversaries \cite{dey-sufficiently-2015,zhang-2018-myopic,bdjlsw-2020-myopic_symm,djlsw-2019-myop_causal}, 
causal/online adversaries \cite{djls-causal-it2013,chen-2015-online,ldjls-2018-quadratic-causal}, 
adversaries with delay, adversaries with lookahead \cite{djls-2016-delay-lookahead}; 
and other types of channels such as channels with feedback \cite{berlekamp-thesis,zigangirov-feedback,ahlswede-cai-feedback,hkv-partial-feedback}, 
channels with common randomness \cite{csiszar-narayan-1988-oblivioius-avc-common-rand,ahlswede-1978-unconstr-oblivious-avc-common-rand}, 
two-way adversarial channels \cite{jaggi-langberg-2017-two-way,zhang-2020-twoway}, 
adversarial broadcast channels \cite{pereg-steinberg-2017-broadcast,hosseinigoki-kosut-2020-broadcast}, 
adversarial interference channels \cite{hosseinigoki-kosut-2016-gaussian-interference}, 
adversarial relay channels \cite{pereg-steinberg-2019-relay,pereg-steinberg-2018-gaussian-relay}, 
adversarial Multiple Access Channels (MACs) \cite{jahn-1981-avmac,ahlswedecai-1999-obli-avmac-no-constr,pereg-steinberg-2019-mac}, 
adversarial fading channels \cite{hosseinigoki-kosut-2019-fading} etc. were also studied in the literature.
In each of these models, the adversaries may exhibit starkly contrasting behaviours. 
We do not intend to provide an exhaustive list of prior works.

\section{Overview of our results and techniques}
\label{sec:overview_results_techniques}
In this work, we provide the correct notion of list-symmetrizability $ L^* $ which we call $\cp$-symmetrizability, denoted by $ L_\cp^* $. 
We show  $ L^* = L_\cp^* $ and prove bounds on the $L$-list-decoding capacity for any $ L\in\bZ_{\ge1} $ using this new notion of list-symmetrizability.
(See \Cref{thm:cap_obli_avc_list_dec} for formal statements.)
Specifically, given any oblivious AVC with input \& state constraints, we show the following.
\begin{enumerate}
	\item When a given target list-size $ L $ is at most $ L_\cp^* $, then the channel is $L$-symmetrizable and the $L$-list-decoding capacity is zero.
	See \Cref{thm:symm_converse} for a formal statement.
	\item When $ L$ is strictly greater than $L_\cp^* $, we prove a natural lower bound on the $L$-list-decoding capacity using techniques that slightly extends those in \cite{csiszar-narayan-it1988-obliviousavc,hughes-1997-list-avc,sarwate-gastpar-2012-list-dec-avc-state-constr}. 
	In particular, the capacity is positive in this case. 
	See \Cref{thm:achievability} for a formal statement.
	\item When $L$ is strictly greater than $ L_\cp^* $, we did not manage to prove a matching upper bound on capacity. 
	En route to a tight characterization, we propose a conjecture conditioned on which we show that our lower bound is tight. 
	See \Cref{thm:converse_rate_ub} for a formal statement.
	The conjecture (\Cref{conj:comb_conj}) is concerned with basic structures of sets of vectors (over finite alphabets).
	It is of combinatorial nature and does not require backgrounds in AVCs. 
	We propose a natural subcode construction (\Cref{sec:subcode_extraction_towards_conj}) towards the resolution of this conjecture. 
\end{enumerate}

Curiously, our proof techniques crucially hinge on the recent development in the study of \emph{omniscient} AVCs. 
As alluded to in \Cref{sec:prior_obli}, incorporating constraints (especially state constraints) into Hughes's  \cite{hughes-1997-list-avc} definition of list-symmetrizability is a challenging task.
Sarwate and Gastpar \cite{sarwate-gastpar-2012-list-dec-avc-state-constr} made the first attempt by providing two candidates (strong and weak list-symmetrizability, denoted by $ L_\strong^* $ and $ L_\weak^* $, respectively) of extension of Hughes's notion to the constrained case. 
However, these extended notions are not tight in the sense of dichotomy, i.e., for any given list-size $ L\in\bZ_{\ge1} $, the $ L $-list-decoding capacity is zero if $ L\le L^* $ and is positive otherwise. 

We now explain how we close the gap between $ L_\strong^* $ and $ L_\weak^* $ using the notion of $ \cp $-symmetrizability $ L_\cp^* $ (where $\cp $ stands for \emph{completely positive}). 
We first equip James with an improved jamming strategy (called \emph{$\cp$-symmetrization}) which allows him to enforce a zero communication rate.
Fix a list-size $ L\in\bZ_{\ge2} $\footnote{When  $ L=1 $, the problem at hand collapses to the unique-decoding problem which was solved in \cite{csiszar-narayan-it1988-obliviousavc}.} and a code $ \cC\subseteq\cX^n $ of positive rate\footnote{Positive rate of a code $ \cC\subseteq\cX^n $ simply means the code size is exponentially large in $n$, i.e., $ \cardC = \cardX^{nR} $ for some constant $ R\in(0,1] $.} satisfying input constraints. 
Suppose Alice transmitted a  codeword $ \vbfx_{i_0}\sim\cC $ corresponding to a random message $ i_0 $. 
As suggested by the intuition mentioned in \Cref{sec:prior_obli}, a natural way to ``symmetrize''  the channel is to let James sample a ``spoofing'' list $ \cL\coloneqq (\vbfx_{i_1},\cdots,\vbfx_{i_L}) $ of $L$ codewords uniformly from $\cC$ such that the output distribution looks identical if the $ (L+1) $-list $ \cL' \coloneqq (\vbfx_{i_0},\vbfx_{i_1},\cdots,\vbfx_{i_L}) $ is permuted arbitrarily, i.e., $ P_{\vbfy|\vbfx_{i_0}, \vbfx_{i_1},\cdots,\vbfx_{i_L}} = P_{\vbfy|\vbfx_{\pi(i_0)},\vbfx_{\pi(i_1)},\cdots,\vbfx_{\pi(i_L)}} $ for all permutations $ \pi $ on $ (L+1) $ elements.
Specifically, if James adopts the following strategy, then we will argue that the $L$-list-decoder of Bob must make an error (i.e., decode to a list which does not contain the truly transmitted message) with nonvanishing probability.
If we provide James with  a (discrete memoryless) ``jamming channel'' $ U_{\bfs|\bfx_{1},\cdots,\bfx_{L}} $, then he can use it to generate a random jamming sequence $ \vbfs_\cL $ once the spoofing list $\cL$ is fed into it. 
Suppose we can find a $ U_{\bfs|\bfx_1,\cdots,\bfx_L} $ satisfying the following property.
The distribution of the channel output $ \vbfy $ obtained from $ (\vbfx_{i_0},\vbfs_\cL)\xrightarrow{W_{\bfy|\bfx,\bfs}^\tn}\vbfy $ (where $ \cL \xrightarrow{U_{\bfx_1,\cdots,\bfx_L}^\tn}\vbfs_\cL $) remains the same if $ \vbfy $ is obtained from $ (\vbfx_{\pi(i_0)},\vbfs_{\pi(\cL)})\xrightarrow{W_{\bfy|\bfx,\bfs}^\tn}\vbfy $
(where $ \pi(\cL)\xrightarrow{ U_{\bfs|\bfx_1,\cdots,\bfx_L}^\tn }\vbfs_{\pi(\cL)} $ and $ \pi(\cL)\coloneqq(\vbfx_{\pi(i_1)}, \cdots,\vbfx_{\pi(i_L)}) $) for any permutation $ \pi $ on $ (L+1) $ elements. 
For example, if $L=2$, this property guarantees that $ \vbfy_1,\cdots,\vbfy_6 $ (where $ 6 = (2+1)! $) are statistically identical, where $ \vbfy_1,\cdots,\vbfy_6 $ follow the Bayesian networks shown in \Cref{fig:list_symm_eg}, respectively. 
\begin{figure}[htbp]
	\centering
	\includegraphics[width=0.8\textwidth]{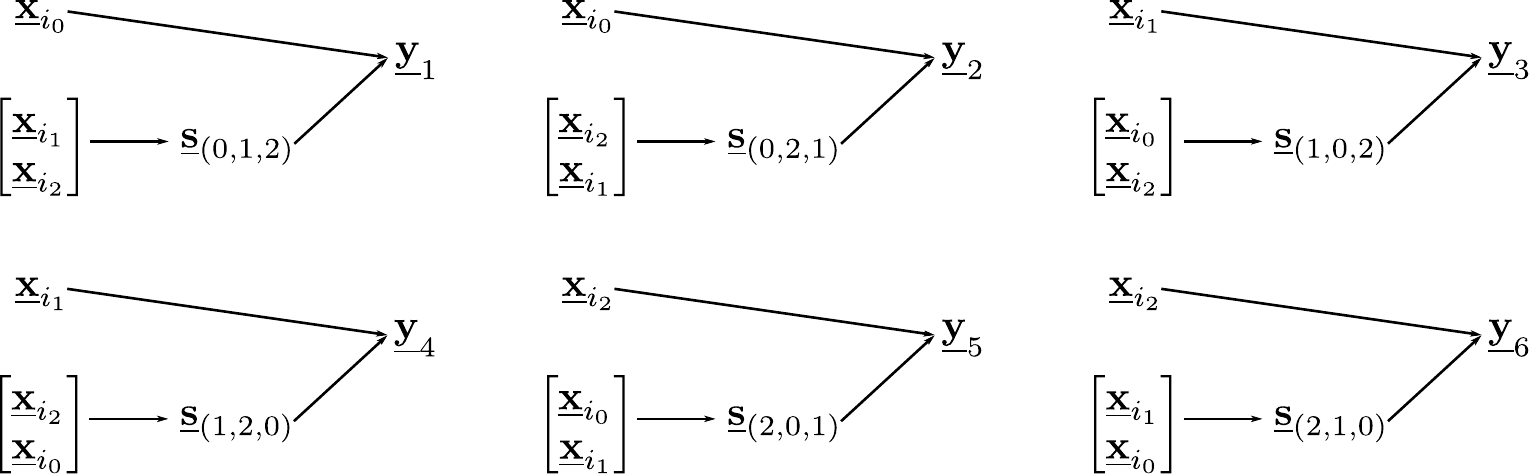}
	\caption{Example of list-symmetrizability.}
	\label{fig:list_symm_eg}
\end{figure}
As a consequence of such a property of $ U_{\bfs|\bfx_1,\cdots,\bfx_L} $, every $L$-sized sublist of $\cL'$ appears to be equally likely a posteriori. 
Bob's best list-decoder is essentially to output a random $ L $-sublist $ \wh\cL $ in $ \cL' $. 
The probability that Alice's transmitted message $ {i_0} $ falls outside $ \wh \cL $ is about $ \frac{1}{L+1} $.
For constant list-size $L$, the probability of error is bounded away from zero.
In the above example, among messages in the the candidate list $ \curbrkt{i_0,i_1,i_2} $, Bob randomly outputs two messages. 
The probability that he chooses $ \curbrkt{i_1,i_2} $ (which does not contain message $i_0$) instead of $ \curbrkt{i_0,i_1} $ or $ \curbrkt{i_0,i_2} $ (which are desirable output lists) is $1/3>0$. 
This  means that James has successfully jammed the communication and no positive rate of $L$-list-decodable codes with vanishing average probability of error can be achieved\footnote{The word ``achieve'' will be formally defined in \Cref{def:ach_rate_listdec_cap}. } under this type of jamming strategy.
It turns out the above heuristic strategy can be formalized to give a sufficient condition for zero list-decoding capacity in the absence of constraints \cite{hughes-1997-list-avc}.

For general  oblivious AVCs under constraints, one caveat in the above heuristics is that we also need to ensure that $ {\vbfs_{\pi(\cL)}} $ satisfies the state constraints for every permutation $ \pi $ on $(L+1)$ elements. 
Note that the distribution  of $ \vbfs_{\pi(\cL)} $ is essentially (with high probability tightly concentrated around) $ P_{\vbfx_{\pi(i_0)}, \vbfx_{\pi(i_1)},\cdots,\vbfx_{\pi(i_L)}} U_{\bfs|\bfx_1,\cdots,\bfx_L}^\tn $ marginalized to $ \vbfs $. 
A priori, it is  unclear whether this distribution remains (approximately) the same under different permutations. 
Thus it is challenging to bound the distribution of the jamming sequence. 
Therefore, in the presence of state constraints, we cannot yet claim that the above jamming strategy works. 
We bypass this obstacle by invoking a list-decoding version of generalized Plotkin bound recently developed by Zhang, Budkuley and Jaggi \cite{zbj-2019-generalized-ld} (which in turn built upon \cite{wbbj-2019-gen_plotkin}). 
This theorem (informally stated below) is concerned with basic structures of sets of vectors over finite alphabets (in particular, codes). 
\begin{theorem}[(A corollary of) generalized Plotkin bound for list-decoding, informal, \cite{zbj-2019-generalized-ld}]
\label{thm:gen_plotkin_listdec_informal}
Let $ \cX $ be a finite alphabet.
In any sufficiently large\footnote{The condition ``sufficiently large'' here means larger than some constant independent of $n$. In our applications, the positive rate of a code is more than enough to certify the ``sufficiently large'' criterion.}  set $ \cC = \curbrkt{\vx_i}_{i} $ of vectors in $ \cX^n $,
there is a completely positive distribution $ P_{\bfx_1,\cdots,\bfx_L} $ such that,
with constant (independent of $n$) probability, a uniformly random (ordered) list $ (\vbfx_{i_1},\cdots,\vbfx_{i_L}) $ ($ i_1<\cdots<i_L $) has empirical distribution\footnote{Given a list $ (\vx_1,\cdots,\vx_L) $ of vectors in $ \cX^n $ (where $ \cX $ is a finite set), the \emph{empirical distribution}, a.k.a. the \emph{type}/\emph{histogram}, of $ (\vx_1,\cdots,\vx_L) $ is the distribution $ P_{\bfx_1,\cdots,\bfx_L} $ defined as $ P_{\bfx_1,\cdots,\bfx_L}(x_1,\cdots,x_L) = \frac{1}{n}\card{\curbrkt{i\in[n]\colon \vx_1(i) = x_1,\cdots,\vx_L(i) = x_L}} $ for all $ (x_1,\cdots,x_L)\in\cX^L $. 
See \Cref{def:type}}  approximately equaling $ P_{\bfx_1,\cdots,\bfx_L} $. 
\end{theorem}

\begin{remark}
What was actually proved in \cite{zbj-2019-generalized-ld} is a Plotkin-type bound giving rise to a sufficient and necessary condition for list-decoding capacity positivity of general \emph{omniscient} AVCs.
It can be viewed as a characterization of the phase transition threshold of the sizes of multiple packings using general shapes in a finite product space. 
Combining it with the hypergraph \turan's theorem (\Cref{thm:hypergraph_turan}) allows us to prove \Cref{thm:gen_plotkin_listdec_informal}. 
The formal version of \Cref{thm:gen_plotkin_listdec_informal} is stated and proved in \Cref{lem:apply_turan}. 
\end{remark}

\emph{Completely positive} ($\cp$) distributions are joint distributions of tuples of random variables that can be written as a convex combination of tensor products of identical distributions, e.g., $ P_{\bfx_1,\cdots,\bfx_L} = \sum_{i} \lambda_i P_{\bfx_i}^\tl $ for some coefficients $ \curbrkt{\lambda_i}_i $ and distributions $ \curbrkt{P_{\bfx_i}}_i $. 
The above theorem allows us to bound the the distribution of $ \vbfs $ and to justify the validity of the previously described  jamming strategy.
This is because, $ \cp $ distributions are invariant under permutations by definition, i.e., $ P_{\bfx_1,\cdots,\bfx_L} = P_{\bfx_{\pi(1)},\cdots,\bfx_{\pi(L)}} $ for any $\pi$.
By \Cref{thm:gen_plotkin_listdec_informal} (with $L$ in the theorem being $L+1$), 
the event that $ \cL' = (\vbfx_0,\cL) $ has empirical distribution approximately $ P_{\bfx_0,\bfx_1,\cdots,\bfx_L} $ for some order-$(L+1)$ $\cp$ distribution $ P_{\bfx_0,\bfx_1,\cdots,\bfx_L} $ happens with constant probability.
Conditioned on this event, the distribution of $ \vbfs_{\pi(\cL)} $ for any permutation $ \pi $ is with high probability tightly concentrated around $ \sqrbrkt{P_{\bfx_0,\bfx_1,\cdots,\bfx_L}U_{\bfs|\bfx_1,\cdots,\bfx_L}}_\bfs^\tn $\footnote{Here the notation $ \sqrbrkt{\cdot}_\bfs $ refers to the marginal on the variable $\bfs$  of the joint distribution in the bracket.}.
This is  a product distribution independent of the particular realization of James's spoofing list $ \cL $ and the permutation $\pi$. 
The above argument hence basically  justifies the effectiveness of the translation from symbol-wise list-symmetrizability to vector-wise list-symmetrizability.
Indeed, even in the presence of state constraints,  such a translation can be operationally realized by the  previously described jamming strategy which we call $\cp$-symmetrization. 
(See \Cref{sec:cp_symm} for the precise description of $\cp$-symmetrization). 
The rest of the proof (which is deferred to \Cref{app:converse_rate_zero}) can be finished using relatively ``standard'' techniques (with some careful tweaks) along the lines of \cite{csiszar-narayan-it1988-obliviousavc,hughes-1997-list-avc,sarwate-gastpar-2012-list-dec-avc-state-constr}.

The above argument shows that if a channel is $ L $-$\cp$-symmetrizable (see \Cref{def:cp_symm} for the formal definition) then the $L$-list-decoding capacity is zero.
That is, $L$-$\cp$-symmetrizability is a sufficient condition of zero $L$-list-decoding capacity.
We then sketch a matching achievability argument showing that this condition is also necessary.
This turns out to be a relatively straightforward extension to the classical results by \cite{csiszar-narayan-it1988-obliviousavc,hughes-1997-list-avc,sarwate-gastpar-2012-list-dec-avc-state-constr}.
Specifically, by non-list-symmetrizability, James could not find a jamming channel $ U_{\bfs|\bfx_1,\cdots,\bfx_L} $ satisfying the aforementioned property. 
This means for some $\cp$ distributions, no matter which $ U_{\bfs|\bfx_1,\cdots,\bfx_L} $ James uses, the aforementioned jamming strategy fails since the jamming sequence $ \vbfs_\cL $ generated from a spoofing list $ \cL $ of a $\cp$ type $ P_{\bfx_1,\cdots,\bfx_L} $ violates some state constraints.
Alice could leverage that particular completely positive distribution $ P_{\bfx_1,\cdots,\bfx_L} $ to  construct a ``good'' code $\cC$. 
By complete positivity, assume $ P_{\bfx_1,\cdots,\bfx_L} $ can be decomposed as $ \sum_{i = 1}^k\lambda_iP_{\bfx_i}^\tl $ for some coefficients $ \curbrkt{\lambda_i}_{i = 1}^k $ and some distributions $ \curbrkt{P_{\bfx_i}}_{i = 1}^k $. 
Alice simply samples $ \cardX^{nR} $ codewords for some constant $ R\in(0,1] $ each of which is independently generated using the following distribution.
To sample a codeword $ \vbfx $, sample the first $ n\lambda_1 $ components independently from distribution $ P_{\bfx_1} $, sample the next $ n\lambda_2 $ components independently from distribution $ P_{\bfx_2} $, ..., sample the last $ n\lambda_k $ components independently from distribution $ P_{\bfx_k} $. 
By measure concentration, with high probability every size-$L$ list in $\cC$ has joint type approximately $ \sum_{i = 1}^k\lambda_iP_{\bfx_i}^\tl $. 
Following the classical techniques by \csiszar and Narayan \cite{csiszar-narayan-it1988-obliviousavc}, this further implies that with high probability such a random code is resilient to any feasible jamming strategy.
This claim may not be immediately clear to readers who are not familiar with the AVC literature, since we are claiming the possibility of reliable communication robust to \emph{any}  jamming strategy, not necessarily of the form of $\cp$-symmetrization introduced before.
However, thanks to \csiszar and Narayan \cite{csiszar-narayan-it1988-obliviousavc}, this claim does hold and the proof is nowadays standard.

\section{Organization of this paper}
The rest of the paper is organized as follows. 
The notational convention followed in this paper is introduced in \Cref{sec:notation}. 
Preliminaries on probability theory,  oblivious AVC model, channel coding, list-decodable codes and information measures are given in \Cref{sec:prelim}. 
Formal presentation starts from \Cref{sec:obli_avc_and_cp_symm} onwards. 
The core definition of $ \cp $-symmetrizability is introduced in \Cref{sec:obli_avc_and_cp_symm}. 
Given this, formal statements of our main theorems are stated and compared with those in \cite{sarwate-gastpar-2012-list-dec-avc-state-constr} in \Cref{sec:results}. 
Technical proofs start from \Cref{sec:cp_symm} onwards. 
We describe and (partly) analyze $ \cp $-symmetrization, the most conceptually novel and technically challenging part of this work, in \Cref{sec:cp_symm}. 
Part of the proof is delegated to \Cref{app:robust_plotkin,app:converse_rate_zero}.
The above three sections jointly prove \Cref{thm:symm_converse}.
We prove the capacity lower bound in \Cref{sec:achievability}  by designing and analyzing a coding scheme.
Part of the proof is deferred to \Cref{app:cw_select,app:unambiguity_pf}. 
The above three sections jointly prove \Cref{thm:achievability}.
We prove, conditioned on \Cref{conj:comb_conj}, a matching capacity upper bound in \Cref{sec:converse_rate_ub} with part of the proof left for \Cref{app:strong_conv_fading_pf}. 
The above two sections jointly prove \Cref{thm:converse_rate_ub}.

\section{Notation}
\label{sec:notation}
Sets are denoted by capital letters in calligraphic typeface, e.g., $ \cX,\cS,\cY $, etc. 
All alphabets in this paper are finite sized. 
For a positive integer $ M $, we use $ [M] $ to denote $ \curbrkt{1,\cdots, M} $. 
Let $\cX$ be a finite set. 
For an integer $ k\le\cardX $, we use $ \binom{\cX}{k} $ to denote $ \curbrkt{\cX'\subseteq\cX\colon \card{\cX'}= k} $. 
When we write $ \binom{[M]}{L} $, we think of an element $ \curbrkt{i_1,\cdots,i_L} $ in it as in ascending order, i.e., $ i_1<\cdots<i_L $. 
Similarly, we define $ \binom{\cX}{\le k}\coloneqq\curbrkt{ \cX'\subseteq\cX\colon \card{\cX'}\le k } $. 

Random variables are denoted by lowercase letters in boldface, e.g., $\bfx,\bfs,\bfy $, etc. 
Their realizations are denoted by corresponding lowercase letters in plain typeface, e.g., $x, s, y$, etc. Vectors (random or fixed) of length $n$, where $n$ is the blocklength of the code without further specification, are denoted by lowercase letters with  underlines, e.g., $\vbfx,\vbfs,\vbfy,\vx,\vs,\vy$, etc. 
The $i$-th entry of a vector $\vx\in\cX^n$ (resp. $\vbfx\in\cX^n$) is denoted by $\vx(i)$ (resp. $\vbfx(i) $). 
Let $ \cL = \curbrkt{i_1,\cdots,i_L} $ be a finite set of nonnegative integers such that $ i_1<\cdots<i_L $. 
We use $ x_\cL $ (resp. $ \bfx_\cL $) to denote $ (x_{i_1},\cdots,x_{i_L}) $ (resp. $ (\bfx_{i_1},\cdots,\bfx_{i_L}) $). 
Matrices are denoted by capital letters in boldface, e.g., $\bfA,\bfB$, etc. 
Similarly, the $(i,j)$-th entry of a matrix $\bfG\in\bF^{n\times m}$ is denoted by $\bfG(i,j)$. 
We sometimes write $\bfG_{n\times m}$ to explicitly specify its dimension. 
For square matrices, we write $\bfG_n$ for short. 
The letter $\bfI$ is reserved for the identity matrix.

We use the standard Bachmann--Landau (Big-Oh) notation. 
For $x\in\bR$, let $[x]^+\coloneqq\max\curbrkt{x,0}$.
For two real-valued functions $f(n),g(n)$ of positive integers, we say that $f(n)$ \emph{asymptotically equals} $g(n)$, denoted by $f(n)\asymp g(n)$, if 
$\lim_{n\to\infty}{f(n)}/{g(n)} = 1$.
We write $f(n)\doteq g(n)$ (read $f(n)$ \emph{dot equals} $g(n)$) if
$\lim_{n\to\infty}\paren{\log f(n)}/\paren{\log g(n)} = 1$.
Note that $f(n)\asymp g(n)$ implies $f(n)\doteq g(n)$, but the converse is not true.
Related notations such as dot less/larger than (or equal to), denoted by $ \dotl/\dotg $ (or $ \dotle/\dotge $) can be similarly defined.
For any $\cA\subseteq\cX$, the indicator function of $\cA$ is defined as, for any   $x\in\cX$, 
\[\one_{\cA}(x)\coloneqq\begin{cases}1,&x\in \cA\\0,&x\notin \cA\end{cases}.\]
At times, we will slightly abuse notation by saying that $ \indicator{\sfA} $ is $1$ when event $\sfA$ happens and $0$ otherwise. 
Note that $\one_{\cA}(\cdot)=\indicator{\cdot\in\cA}$.
In this paper, all logarithms are to the base 2. 

We use $ \Delta(\cX) $ to denote the probability simplex on $\cX$.
Related notations such as $ \Delta(\cX\times\cY) $ and $ \Delta(\cY|\cX) $ are similarly defined. 
For a distribution $ P_{\bfx,\bfy|\bfu}\in\Delta(\cX\times\cY|\cU) $, we use $ \sqrbrkt{P_{\bfx,\bfy|\bfu}}_{\bfx|\bfu}\in\Delta(\cX|\cU) $ to denote the marginal distribution onto $ \bfx $ given $ \bfu $, i.e., for every $ x\in\cX,u\in\cU $, $ \sqrbrkt{P_{\bfx,\bfy|\bfu}}_{\bfx|\bfu}(x|u) = \sum_{y\in\cY}P_{\bfx,\bfy|\bfu}(x,y|u) $. 
We use $ \Delta^{(n)}(\cX) $ to denote the set of types (i.e., empirical distributions/histograms, see \Cref{def:type} for formal definitions) of length-$n$ vectors over alphabet $ \cX $.
That is, $ \Delta^{(n)}(\cX) $ consists of all distributions $ P_\bfx\in\Delta(\cX) $ that can be induced by $ \cX^n $-valued vectors. 
Other notations such as $ \Delta^{(n)}(\cX\times\cY) $ and $ \Delta^{(n)}(\cY|\cX) $ are similarly defined. 
The notation $ \bfx\sim P_\bfx $ (resp. $ \vbfx\sim P_{\vbfx} $) means that the p.m.f. of a random variable (resp. vector) $ \bfx $ (resp. $ \vbfx $) is $ P_\bfx $ (resp. $ P_\vbfx $). 
If $ \bfx $ is uniformly distributed in $ \cX $, then we write $ \bfx\sim\cX $. 
The symmetric  group of degree $N\in\bZ_{\ge1} $ is denoted by $ S_N $. 
It consists of all permutations, typically denoted by lowercase Greek letters,  on a set of $N$ elements. 
When the set is a subset of nonnegative integers, we think of its elements as listed in ascending order. 
If $ \cN = \curbrkt{i_1,\cdots,i_N}\subset\bZ_{\ge1} $ such that $ i_1<\cdots<i_N $, then for any $ \pi\in S_N $, we define $ \pi(\cN) \coloneqq (\pi(i_1),\cdots,\pi(i_N)) $. 
Throughout this paper, we use $ d(\cdot,\cdot) $ to denote the $\ell^1 $ distance between two distributions. 
Specifically, for $ P,Q\in\Delta(\cX) $, $ d(P,Q)\coloneqq \sum_{x\in\cX}\abs{P(x) - Q(x)} $. 
For a subset $ \cA\subseteq\Delta(\cX) $, the distance between $ P $ and $\cA$ is defined as $ d(P,\cA)\coloneqq\min_{Q\in\cA}d(P,Q) $. 
The inner product between $ P $ and $ Q $ is defined as $ \inprod{P}{Q}\coloneqq\sum_{x\in\cX}P(x)Q(x) $. 
The $ \ell^p $-norm of a vector is denoted by $ \norm{p}{\cdot} $.


\section{Preliminary}
\label{sec:prelim}

\subsection{Probability}
\begin{lemma}[Markov]
\label{lem:markov}
If $X$ is a nonnegative random variable, then for any $a>0$, $ \prob{X\ge a}\le \expt{X}/a $. 
\end{lemma}

\begin{lemma}[Chebyshev]
\label{lem:chebyshev}
If $X$ is an integrable random variable with finite expectation and finite nonzero variance, then for any $a>0$, $ \prob{\abs{X-\expt{X}}\ge a}\le\var{X}/a^2 $.
\end{lemma}


\begin{lemma}[Sanov]
\label{lem:sanov}
Let $\cQ\subset\Delta\paren{\cX}$ be a subset of distributions such that it is equal to the closure of its interior. 
Let $\vbfx\sim P_\bfx^{\otimes n}$ for some $P_\bfx\in\Delta(\cX)$. 
Then
\begin{align}
\prob{\tau_{\vbfx}\in\cQ}\doteq 2^{ -n\inf_{Q_\bfx\in\cQ}\kl{Q_\bfx}{P_\bfx}  }, \notag
\end{align}
where the {Kullback--Leibler divergence} $ \kl{\cdot}{\cdot} $ between two distributions is defined in \Cref{def:info_mes}. 
\end{lemma}

\begin{lemma}[Pinsker]
\label{lem:pinsker}
Let $ P, Q\in\Delta(\cX) $. 
Then $ \kl{P}{Q} \ge\frac{1}{2\ln2}\normone{P-Q}^2 $. 
\end{lemma}

\subsection{Oblivious AVCs and list-decoding}
We first formally define the oblivious adversarial channel model concerned with in this paper. 
Oblivious AVCs are communication channels governed by an \emph{oblivious} adversary in the following sense.
The adversary, referred to as James, aims to prevent reliable transmission from Alice to Bob from happening by injecting carefully designed noise to the channel. 
Importantly, we assume that an oblivious adversary does \emph{not} get to see Alice's transmission (though  he \emph{does} know the codebook which is always assumed to be revealed to everyone before communication happens). 
\begin{definition}[Oblivious AVCs]
\label{def:obli_avc}
An oblivious AVC $ \obliavc $ consists of three alphabets $ \cX,\cS,\cY $ for the input, jamming and output sequences, respectively; 
input constraints $ \lambda_\bfx\subseteq\Delta(\cX) $ and state constraints $ \lambda_\bfs\subseteq\Delta(\cS) $; 
and an adversarial channel $ W_{\bfy|\bfx,\bfs} $ from Alice to Bob governed by James.
To avoid peculiar behaviours, we assume that both $ \lambda_\bfx $ and $ \lambda_\bfs $ are convex. 

Codewords of types from $ \lambda_\bfx $ are allowed to be input to the channel. 
Knowing the codebook $ \cC $, Alice's encoder $ \phi $ and Bob's decoder $\psi$ (both of which are formally defined in \Cref{def:enc_list_dec}),  James generates a jamming sequence $ \vbfs $ and sends it through the channel $ W_{\bfy|\bfx,\bfs} $. 
The channel generates $ \vbfy $  in a memoryless manner: $ \prob{\vbfy = \vy\condon\vbfx = \vx,\vbfs = \vs} = W_{\bfy|\bfx,\bfs}^\tn(\vy|\vx,\vs) = \prod_{i = 1}^nW_{\bfy|\bfx,\bfs}(\vy(i)|\vx(i),\vs(i)) $. 
The channel output $ \vbfy $ is then received by Bob. 
See \Cref{fig:diag_obli_listdec} for a block diagram of the oblivious AVC model under list-decoding (see \Cref{def:enc_list_dec} for the definition of list-decodability). 
\end{definition}
\begin{figure}[htbp]
	\centering
	\includegraphics[width=0.7\textwidth]{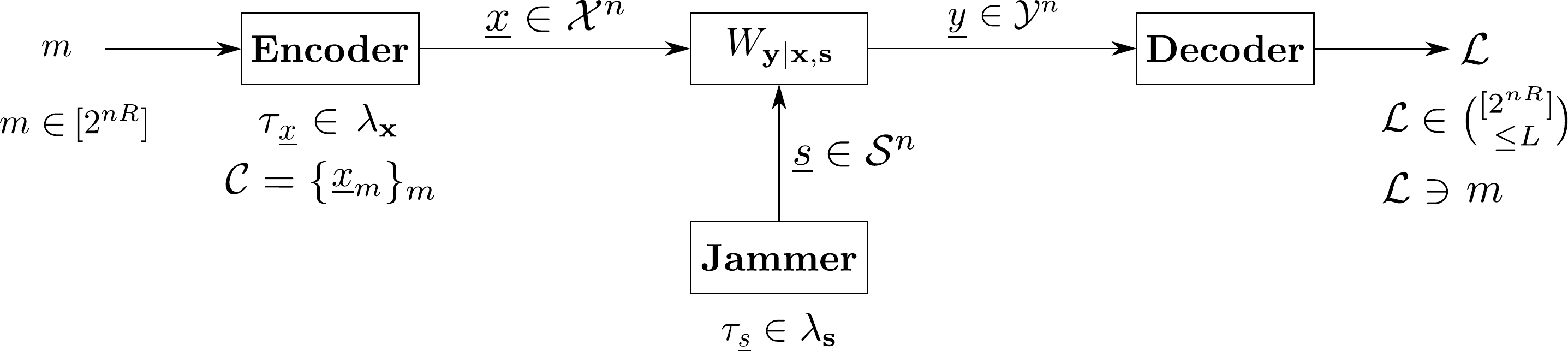}
	\caption{Block diagram of an oblivious AVC under list-decoding.}
	\label{fig:diag_obli_listdec}
\end{figure}
\begin{remark}
Though the channel law is a product distribution, the jamming sequence may not follow a product distribution. 
Indeed, they do not necessarily follow any fixed distribution. 
This makes robust communication against adversaries challenging. 
\end{remark}

\begin{remark}
We have already seen an example of oblivious AVC --- the oblivious bitflip channel --- in \Cref{sec:prior_obli}. 
\end{remark}

\begin{definition}[$L$-list-decodable codes for oblivious AVCs]\label{def:enc_list_dec}
An $L$-list-decodable code $\cC$ for an oblivious AVC $$ \obliavc $$ consists of an encoder $ \phi\colon [M]\to\cX^n $ mapping any $ m\in[M] $ to $ \phi(m) = \vx_m $ satisfying $ \tau_{\vx_m}\in\lambda_\bfx $; and an $L$-list-decoder $ \psi\colon\cY^n\to\binom{[M]}{\le L} $ such that $ \card{\psi(\vy)}\le L $ for all $ \vy\in\cY^n $.
We  call the image of $ \phi $ the \emph{codebook}, denoted by the same symbol $\cC$ (or simply a code, with slight abuse of terminology), i.e., $ \cC\coloneqq\curbrkt{\vx_i}_{i\in[M]} $. 
The length $n$ of each codeword is called the \emph{blocklength} of $\cC$. 
The \emph{rate} of $\cC$ is defined as $ R(\cC)\coloneqq\frac{\log M}{n\log\cardX} $. 
\end{definition}

\begin{remark}
In this paper, when we talk about ``a code'', we always mean a sequence of codes of increasing blocklengths, i.e.,  $ \curbrkt{\cC_{i}}_{i\ge1} $ each of blocklength $ n_i $, where $ n_1<n_2<\cdots\in\bZ_{\ge1} $ is an infinite sequence of increasing integers. 
\end{remark}

\begin{definition}[Average  probability of error]\label{def:avg_error_prob}
The \emph{average probability of $L$-list-decoding error} of a code $ \cC $ when used over an oblivious channel $\cA_\obli $ is defined as 
\begin{align}
P_{\e,\avg}(\cC)\coloneqq& \max_{\vs\in\cS^n} \frac{1}{M} \sum_{i\in[M]} \prob{ \psi(\vbfy)\not\ni i\condon \bfi = i,\vbfs = \vs } = \max_{\vs\in\cS^n}\frac{1}{M}\sum_{i\in[M]} \sum_{\vy\in\cY^n\colon \psi(\vy)\not\ni i} W_{\bfy|\bfx,\bfs}^\tn(\vy|\vx_i,\vs) . \notag 
\end{align}
\end{definition}

\begin{definition}[Achievable rates and $L$-list-decoding capacity]
\label{def:ach_rate_listdec_cap}
A rate $R$ is said to be  \emph{achievable} for an oblivious AVC $ \cA_\obli $ under $L$-list-decoding if there is an infinite sequence of $L$-list-decodable codes $ \curbrkt{\cC_i}_{i\ge1} $ for $ \cA_\obli $ of increasing blocklengths such that $ R(\cC_i)\ge R $ for all $i$ and $ P_{\e,\avg}(\cC_i)\xrightarrow{i\to\infty}0 $. 
The supremum of all achievable rates is called the \emph{$L$-list-decoding capacity} of $\cA_\obli $, denoted by $ C_L(\cA_\obli) $.
\end{definition}

\begin{definition}[List-symmetrizability]
\label{def:list_symm}
For any oblivious AVC $ \obliavc $, define the \emph{list-symmetrizability} of $ \cA_\obli $, denoted by $ L^* $, as the minimum $L$ such that $ C_L(\cA_\obli)>0 $. 
\end{definition}


\subsection{Approximate constant-composition codes}
\begin{definition}[Approximate constant-composition codes]
\label{def:apx_cc_code}
A code $ \cC\subseteq\cX^n $ is said to be $ (\lambda, P_\bfx) $-constant-composition for some $  P_\bfx\in\Delta(\cX) $ if every codeword $ \vx\in\cC $ has type $ \tau_\vx $ satisfying $ d\paren{\tau_{\vx}, P_\bfx}\le\lambda $.
\end{definition}

\begin{definition}[Approximate constant-composition codes with time-sharing]
\label{def:apx_cc_code_time_sharing}
A code $ \cC $ is said to be $ (\lambda,\vu, P_{\bfx|\bfu}) $-constant-composition for  some $ \vu\in\cU^n $ and some $  P_{\bfx|\bfu}\in\Delta(\cX|\cU) $  if every codeword $ \vx\in\cC $ satisfies $ d\paren{\tau_{\vu,\vx} , \tau_{\vu}P_{\bfx|\bfu}}\le\lambda $. 
\end{definition}

\begin{definition}[Quantization/net]\label{def:quant}
Given a metric space $ (\cX,\dist) $ and a constant $ \eta>0 $, an \emph{$ \eta $-net} or an \emph{$\eta$-quantization} $\cN$ of $ \cX $ w.r.t. the metric $\dist$ is a subset $ \cN\subset\cX $ satisfying: for every $ x\in\cX $, there is an $ x'\in\cN $ such that $ \dist(x,x')\le\eta $. 
\end{definition}

Taking a simple coordinate-wise quantization allows us to get an upper bound on the size of a net of a probability simplex. 
A proof of the following lemma can be found in \cite{zbj-2019-generalized-ld}. 
\begin{lemma}[Bounds on the size of nets]\label{lem:quant}
Let $ \cX $ be a finite set.
For any constant $ \eta>0 $, there is an $ \eta $-net of $ (\Delta(\cX), d) $ of size at most $ \ceil{ \frac{\cardX}{2\eta} }^{\cardX} \le\paren{ \frac{\cardX}{2\eta}+1 }^\cardX $. 
\end{lemma}

A straightforward quantization argument allows us to reduce a general code to an approximate constant-composition code with only a constant multiplicative factor loss in the code size. 
\begin{lemma}[Approximate constant-composition reduction]\label{lem:apx_cc_reduction}
For any $ \lambda\in(0,1) $,  any code $ \cC\subseteq\cX^n $ contains a $ (\lambda,P_\bfx) $-constant-composition subcode $ \cC'\subseteq\cC $ for some $ P_\bfx\in\Delta(\cX) $ of size at least $ \cardC/N $, where $ N \le \paren{\frac{\cardX}{2\lambda} + 1}^{\cardX} $.
In particular, $ R(\cC)\stackrel{n\to\infty}{\asymp}R(\cC') $. 
\end{lemma}

\subsection{Information measures and method of types}
\begin{definition}[Information measures]
\label{def:info_mes}
Let $ \cX,\cY $ be two finite sets and $ P_\bfx\in\Delta(\cX) $. 
The \emph{Shannon entropy} of $ P_\bfx $ is defined as $ H(P_\bfx)\coloneqq\sum_{x\in\cX}P_\bfx(x)\log\frac{1}{P_\bfx(x)} $. 
It is alternatively written as $ H(\bfx) $ where $\bfx$ is an $ \cX $-valued random variable whose p.m.f. is $ P_\bfx $. 

Let $ \bfy $ be a $\cY$-valued random variable with joint p.m.f. $ P_{\bfx,\bfy} $ with $\bfx$. 
The \emph{conditional entropy} of $\bfx$ given $\bfy$ is defined as $ H(\bfx|\bfy) \coloneqq\sum_{(x,y)\in\cX\times\cY}P_{\bfx,\bfy}(x,y)\log\frac{1}{P_{\bfx|\bfy}(x|y)} $. 
The \emph{mutual information} between $ \bfx $ and $ \bfy $ is defined as $ I(\bfx;\bfy) \coloneqq\sum_{(x,y)\in\cX\times\cY}P_{\bfx,\bfy}(x,y)\log\frac{P_{\bfx,\bfy}(x,y)}{P_\bfx(x)P_\bfy(y)} $. 

If $\bfx$ and $\bfy$ together with some $ \cZ $-valued random variable $\bfz$ have joint distribution $ P_{\bfx,\bfy,\bfz} $, then the \emph{conditional mutual information} between $ \bfx $ and $ \bfy $ given $ \bfz $ is defined as $$ I(\bfx;\bfy|\bfz)\coloneqq\sum_{z\in\cZ}P_\bfz(z)\sum_{(x,y)\in\cX\times\cY} P_{\bfx,\bfy|\bfz}(x,y|z)\log\frac{P_{\bfx,\bfy|\bfz}(x,y|z)}{P_{\bfx|\bfz}(x|z)P_{\bfy|\bfz}(y|z)} .$$ 

If $ P,Q\in\Delta(\cX) $ and $P$ is absolutely continuous w.r.t. $Q$ (i.e., $ \supp(P)\subseteq\supp(Q) $), then the \emph{Kullback--Leibler (KL) divergence} between $P$ and $Q$ is defined as $ \kl{P}{Q}\coloneqq\sum_{x\in\cX}P(x)\log\frac{P(x)}{Q(x)} $. 
\end{definition}

Readers who are not familiar with the basics of information measures are encouraged to refer to, e.g., \cite{cover-thomas}. 
We list below several basic and well-known properties of information measures that will be frequently used throughout the paper.
\begin{lemma}[Properties of information measures]
\label{lem:properties_info_mes}
The information measures defined in \Cref{def:info_mes} satisfy the following properties. 
\begin{enumerate}
	\item Entropy, mutual information, KL divergence and their conditional versions are all nonnegative. 
	\item Conditioning reduces entropy: $ H(\bfx|\bfy)\le H(\bfx) $. 
	\item Alternative definitions: $ I(\bfx;\bfy|\bfz) = H(\bfx|\bfz) - H(\bfx|\bfz,\bfy) $; $ I(\bfx;\bfy) = \kl{P_{\bfx,\bfy}}{P_\bfx P_\bfy} $. 
	\item Chain rule: $ I(\bfx;\bfy_1,\bfy_2|\bfz) = I(\bfx;\bfy_1|\bfz) + I(\bfx;\bfy_2|\bfz,\bfy_1) $. 
\end{enumerate}
\end{lemma}

\begin{definition}[Types, joint types and conditional types]
\label{def:type}
Let $ \cX $ be a finite set and $ n\in\bZ_{\ge1} $. 
The \emph{type} of a vector $ \vx\in\cX^n $, denoted by $ \tau_\vx\in\Delta(\cX) $, is the empirical distribution/histogram of $\vx$ defined as: for every $ x\in\cX $, $ \tau_\vx(x) = \frac{1}{n}\card{\curbrkt{i\in[n]\colon \vx(i) = x}} $. 
The set of all types of $ \cX^n $-valued vectors is denoted by $ \Delta^{(n)}(\cX) $. 
Let $ \cY $ be another finite set and $ \vy\in\cY^n $. 
The \emph{joint type} $ \tau_{\vx,\vy} $ (and $ \Delta^{(n)}(\cX\times\cY) $ correspondingly) and the \emph{conditional type} $ \tau_{\vx|\vy} $ (and $ \Delta^{(n)}(\cX|\cY) $ correspondingly) are defined in a similar manner. 
Furthermore, these definitions can be extended to tuples of vectors in the canonical way.
\end{definition}

\begin{lemma}[Size of typical sets]\label{lem:aep}
Let $ \vx\in\cX^n $ and let $ P_{\bfy,\bfx}\in\Delta(\cY\times\cX) $ be such that $ \tau_\vx = \sqrbrkt{P_{\bfy,\bfx}}_\bfx $. 
Define the {$\eps$-conditionally typical set} $ \cA_{\vbfy|\vx}^\eps(P_{\bfy,\bfx})  $ of $ \cY^n $-valued sequences given $ \vx $ w.r.t. $ P_{\bfy,\bfx} $ as 
\begin{align}
\cA_{\vbfy|\vx}^\eps(P_{\bfy,\bfx}) \coloneqq\curbrkt{\vy\in\cY^n \colon \forall(y,x)\in\cY\times\cX,\; \frac{\tau_{\vy,\vx}(y,x)}{P_{\bfy,\bfx}(y,x)} \in[1-\eps,1+\eps] }. \notag 
\end{align}
Then  $ 2^{n\paren{H(\bfy|\bfx) - f(\eps)}}\le \card{\cA_{\vbfy|\vx}^\eps(P_{\bfy,\bfx})} \le 2^{n\paren{H(\bfy|\bfx) + f(\eps)}} $
for some constant $ f(\eps)>0 $ such that $ f(\eps)\xrightarrow{\eps\to0}0 $. 
\end{lemma}

\begin{lemma}\label{lem:prob_typ}
Fix a channel $ W_{\bfy|\bfx}\in\Delta(\cY|\cX) $. 
Let $ P_{\bfx,\bfs,\bfy}\in\Delta(\cX\times\cS\times\cY) $ and $ P_{\bfx,\bfs}\coloneqq\sqrbrkt{P_{\bfx,\bfs,\bfy}}_{\bfx,\bfs} $. 
Then for any $ (\vx,\vs)\in\cX^n\times\cS^n $ such that $ \tau_{\vx,\vs} = P_{\bfx,\bfs} $, 
\begin{align}
\sum_{\vy\in\cY^n\colon \tau_{\vx,\vs,\vy} = P_{\bfx,\bfs,\bfy}}W_{\bfy|\bfx}(\vy|\vx) \le 2^{-n\kl{P_{\bfx,\bfs,\bfy}}{{P_{\bfx,\bfs}W_{\bfy|\bfx}}}}. \notag 
\end{align}
\end{lemma}

\section{$\cp$-symmetrizability: a refined notion of symmetrizability}
\label{sec:obli_avc_and_cp_symm}
In this section, we introduce the notion of \emph{$\cp$-symmetrizability}  which is one of the core definitions of this work. 

Consider an oblivious AVC $ \obliavc $. 
Without loss of generality, we assume that $ \lambda_\bfx\subseteq\Delta(\cX) $ and $ \lambda_\bfx\subseteq\Delta(\cS) $ are convex polytopes defined by a finite number of linear inequalities and/or equalities, respectively.
Specifically, the \emph{input constraints} are given by
\begin{align}
\lambda_\bfx \coloneqq& \curbrkt{ P_\bfx\in\Delta(\cX) \colon \forall j\le\alpha,\;\sum_{x\in\cX} A_j(x)P_\bfx(x)\le \Gamma_j } = \curbrkt{ P_\bfx\in\Delta(\cX)\colon \bfA P_\bfx\le\vGamma } , \notag
\end{align}
where $\bfA\in\bR^{\alpha\times|\cX|} $ is a matrix whose $(j,x)$-th entry is $ A_{j}(x) $ and $\vGamma\in\bR^{\alpha} $ is a vector whose $j$-th component is $\Gamma_j $.
Similarly, the \emph{state constraints} are
\begin{align}
\lambda_\bfs \coloneqq& \curbrkt{ P_\bfs\in\Delta(\cS)\colon \forall j\le\beta,\;\sum_{s\in\cS} B_j(s)P_\bfs(s)\le \Lambda_j } = \curbrkt{ P_\bfs\in\Delta(\cS)\colon \bfB P_\bfs\le\vLambda } ,\notag
\end{align}
where $\bfB\in\bR^{\beta\times|\cS|} $ is a matrix whose $(j,s)$-th entry is $ B_{j}(s) $  and $\vLambda\in\bR^{\beta} $ is a vector whose $j$-th component is $\Lambda_j $.
Define the corresponding sets of length-$n$ $ \cX $-/$\cS$-sequences satisfying constraints $ \lambda_\bfx $/$\lambda_\bfs $ as
\begin{align}
\Lambda_\bfx\coloneqq&\curbrkt{ \vx\in\cX^n\colon \tau_\vx \in \lambda_\bfx }, \quad \Lambda_\bfs \coloneqq\curbrkt{\vs\in\cS^n\colon \tau_\vs\in\lambda_\bfs}. \notag 
\end{align}

To define $\cp$-symmetrizability, we need several preliminary definitions.

\begin{definition}[Obliviously $L$-symmetrizing distributions]
\label{def:def_symm_distr} 
Fix $ L\in\bZ_{\ge1} $.
Define the set of \emph{obliviously symmetrizing distributions} as 
\begin{align}
\cU_{\obli,\symml{L}} \coloneqq& \curbrkt{ U_{\bfs|\bfu,\bfx_{[L]}}\in\Delta(\cS|\cU\times\cX^L) \colon 
\begin{array}{rl}
&\forall u\in\cU, x_0\in\cX, x_{[L]}\in\cX^L, y\in\cY,\pi\in S_{ L+1 } , \\
&\displaystyle
\sum_{s\in\cS} W_{\bfy|\bfx,\bfs}(y|x_0,s)U_{\bfs|\bfu,\bfx_{[L]}}(s|u,x_{[L]}) \\
=&\displaystyle \sum_{s\in\cS} W_{\bfy|\bfx,\bfs}(y|x_{\pi(0)},s)U_{\bfs|\bfu,\bfx_{[L]}}(s|u,x_{\pi({[L]})})
\end{array}
}. \notag 
\end{align}
\end{definition}

\begin{remark}
For technical reasons that will be clear momentarily, the symmetrizing distributions we need take the slightly more complicated form of $ U_{\bfs|\bfu,\bfx_1,\cdots,\bfx_L} $  than $ U_{\bfs|\bfx_1,\cdots,\bfx_L} $ mentioned in the overview section (\Cref{sec:overview_results_techniques}). 
One should think of $ \bfu $ as a time-sharing variable used by Alice as part of the code design.
We assume that the coding scheme, in particular $\bfu$, is also known to James.
This is the reason why we allow the symmetrizing distribution to be conditioned on $\bfu$. 
\end{remark}

\begin{definition}[Self-couplings]
\label{def:def_self_coupling} 
Let $ P_\bfx\in\Delta(\cX) $ and $ L\in\bZ_{\ge2} $.  
Define the set of order-$L$ \emph{$ P_\bfx $-self-couplings} over $\cX$ as 
\begin{align}
\cJ^\tl(P_\bfx)\coloneqq& \curbrkt{ P_{\bfx_1,\cdots,\bfx_L}\in\Delta(\cX^L)\colon \forall i\in[L],\;\sqrbrkt{P_{\bfx_1,\cdots,\bfx_L}}_{\bfx_i} = P_\bfx }. \notag 
\end{align}
\end{definition}

\begin{definition}[Complete positivity ($\cp$)]\label{def:cp}
Let $ P_\bfx\in\Delta(\cX) $ and  $ L\in\bZ_{\ge2} $.
A distribution $ P_{\bfx_1,\cdots,\bfx_L}\in\cJ^\tl(P_\bfx) $ is called \emph{$ P_\bfx $-completely positive} ($\cp$) if it can be written as a convex combination of tensor products of identical distributions.
That is, there exist $ k\in\bZ_{>0} $, $ \lambda_1,\cdots,\lambda_k\in[0,1] $ satisfying $ \sum_{i=1}^k\lambda_i = 1 $, and $ P_{\bfx_1},\cdots,P_{\bfx_k}\in\Delta(\cX) $ such that
\begin{align}
P_{\bfx_1,\cdots,\bfx_L} =& \sum_{i = 1}^k\lambda_iP_{\bfx_i}^\tl. \label{eqn:cp_def_conv_comb}
\end{align}
Or equivalently, there exist $ k\in\bZ_{>0} $, a time-sharing variable $ \bfu\in\Delta(\cU)  $ (where $ \cU\coloneqq [k] $) and a conditional distribution $ P_{\bfx|\bfu}\in\Delta(\cX|\cU) $ such that
\begin{align}
P_{\bfx_1,\cdots,\bfx_L} = & \sqrbrkt{P_\bfu P_{\bfx|\bfu}^\tl}_{\bfx_1,\cdots,\bfx_L}. \label{eqn:cp_def_time_shar}
\end{align}
The set of all order-$L$ $ P_\bfx $-completely positive distributions  is denoted by $ \cp^\tl(P_\bfx) $. 
For any $ P_{\bfx_1,\cdots,\bfx_L}\in\cp^\tl(P_\bfx) $, we call a decomposition of the form \Cref{eqn:cp_def_conv_comb} or \Cref{eqn:cp_def_time_shar} a \emph{$\cp$-decomposition}. 
\end{definition}
\begin{remark}
Since a $ P_\bfx $-$\cp$-distribution should be a $ P_\bfx $-self-coupling in the first place, any $ (P_\bfu, P_{\bfx|\bfu}) $ given by a $\cp$-decomposition satisfies $ \sqrbrkt{P_\bfu P_{\bfx|\bfu}}_\bfx = P_\bfx $. 
\end{remark}
\begin{remark}
\label{rk:cp-decomp-non-unique}
The $\cp$-decomposition of a $\cp$-distribution is not necessarily unique. 
In particular, the number $k$ of components in a $\cp$-decomposition (\Cref{eqn:cp_def_conv_comb}) of $ P_{\bfx_1,\cdots,\bfx_L}\in\cp^\tl(P_\bfx) $ is  not necessarily unique.
Among all $\cp$-decompositions, the smallest $k =|\cU| $ given by a $ \cp $-decomposition is called the \emph{$\cp$-rank}, denoted by $ \cprk $, of $ P_{\bfx_1,\cdots,\bfx_k}\in\cp^\tl(P_\bfx) $.
In fact, even if  a $\cp$-decomposition is required to contain the same number of terms as  $\cp$-rank, it still may not be  unique. 
\end{remark}

\begin{definition}[$\cp$-symmetrizability]
\label{def:cp_symm}
Fix an oblivious AVC $ \obliavc $. 
For $ L\in\bZ_{\ge1} $, an input distribution $ P_\bfx\in\lambda_\bfx $ is said to be \emph{obliviously $\cp$-$L$-symmetrizable} (or $L$-symmetrizable for short) if for every $ P_{\bfx_1,\cdots,\bfx_L}\in\cp^\tl(P_\bfx) $ and every its $\cp$-decomposition $ P_{\bfx_1,\cdots,\bfx_L} = \sqrbrkt{P_\bfu P_{\bfx|\bfu}}^\tl $ where $ P_\bfu\in\Delta(\cU),P_{\bfx|\bfu}\in\Delta(\cX|\cU) $, there is a symmetrizing distribution $ U_{\bfs|\bfu,\bfx_1,\cdots,\bfx_L}\in\cU_{\obli,\symml{L}} $ such that 
\begin{align}
P_\bfs = \sqrbrkt{ P_\bfu P_{\bfx|\bfu}^\tl U_{\bfs|\bfu,\bfx_1,\cdots,\bfx_L} }_\bfs \in \lambda_\bfs,\label{eqn:l_symm_cond1} 
\end{align}
or equivalently, for each $ i\in[\beta] $, the ``jamming cost'' does not exceed the $i$-th constraint, i.e., 
\begin{align}
\cost_i(P_{\bfu,\bfx_{[L]}}, U_{\bfs|\bfu,\bfx_{[L]}}) = \cost_i((P_\bfu,P_{\bfx|\bfu}),U_{\bfs|\bfu,\bfx_{[L]}}) \coloneqq& \sum_{(u,x_{[L]},s)\in\cU\times\cX^L\times\cS}P_{\bfu,\bfx_{[L]}}(u,x_{[L]})U_{\bfs|\bfu,\bfx_{[L]}}(s|u,x_{[L]}) B_i(s) \le \Lambda_i. \label{eqn:l_symm_cond2} 
\end{align}
For an input distribution $ P_\bfx\in\lambda_\bfx $, define the \emph{oblivious $ P_\bfx $-$\cp$-symmetrizability} (or $ P_\bfx $-symmetrizability for short) of $ \cA_\obli $ as $$ L_\cp^*(P_\bfx) \coloneqq \max\curbrkt{ L\in\bZ_{\ge1}\colon P_\bfx\text{ is obliviously $\cp$-$L$-symmetrizable} }. $$
Define the \emph{oblivious $\cp$-symmetrizability} of $ \cA_\obli $ as $ L_\cp^* \coloneqq \min_{P_\bfx\in\lambda_\bfx}L_\cp^*(P_\bfx) $. 
\end{definition}

\begin{remark}
With slight abuse of terminology, we interchangeably call $ (P_\bfu, P_{\bfx|\bfu}) $ a $\cp$-decomposition.
We at times call $ P_{\bfx_1,\cdots,\bfx_L}\in\cp^\tl(P_\bfx) $ \emph{$L$-symmetrizable} if for every $\cp$-decomposition $ (P_\bfu,P_{\bfx|\bfu}) $,  \Cref{eqn:l_symm_cond1} or \Cref{eqn:l_symm_cond2} holds for some $ U_{\bfs|\bfu,\bfx_1,\cdots,\bfx_L}\in\cU_{\obli,\symml{L}} $. 
Also, we may call $ (P_\bfu,P_{\bfx|\bfu})\in\Delta(\cU)\times\Delta(\cX|\cU) $ \emph{$L$-symmetrizable} if \Cref{eqn:l_symm_cond1} or \Cref{eqn:l_symm_cond2} holds for some $ U_{\bfs|\bfu,\bfx_1,\cdots,\bfx_L}\in\cU_{\obli,\symml{L}} $.
\end{remark}

\begin{remark}
In the absence of constraints, \Cref{def:cp_symm} collapses to Hughes's \cite{hughes-1997-list-avc} notion of list-symmetrizability. 
When $ L = 1 $, \Cref{def:cp_symm} collapses to \csiszar--Narayan's \cite{csiszar-narayan-it1988-obliviousavc} notion of symmetrizability.
See \Cref{sec:myop_obli_symm} for more details on this reduction. 
However, \Cref{def:cp_symm} does not specialize to either the strong or weak list-symmetrizability (denoted by $ L_\strong^* $ and $ L_\weak^* $, respectively) due to Sarwate and Gastpar \cite{sarwate-gastpar-2012-list-dec-avc-state-constr}. 
See below (\Cref{def:strong_symm} and \Cref{def:weak_symm}) for the definition of $ L_\strong^*,L_\weak^* $ and see \Cref{sec:comparison_our_sg} for a proper comparison of $ L_\cp^* $, $ L_\strong^* $ and $ L_\weak^* $. 
In fact, one can find examples of oblivious AVCs for which the values of $ L_\strong^* $, $ L_\cp^* $ and $ L_\weak^* $ are \emph{strictly} different. 
We will do this by developing a machinery that we call \emph{canonical constructions} of oblivious channels. 
Taking proper canonical constructions allows us to show that there are channels for which $ L_\strong^*<L_\cp^* $ and there are channels for which $ L_\cp^*<L_\weak^* $. 
These constructions and their analysis will be presented in \Cref{sec:canonical_constr}. 
\end{remark}

In \cite{sarwate-gastpar-2012-list-dec-avc-state-constr}, the authors defined two notions  of list-symmetrizability, known as the \emph{strong} and \emph{weak} list-symmetrizability.
They were used for giving outer and inner bounds, respectively, on the $L$-list-decoding capacity of oblivious AVCs. 
Their definitions read as follows.
Weak symmetrizability of an input distribution $ P_\bfx $ replaces the quantifier ``$ \forall P_{\bfx_1,\cdots,\bfx_L}\in\cp^\tl(P_\bfx) $'' by a fixed distribution $ P_\bfx^\tl $. Strong symmetrizability instead replaces that by all $ P_\bfx $-self-couplings, i.e., the quantifier becomes ``$ \forall P_{\bfx_1,\cdots,\bfx_L}\in\cJ^\tl(P_\bfx) $''. Apparently, for any $ L\in\bZ_{\ge1} $, if $ P_\bfx $ is strongly $L$-symmetrizable, then it is $\cp$-$L$-symmetrizable; if $ P_\bfx $ is $\cp$-$L$-symmetrizable, then it is weakly $L$-symmetrizable. The notions of $ L_{\strong}^* $ and $ L_{\weak}^* $ are defined in the same way. We have the obvious relation $ L_{\strong}^*\le L_{\cp}^*\le L_{\weak}^* $.

We state below the formal definitions of strong and weak list-symmetrizability for completeness. 
\begin{definition}[Strong symmetrizability]
\label{def:strong_symm}
Fix an oblivious AVC $ \obliavc $.
For $ L\in\bZ_{\ge1} $,
an input distribution $ P_\bfx\in\lambda_\bfx $ is said to be \emph{obliviously strongly $L$-symmetrizable} if for every $ P_{\bfx_{[L]}}\in\cJ^\tl(P_\bfx) $, there is a symmetrizing distribution $ U_{\bfs|\bfx_{[L]}}\in\cU_{\obli,\symml{L}} $\footnote{Here and later in \Cref{def:weak_symm}, since there is no $\cp$-distributions, the time-sharing variable $\bfu$ is absent.
Hence the definition of $ \cU_{\obli,\symml{L}} $ should be slightly changed to be the set of distributions $ U_{\bfs|\bfx_{[L]}}\in\Delta(\cS|\cX^L) $ satisfying a certain system of identities.} such that $ \sqrbrkt{ P_{\bfx_{[L]}}U_{\bfs|\bfx_{[L]}} }_\bfs\in\lambda_\bfs $. 
The \emph{oblivious strong symmetrizability} of $ \cA_\obli $ is defined as $ L_\strong^* \coloneqq \min_{P_\bfx\in\lambda_\bfx}\max\curbrkt{L\in\bZ_{\ge1}\colon P_\bfx\text{ is obliviously strongly $L$-symmetrizable}} $. 
\end{definition}

\begin{definition}[Weak symmetrizability]
\label{def:weak_symm}
Fix an oblivious AVC $ \obliavc $.
For $ L\in\bZ_{\ge1} $,
an input distribution $ P_\bfx\in\lambda_\bfx $ is said to be \emph{obliviously weakly $L$-symmetrizable} if  there is a symmetrizing distribution $ U_{\bfs|\bfx_{[L]}}\in\cU_{\obli,\symml{L}} $ such that $ \sqrbrkt{ P_\bfx^\tl U_{\bfs|\bfx_{[L]}} }_\bfs\in\lambda_\bfs $. 
The \emph{oblivious weak symmetrizability} of $ \cA_\obli $ is defined as $$ L_\weak^* \coloneqq \min_{P_\bfx\in\lambda_\bfx}\max\curbrkt{L\in\bZ_{\ge1}\colon P_\bfx\text{ is obliviously weakly $L$-symmetrizable}} .$$
\end{definition}


\section{Our results}\label{sec:results}
In this section, we give formal statements of our results and compare them in details with the  closely related work by Sarwate and Gastpar \cite{sarwate-gastpar-2012-list-dec-avc-state-constr}. 

\begin{definition}[Capacity expression]\label{def:cap_expr}
Let $ \cA_\obli = (\cX,\cS,\cY, \lambda_\bfx,\lambda_\bfs, W_{\bfy|\bfx,\bfs}) $ be an oblivious AVC. 
Let $ L\ge1 $ be a certain list-size. 
Define 
\begin{align}
C_L =& \max_{{P_\bfx\in\lambda_\bfx\colon P_\bfx\text{ non-$L$-symmetrizable}}} \min_{U_{\bfs|\bfu}\in\Delta(\cS|\cU)\colon \sqrbrkt{P_\bfu U_{\bfs|\bfu}}_\bfs\in\lambda_\bfs} I(\bfx;\bfy|\bfu). \label{eqn:cap_obli_avc_list_dec}
\end{align}
Here the maximization is taken over all feasible input distributions $ P_\bfx $ that are non-$L$-symmetrizable.
By \Cref{def:cp_symm}, each such $ P_\bfx $ induces at least one non-$L$-symmetrizable $\cp$-distribution $ P_{\bfx_1,\cdots,\bfx_L} $.
By \Cref{def:cp}, each non-$L$-symmetrizable $\cp$-distribution $ P_{\bfx_1,\cdots,\bfx_L} $ further induces a non-$L$-symmetrizable pair $ (P_\bfx , P_{\bfx|\bfu}) $ through \Cref{eqn:cp_def_time_shar}. 
Therefore for each $ P_\bfx\in\lambda_\bfx $, the maximization is also implicitly taken over non-$L$-symmetrizable $\cp$-distributions $ P_{\bfx_1,\cdots,\bfx_L} $  and a pair $ (P_\bfx,P_{\bfx|\bfu}) $ induced by a non-$L$-symmetrizable $\cp$-decomposition of $ P_{\bfx_1,\cdots,\bfx_L} $. 
In the minimization, the distribution $ P_\bfu $ is the one given by the non-$L$-symmetrizable $\cp$-decomposition in the maximization. 
The mutual information is evaluated using the joint distribution 
\begin{align}
P_{\bfu,\bfx,\bfy} =& \sqrbrkt{ P_\bfu P_{\bfx|\bfu}U_{\bfs|\bfu}W_{\bfy|\bfx,\bfs} }_{\bfu,\bfx,\bfy}, \notag
\end{align}
where $ (P_\bfu,P_{\bfx|\bfu}) $ is the $\cp$-decomposition in the maximization. 
\end{definition}

\begin{remark}
\label{rk:seemingly_natural}
Note that the minimization is over distributions of the form $ U_{\bfs|\bfu} $ such that the jamming distribution $ P_\bfs $ satisfies the state constraints.
Operationally, this corresponds to James jamming the channel using a noise sequence $ \vbfs $ sampled from $ U_{\bfs|\bfu}^\tn $. 
Intuitively, given the definition of $L$-symmetrizability (\Cref{def:cp_symm}), it might seem more natural to let James minimize the mutual information over distributions of the form $ U_{\bfs|\bfu,\bfx_{[L]}} $ such that $ \sqrbrkt{P_\bfu P_{\bfx}^\tl U_{\bfs|\bfu,\bfx_{[L]}}}_\bfs\in\lambda_\bfs $. 
Operationally, this corresponds to James first  sampling a list $ (\vbfx_{i_1},\cdots,\vbfx_{i_L}) $ of $L$ codewords from the codebook and then sampling the noise sequence $ \vbfs $ from $ U_{\vbfs|\vbfx_{i_1},\cdots,\vbfx_{i_L}} $. 
However, somewhat counterintuitively, we show in \Cref{sec:seemingly_natural} that the value of the resulting capacity expression under such a minimization is no smaller than \Cref{eqn:cap_obli_avc_list_dec}. 
That is, this seemingly more natural jamming strategy does not give a better converse bound. 
Hence, for the purpose of upper bounding the \emph{value} of the capacity (whenever it is positive), we may restrict James to the less general type of jamming distributions $ U_{\bfs|\bfu} $. 
To avoid confusion, we stress that: when James would like to test the \emph{positivity} of the capacity, he \emph{does} employ jamming distributions of $ U_{\bfs|\bfu,\bfx_{[L]}} $ kind. 
Indeed, this is basically the $\cp$-symmetrization attack that will be described and (partly) analyzed in \Cref{sec:cp_symm}. 
\end{remark}

The main results of this paper read as follows.
\begin{theorem}[List-decoding capacity of general oblivious AVCs]\label{thm:cap_obli_avc_list_dec}
Let $ \obliavc $ be an oblivious AVC. 
Let $L\in \bZ_{\ge1} $ be the list-size. 
Let $ C_L(\cA_\obli ) $ denote the $L$-list-decoding capacity of $\cA_\obli $. 
Conditioned on \Cref{conj:comb_conj}, we have the following capacity results. 
If $ L  > L_\cp^* $, then $ C_L(\cA_\obli ) = C_L $ where $ L_\cp^* $ and $C_L $ were defined in \Cref{def:cp_symm} and \Cref{def:cap_expr}, respectively.
If $ L\le L_\cp^* $, then $ C_L(\cA_\obli )= 0 $. 
\end{theorem}

More specifically, we prove the following three theorems.

\begin{theorem}[Converse: $\cp$-symmetrization]
\label{thm:symm_converse}
Fix an oblivious AVC $ \obliavc $ and a list-size $ L\in\bZ_{\ge1} $.
If $ L\le L_\cp^* $, then $ C_L(\cA_\obli) = 0 $.
That is, there is a jamming strategy and a constant $ c_1>0 $ such that under this jamming strategy, any code $ \cC $ satisfying input constraints of rate $ R(\cC)>0 $ has average probability of error $ P_{\e,\avg}(\cC)\ge c_1 $. 
\end{theorem}
\begin{proof}
See \Cref{sec:cp_symm}. 
\end{proof}

\begin{theorem}[Achievability]
\label{thm:achievability}
Fix an oblivious AVC $ \obliavc $ and a list-size $ L\in\bZ_{\ge1} $.
If $ L>L_\cp^* $, then $ C_L(\cA_\obli)\ge C_L $ where $ C_L $ was defined in \Cref{eqn:cap_obli_avc_list_dec}. 
That is, for any $ \delta>0 $, there is a distribution over input-feasible codes  of rate $ R\le C_L -\delta $ such that, with high probability, a random code $\cC$ sampled from this ensemble has average probability of error $ P_{\e,\avg}(\cC)=o(1) $. 
\end{theorem}
\begin{proof}
See \Cref{sec:achievability}. 
\end{proof}

\begin{theorem}[Converse: capacity upper bound]
\label{thm:converse_rate_ub}
Fix an oblivious AVC $ \obliavc $ and a list-size $ L\in\bZ_{\ge1} $.
If $ L>L_\cp^* $, then conditioned on \Cref{conj:comb_conj}, $ C_L(\cA_\obli)\ge C_L $ where $ C_L $ was defined in \Cref{eqn:cap_obli_avc_list_dec}. 
That is, for any $ \delta>0 $, there is a jamming strategy and a constant $ c_2>0 $ such that under this jamming strategy,   any code $ \cC $ satisfying input constraints of rate $ R(\cC)\ge C_L +\delta $ has average probability $ P_{\e,\avg}(\cC)\ge c_2 $. 
\end{theorem}
\begin{proof}
See \Cref{sec:converse_rate_ub}. 
\end{proof}

\subsection{Comparison of our results with \cite{sarwate-gastpar-2012-list-dec-avc-state-constr}}
\label{sec:comparison_our_sg}
Using the notions of strong (\Cref{def:strong_symm}) and weak (\Cref{def:weak_symm}) symmetrizability mentioned in \Cref{sec:obli_avc_and_cp_symm}, 
Sarwate and Gastpar \cite{sarwate-gastpar-2012-list-dec-avc-state-constr} proved upper and lower bounds, respectively, on the $L$-list-decoding capacity.
Their upper and lower bounds read as follows:
\begin{align}
\wc C_L \coloneqq& \max_{ \substack{P_\bfx\in\lambda_\bfx\colon \\ P_\bfx\mathrm{\ non}\hyphen\mathrm{strongly}\hyphen L\hyphen\mathrm{symmetrizable}} } \min_{U_\bfs\in\lambda_\bfs} I(\bfx;\bfy), \notag \\
\wh C_L \coloneqq& \max_{ \substack{P_\bfx\in\lambda_\bfx\colon \\ P_\bfx\mathrm{\ non}\hyphen\mathrm{weakly}\hyphen L\hyphen\mathrm{symmetrizable}} } \min_{U_\bfs\in\lambda_\bfs} I(\bfx;\bfy), \notag 
\end{align}
where the mutual information terms are evaluated according to $ P_{\bfx,\bfy} = \sqrbrkt{P_\bfx U_\bfs W_{\bfy|\bfx,\bfs}}_{\bfx,\bfy} $. 

\begin{theorem}
Let $ \cA_\obli $ be an oblivious AVC. Let $ L\in\bZ_{\ge1} $ be the list-size. 
Then the $L$-list-decoding capacity $ C_L(\cA_\obli) $ of $ \cA_\obli $ satisfies the following.
\begin{enumerate}
	\item If $ L\le L_\strong^* $, then $ C_L(\cA_\obli) = 0 $.
	\item If $ L>L_\strong^* $, then $ C_L(\cA_\obli)\le \wc C_L $. 
	\item If $ L>L_\weak^* $, then $ C_L(\cA_\obli)\ge \wh C_L $. 
\end{enumerate}
\end{theorem}

It was open before this work whether positive rate is achievable for $ L_\strong^* <L\le L_\weak^* $. 

A visual comparison of our results with those in \cite{sarwate-gastpar-2012-list-dec-avc-state-constr} is shown in \Cref{fig:comparison}.
\begin{figure}[htbp]
	\centering
	\includegraphics[width=0.9\textwidth]{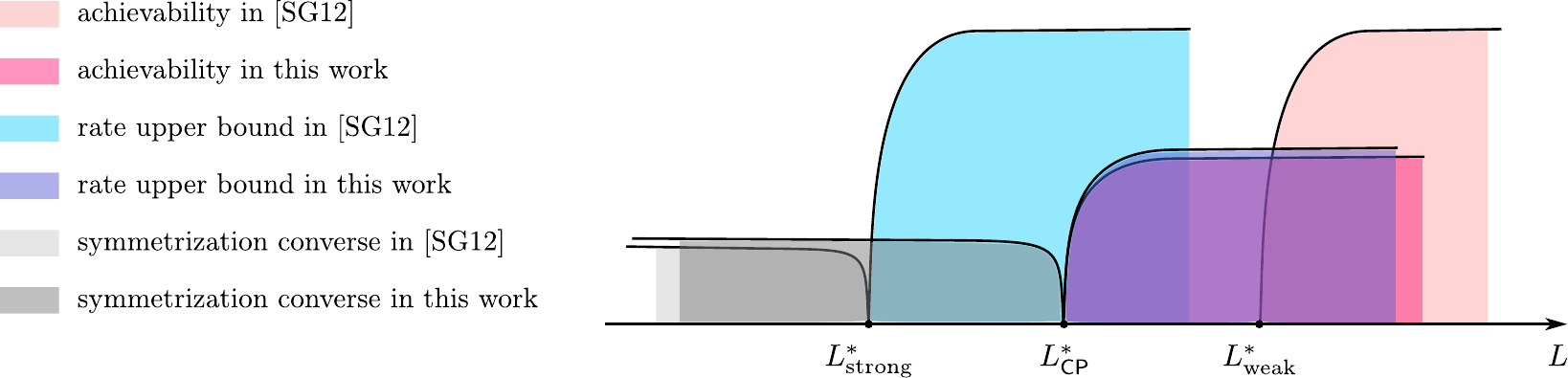}
	\caption{A visual comparison of results in this paper with those in \cite{sarwate-gastpar-2012-list-dec-avc-state-constr}.}
	\label{fig:comparison}
\end{figure}
More specifically, comparisons of achievability, symmetrization converse and capacity upper bounds are shown in \Cref{fig:comparison_ach}, \Cref{fig:comparison_ub} and \Cref{fig:comparison_symm}, respectively. 
\begin{figure}[htbp]
	\centering
	\begin{subfigure}{0.6\textwidth}
		\centering
		\includegraphics[width=1.0\linewidth]{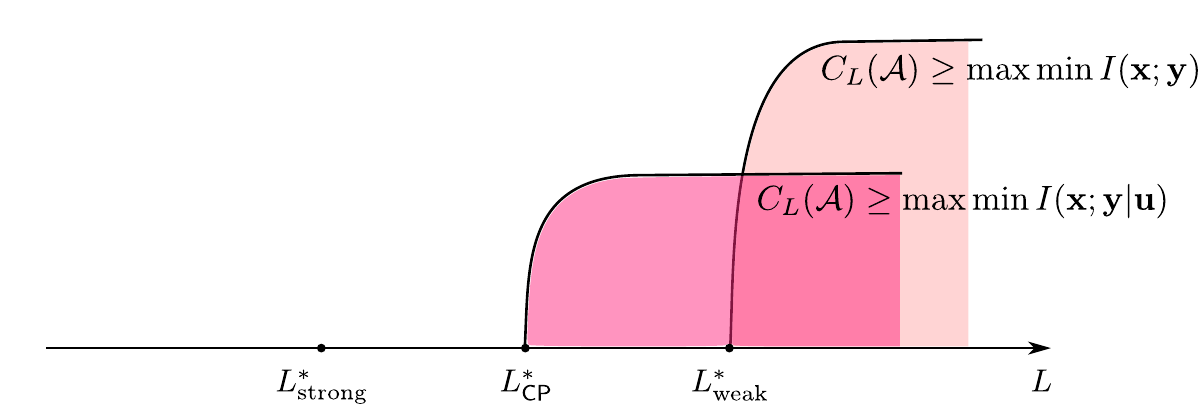}
		\caption{Comparison of achievability results.}
		\label{fig:comparison_ach}
	\end{subfigure} \\ 
	\begin{subfigure}{0.6\textwidth}
		\centering
		\includegraphics[width=1.0\linewidth]{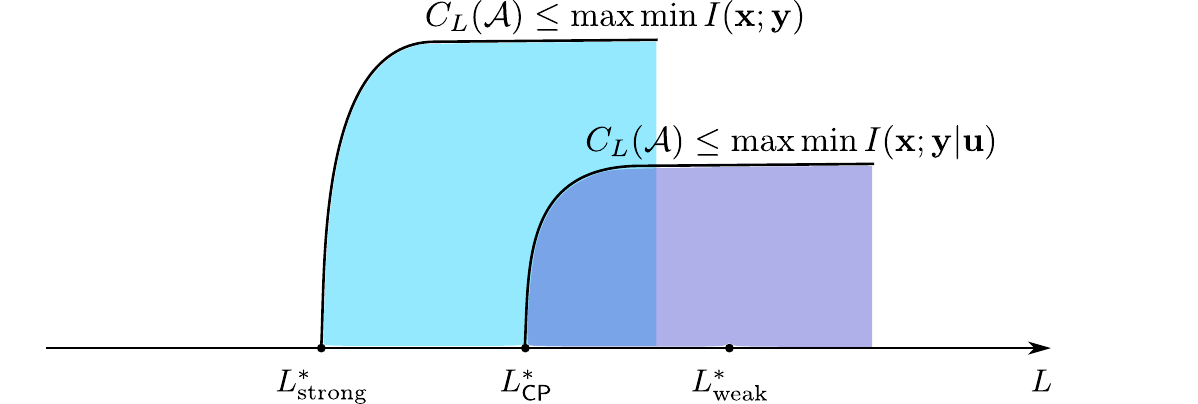}
		\caption{Comparison of capacity upper bounds.}
		\label{fig:comparison_ub}
	\end{subfigure} \\ 
	\begin{subfigure}{0.6\textwidth}
		\centering
		\includegraphics[width=1.0\linewidth]{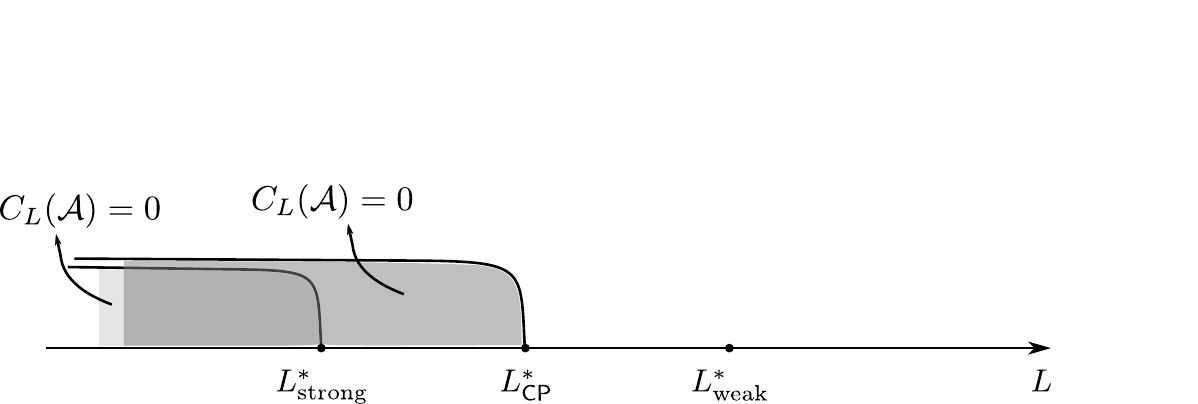}
		\caption{Comparison of symmetrization results.}
		\label{fig:comparison_symm}
	\end{subfigure}
	\caption{Comparisons of achievability, symmetrization converse and capacity upper bounds in this work and \cite{sarwate-gastpar-2012-list-dec-avc-state-constr}.}
	\label{fig:comparison_ach_ub_symm}
\end{figure}

For $L$ in different ranges, we align our results with \cite{sarwate-gastpar-2012-list-dec-avc-state-constr}'s for the readers' convenience. 
\begin{enumerate}
\item When $ L\le L_\strong^* $, our result reconciliates with \cite{sarwate-gastpar-2012-list-dec-avc-state-constr}'s both of which claim that $ C_L(\cA_\obli) = 0 $. 
\item 
When $ L_\strong^*< L\le L_\cp^* $, our symmetrization result claims that $ C_L(\cA_\obli) = 0 $, improving on \cite{sarwate-gastpar-2012-list-dec-avc-state-constr}'s upper bound $ \wc C_L $. 
\item 
When $ L> L_\cp^* $, our upper bound (conditioned on \Cref{conj:comb_conj}) $ C_L $ improves on $ \wc C_L $. 
Indeed,
\begin{align}
I(\bfx;\bfy|\bfu)=& H(\bfy|\bfu) - H(\bfy|\bfx,\bfu) \notag \\
=& H(\bfy|\bfu) - H(\bfy|\bfx) \label{eqn:mc} \\
\le& H(\bfy) - H(\bfy|\bfx) \label{eqn:cre} \\
=& I(\bfx;\bfy). \label{eqn:bound_comparison}
\end{align}
\Cref{eqn:mc} follows since $ \bfu \leftrightarrow \bfx \leftrightarrow \bfy $ forms a Markov chain.
\Cref{eqn:cre} follows from conditioning reduces entropy (\Cref{lem:properties_info_mes}). 
\item 
In the range where $ L_\cp^*<L\le L_\weak^* $, we provide an achievability result whereas there was none in \cite{sarwate-gastpar-2012-list-dec-avc-state-constr}.
In particular, our result indicates that the capacity is positive in this regime which was open before this work.
\item 
When $ L>L_\weak^* $, by \Cref{eqn:bound_comparison}, it seems that our lower bound is no better than that by Sarwate and Gastpar \cite{sarwate-gastpar-2012-list-dec-avc-state-constr}. 
However, in this regime, these two bounds in fact coincide. 
To see this, note that, by non-weak-$L$-symmetrizability, there is an input distribution $ P_\bfx\in\lambda_\bfx $ such that $ P_\bfx^\tl $ is not $L$-symmetrizable. 
Therefore we could simply take $ \cU = \curbrkt{0} $ and take $\bfu = 0$ to be a constant time-sharing variable.
Consequently, $ P_{\bfu,\bfx_1,\cdots,\bfx_L} = P_\bfx^\tl $ and $ I(\bfx;\bfy|\bfu) = I(\bfx;\bfy) $. 

For the same reason, our (conditional) upper bound $ C_L $ does not contradict the lower bound $ \wh C_L $ in \cite{sarwate-gastpar-2012-list-dec-avc-state-constr}. 
\end{enumerate}

In summary, the effect of complete positivity is two-fold.
\begin{enumerate}
	\item The introduction of the notion of $ \cp $-symmetrizability moves the critical $L^* $ (which we show equals $ L_\cp^* $) above which the $L$-list-decoding capacity is positive and at most which the $L$-list-decoding capacity is zero. 
	\item The complete positivity of the spoofing list reduces the value of the list-decoding capacity expression in the manner of conditioning on a time-sharing variable. 
\end{enumerate}

\section{$\cp$-symmetrization: an improved jamming strategy}
\label{sec:cp_symm}

We want to show that given $L$, if every $ P_\bfx\in\lambda_\bfx $ is $L$-obliviously symmetrizable\footnote{This is equivalent to the condition $ L\le L_\cp^* $.}, then no positive rate can be achieved. 
That is,  any positive rate  code $\cC \subseteq\cX^n $ cannot be $L$-list-decodable for $\cA_\obli $.
To this end, we will equip James with a jamming strategy and argue that under such a strategy the average probability of Bob's decoding error is at least some constant (independent of the blocklength $n$). 
To describe the jamming strategy, we need the following definitions, observations and lemmas.

Suppose that Alice uses a (deterministic) codebook $ \cC $. Let $ M\coloneqq|\cC| $.
Take a $\lambda$-net $\cN_1$ of $\Delta(\cX)$ of size $ N_1 = N_1(\lambda,\card{\cX}) $. Take the largest subcode $\cC'\subset\cC$ which is $(\lambda,\wh P_\bfx)$-constant-composition for some $ \wh P_\bfx\in\cN_1 $. 
Let $ M'\coloneqq|\cC'| $.
Note that $ M'\ge M/N_1 $. 

Assume $ \cC = \curbrkt{\vx_i}_{i=1}^M $ and $ \cC' = \curbrkt{\vx_j}_{j=1}^{M'} $.
Define 
\begin{align}
\Gamma_{ L } =& \Gamma^{ \tl }(\cC) \coloneqq \curbrkt{ \tau_{\vx_{i_1},\cdots,\vx_{i_L}}  \colon 1\le i_1<\cdots<i_{ L }\le M }, \notag \\
\Gamma_{ L }' =& \Gamma^{ \tl }(\cC') \coloneqq \curbrkt{ \tau_{\vx_{j_1},\cdots,\vx_{j_L}}\colon 1\le j_1<\cdots<j_{ L }\le M' }. \notag 
\end{align}

\begin{theorem}[Robust generalized Plotkin, \cite{zbj-2019-generalized-ld}]\label{thm:robust_plotkin}
Let $ \wh P_\bfx\in\Delta(\cX) $ and let $ \lambda>0 $ be a constant. 
Any $ (\lambda,\wh P_\bfx) $-constant-composition code $ \cC\subseteq\cX^n $  satisfying
\begin{align}
 d\paren{\Gamma^{ \tl }(\cC),\cp^{\otimes L }(\wh P_\bfx)}>\eps , \label{eqn:cond_robust_plotkin}
\end{align}
for some constant $ \eps\gg\lambda $, has size at most $ K =  K{(\lambda,\eps,L,\cardX)} $ (independent of $n$). 
\end{theorem}
The above theorem which allows the code to be \emph{approximately} constant-composition is in fact a slight extension to \cite{zbj-2019-generalized-ld}.
The proof is similar and we only highlight the differences in \Cref{app:robust_plotkin}.

\begin{theorem}[Hypergraph Tur\'an]\label{thm:hypergraph_turan}
If an $ L $-uniform hypergraph $ \cH^{(L)} $ on $M$ vertices contains as subgraphs no complete $L$-uniform hypergraphs $ \cK_K^{(L)} $ on $K$ vertices, then the edge density of $ \cH^{(L)} $ is at most 
\begin{align}
\frac{e(\cH^{(L)})}{\binom{M}{L}}\le & 1 - \frac{1}{\binom{{K}}{L}}, \notag 
\end{align}
where $ e(\cH^{(L)}) $ denotes the number of hyperedges in $ \cH^{(L)} $. 
\end{theorem}
\begin{remark}
In \cite{decaen-1983-hypergraph_turan}, the author proved a stronger bound 
\begin{align}
\frac{e(\cH^{(L)})}{\binom{M}{L}}\le & 1 - \frac{1}{\binom{{K-1}}{L-1}}. \notag 
\end{align}
However, the simple bound in \Cref{thm:hypergraph_turan} suffices for the purposes of this paper. 
The readers are encouraged to refer to the survey by Keevash \cite{keevash-2011-hypergraph-turan-survey} on hypergraph Tur\'an problem.
\end{remark}

\begin{lemma}[$\cp$-list extraction]\label{lem:apply_turan}
Let $ \wh P_\bfx\in\Delta(\cX) $. 
Fix constants $ \lambda>0 $ and  $ \eps\gg\lambda $.
For any $ (\lambda,\wh P_\bfx) $-constant-composition code $ \cC = \curbrkt{\vx_i}_{i=1}^M \subseteq\cX^n $, 
there exists a constant $  \nu = \nu(\lambda,\eps,L,\cardX) >0 $ (independent of $n$) such that
\begin{align}
\frac{1}{\binom{M}{ L }}\card{ \curbrkt{ \tau_{\vx_{i_1},\cdots,\vx_{ i_L }}\colon 1\le i_1<\cdots<i_{ L }\le M,\;d\paren{ \tau_{\vx_{i_1},\cdots,\vx_{i_L}},\cp^{\otimes L }(\wh P_\bfx) }\le\eps } }\ge \nu. \notag
\end{align}
\end{lemma}
\begin{proof}
Let $\cH_\cC $ be a hypergraph consisting of all codewords in $\cC$ as vertices. An $L$-tuple of codewords $ \vx_{i_1},\cdots,\vx_{i_{ L }} $ ($ 1\le i_1<\cdots<i_{ L }\le M $) is connected by a hyperedge if $ d\paren{\tau_{\vx_{i_1},\cdots,\vx_{i_{ L }}},\cp^{\otimes L }(\wh P_\bfx)} >\eps $. 
Then by Plotkin bound, no subcode $ \cC' = \curbrkt{\vx_{j}}_{j=1}^{M'} \subset\cC $ of size larger than $ M'> K=K\paren{ \lambda,\eps,L,\cardX} $ satisfies $ d\paren{\tau_{\vx_{j_1},\cdots,\vx_{j_{ L }}}, \cp^{\otimes L }(\wh P_\bfx)} > \eps $ for all $ 1\le j_1<\cdots<j_{ L }\le M' $. 
Said differently, $\cH_\cC $ has no clique of size larger than $ K $. Then by the hypergraph Tur\'an's theorem (\Cref{thm:hypergraph_turan}), the edge density of $\cH_\cC$ is at most $ 1- \nu <1 $, where $ \nu = \binom{K}{L}^{-1} $ . That is, at least a $  \nu $ fraction of order-$L$ joint types of $\cC$ are $\eps$-close to $ \cp^{\otimes L }(\wh P_\bfx) $. (See \Cref{fig:many_cp_types}.)
\begin{figure}[htbp]
	\centering
	\includegraphics[width=0.4\textwidth]{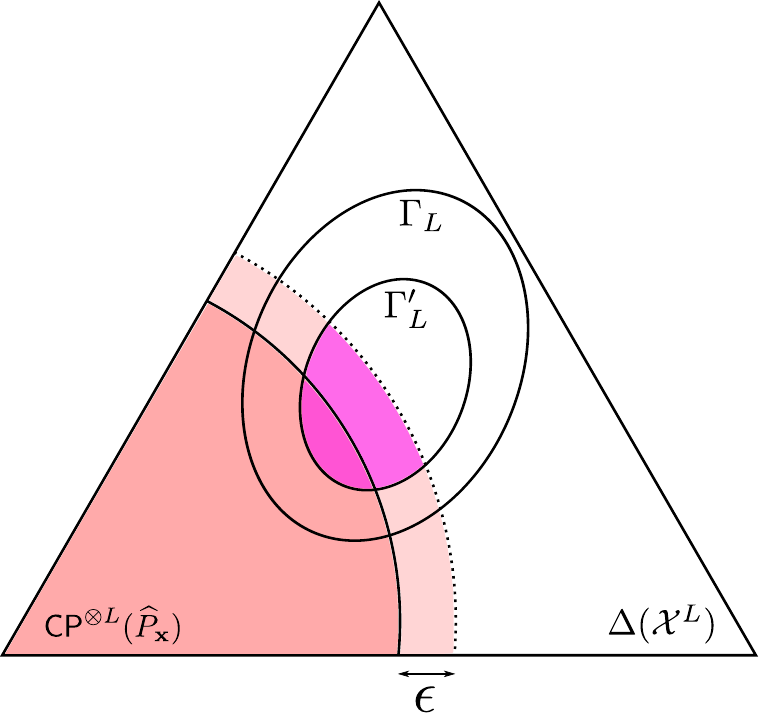}
	\caption{Any sufficiently large code must contain many completely positive joint types.}
	\label{fig:many_cp_types}
\end{figure}
\end{proof}

For a subset $\cA\subseteq\Delta$ of a metric space $ (\Delta,\dist) $, let $(\cA)_\eps $ denote its \emph{$\eps$-enlargement}, i.e., $ (\cA)_\eps\coloneqq\curbrkt{a\in\Delta\colon \dist(a,\cA)\le\eps } $. In particular, a ball centered around a point $a$ of radius $\eps$  can be denoted by $(a)_{\eps}$.

Take an $\eta$-net $ \cN_2 $ of $ \paren{\cp^{\otimes L }(\wh P_\bfx)}_\eps\cap\Delta(\cX^{ L }) $ of size $ N_2\coloneqq N_2(\eps,\eta,L,\cardX) $. 
By \Cref{lem:apply_turan}, we know that the fraction of order-$L$ joint types of $ \cC' $ that are $\eps$-close to $\cp^{\otimes L }(\wh P_\bfx)$ is at least $  \nu  $. 
By Markov's inequality (\Cref{lem:markov}), there exists a distribution $ \wh P_{\bfx_1,\cdots,\bfx_{ L }}\in\cN_2   $ such that 
\begin{align}
\frac{\card{\Gamma_{ L }'\cap\paren{\cp^{\otimes L }(\wh P_\bfx)}_\eps\cap\paren{\wh P_{\bfx_1,\cdots,\bfx_{ L }}}_\eta}}{\card{\Gamma_{ L }'\cap\paren{\cp^{\otimes L }(\wh P_\bfx)}_\eps}}
=&\frac{\card{\curbrkt{ \tau_{\vx_{j_1},\cdots,\vx_{j_{ L }}}\colon \begin{array}{l}
1\le j_1<\cdots<j_{ L }\le M', \\
d\paren{ \tau_{\vx_{j_1},\cdots,\vx_{j_{ L }}},\cp^{\otimes L }(\wh P_\bfx) }\le\eps \\
d\paren{\tau_{\vx_{j_1},\cdots,\vx_{j_{ L }}},\wh P_{\bfx_1,\cdots,\bfx_{ L }}} \le \eta
\end{array} }} }
{ \card{\curbrkt{ \tau_{\vx_{j_1},\cdots,\vx_{j_{ L }}}\colon \begin{array}{l}
1\le j_1<\cdots<j_{ L }\le M', \\
d\paren{ \tau_{\vx_{j_1},\cdots,\vx_{j_{ L }}},\cp^{\otimes L }(\wh P_\bfx) }\le\eps
\end{array} }} } 
\ge 1/N_2. \notag
\end{align}
(See. \Cref{fig:quantize_joint_types}.)
\begin{figure}[htbp]
	\centering
	\includegraphics[width=0.5\textwidth]{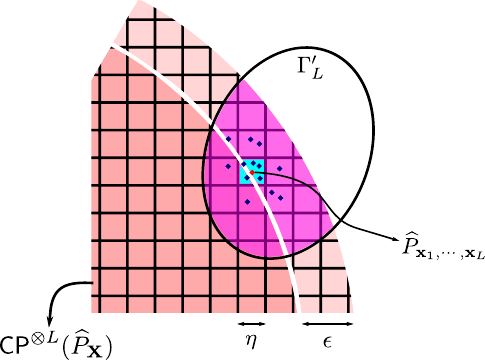}
	\caption{Quantize the set of joint types and find an approximately equicoupled subcode.}
	\label{fig:quantize_joint_types}
\end{figure}
That is
\begin{align}
\card{\Gamma_{ L }'\cap\paren{\cp^{\otimes L }(\wh P_\bfx)}_\eps\cap\paren{\wh P_{\bfx_1,\cdots,\bfx_{ L }}}_\eta}\ge& \card{\Gamma_{ L }'\cap\paren{\cp^{\otimes L }(\wh P_\bfx)}_\eps}/N_2 \ge \frac{\card{\Gamma'_{ L }} \nu }{N_2}. \notag
\end{align}
Note that
\begin{align}
\card{\Gamma_{ L }'} =& \binom{M'}{ L }\ge\binom{M/N_1}{ L }, \quad
\card{\Gamma_{ L }} = \binom{M}{ L }. \notag
\end{align}
Therefore 
\begin{align}
\frac{ \card{\Gamma_{ L }'\cap\paren{\cp^{\otimes L }(\wh P_\bfx)}_\eps\cap\paren{\wh P_{\bfx_1,\cdots,\bfx_{ L }}}_\eta} }{ \card{\Gamma_{ L }} }
\ge& \frac{\binom{M/N_1}{ L } \nu }{\binom{M}{ L } N_2}\ge \frac{\nu\paren{\frac{M/N_1}{L}}^L}{N_2\paren{\frac{Me}{L}}^L} = \frac{\nu}{N_2(N_1e)^L}\eqqcolon c,
\label{eqn:cp_tuple_extraction}
\end{align}
which is at least a constant, denoted by $c>0$,



Note that  $ \wh P_{\bfx_1,\cdots,\bfx_{ L }} \in\cN_2\subset \paren{\cp^{\otimes L }(\wh P_\bfx)}_\eps\cap\Delta(\cX^L) $ may not be exactly $\cp$. 
We project it back to $ \cp^\tl(\wh P_\bfx) $ and get $ \wt P_{\bfx_1,\cdots,\bfx_L}\in\cp^\tl(\wh P_\bfx) $. 
That is
\begin{align}
\wt P_{\bfx_1,\cdots,\bfx_L}\coloneqq& \argmin{P_{\bfx_1,\cdots,\bfx_L}\in\cp^\tl(\wh P_\bfx)}d\paren{ \wh P_{\bfx_1,\cdots,\bfx_L},P_{\bfx_1,\cdots,\bfx_L} }. \notag 
\end{align}
We automatically have 
\begin{align}
 d\paren{\wh P_{\bfx_1,\cdots,\bfx_L}, \wt P_{\bfx_1,\cdots,\bfx_L}}\le\eps. \label{eqn:dist_hat_tilde}
 \end{align} 

We are now ready to specify the jamming strategy James is going to use.
Given a (deterministic) codebook $ \cC $ (which is known to every party), assume $ \vbfx $ (which James does not know) was transmitted by Alice.
James first samples randomly $L$ codewords $ \vbfx_1,\cdots,\vbfx_L $ from $ \cC $.

Since $ L\le L_\cp^* $, every $ P_\bfx\in\lambda_\bfx $ is $L$-symmetrizable. 
In particular, $ \wh P_\bfx\in\lambda_\bfx $ is also symmetrizable. 
Since  $ \wt P_{\bfx_1,\cdots,\bfx_{ L }} $ constructed above is a completely positive $ \wh P_\bfx $-self-coupling, 
for any $\cp$-decomposition $ (\wt P_\bfu,\wt P_{\bfx|\bfu}) $ of $ \wt P_{\bfx_1,\cdots,\bfx_L} $, 
there is a symmetrizing distribution $ U_{\bfs|\bfu,\bfx_1,\cdots,\bfx_{ L }} \in\cU_{ L } $ associated to $(\wt P_\bfu,\wt P_{\bfx|\bfu})$.  
James then sample $ \vbfs\sim \sqrbrkt{ \wt P_\bfu \wt P_{\bfx|\bfu}^\tl U_{\bfs|\bfu,\bfx_1,\cdots,\bfx_{ L }} }_\bfs^{\otimes n} $.

We then briefly sketch why the above strategy enforces the probability of Bob's decoding error to be bounded away from zero. 
The formal proof is relegated to \Cref{app:converse_rate_zero}. 
Assume that for  some positive constants $ \delta_i>0 $ ($1\le i\le\beta$),
\begin{align}
\cost_i((\wt P_\bfu,\wt P_{\bfx|\bfu}) ,U_{\bfs|\bfu,\bfx_{[L]}}) =&  \sum_{(u,x_{[L]},s)\in\cU\times\cX^L\times\cS}\wt P_\bfu(u) \wt P_{\bfx|\bfu}^\tl(x_{[L]}|u)U_{\bfs|\bfu,\bfx_{[L]}}(s|u,x_{[L]})B_i(s) \le \Lambda_i-\delta_i. \label{eqn:cost_assumption} 
\end{align}
By \Cref{eqn:dist_hat_tilde}, it can be shown that
\begin{multline}
\left| \sum_{(u,x_{[L]},s)\in\cU\times\cX^L\times\cS} \tau_{\vu,\vx_{\cL}}(u,x_{[L]})U_{\bfs|\bfu,\bfx_{[L]}}(s|u,x_{ [L] })B_i(s) \right. \\
\left. - \sum_{(u,x_{[L]},s)\in\cU\times\cX^L\times\cS} \wt P_\bfu(u) \wt P_{\bfx|\bfu}^\tl(x_{[L]}|u)U_{\bfs|\bfu,\bfx_{[L]}}(s|u,x_{ [L] })B_i(s) \right| \le g(\eta,\eps) . \notag
\end{multline}
Therefore
\begin{align}
\cost_i(\tau_{\vu,\vx_\cL},U_{\bfs|\bfu,\bfx_{[L]}}) \le& \cost_i((\wt P_\bfu,\wt P_{\bfx|\bfu}),U_{\bfs|\bfu,\bfx_{[L]}}) + g(\eta, \eps) \le \Lambda_i-\delta_i+g(\eta, \eps) < \Lambda_i, \notag
\end{align}
by choosing proper $\eta$ and $\eps$.

Over the randomness of Alice's choice of $\vbfx$ and James's choice of $ \vbfx_1,\cdots,\vbfx_L $, the joint type $ \tau_{\vbfx_1,\cdots,\vbfx_L} $ is $(\eta+\eps)$-close to $ \wt P_{\bfx_{[L]}} $ with probability at least $c$, given $ \vbfx,\vbfx_1,\cdots,\vbfx_L $ are distinct (which is not case with probability only $ 1 - \frac{M-1}{M}\frac{M-2}{M}\cdots\frac{M-L}{M} $ -- vanishingly small in $n$). 
Since James uses a symmetrizing distribution $ U_{\bfs|\bfu,\bfx_{[L]}} $ tailored for $ \wt P_{\bfx_{[L]}} $, the channel is symmetrized under the aforementioned condition. 
Then from Bob's perspective, any $L$-subset of $ \vbfx,\vbfx_1,\cdots,\vbfx_L $ is equally likely and he could not distinguish them.
Consequently, no matter which list-decoder Bob uses, it makes an error with probability at least $1/L$ -- no better than randomly output an $L$-subset of the candidate $(L+1)$-list. 
The above intuition combined with the calculation in this section can be formalized using arguments along the lines of \cite{csiszar-narayan-it1988-obliviousavc,hughes-1997-list-avc,sarwate-gastpar-2012-list-dec-avc-state-constr}.
See \Cref{app:converse_rate_zero} for details.

\section{Achievability}
\label{sec:achievability}
In this section, we prove \Cref{thm:achievability}.
The proof is along the lines of \cite{csiszar-narayan-it1988-obliviousavc,hughes-1997-list-avc,sarwate-gastpar-2012-list-dec-avc-state-constr}.

Assume that $ L>L_{\cp}^* $. 
It suffices to take $ L = L_{\cp}^* + 1 $. 
By the definition of $\cp$-symmetrizability (\Cref{def:cp_symm}), there exists a feasible input distribution $ P_\bfx\in\lambda_\bfx $, a $\cp$ $ P_\bfx $-self-coupling $ P_{\bfx_1,\cdots,\bfx_L}\in\cp^\tl(P_\bfx) $ and a $\cp$-decomposition $ (P_\bfu, P_{\bfx|\bfu}) $ such that there is no symmetrizing distribution $ U_{\bfs|\bfu,\bfx_1,\cdots,\bfx_L}\in\cU_{\obli,\symml{L}} $ satisfying power constraint $ \sqrbrkt{P_{\bfu} P_{\bfx|\bfu}^\tl U_{\bfs|\bfu,\bfx_1,\cdots,\bfx_L} }_\bfs \in\lambda_\bfs $.

Throughout the proof of achievability, fix any distributions $ P_\bfx\in\lambda_\bfx $ and $ P_{\bfx_1,\cdots,\bfx_L}\in\cp^\tl(P_\bfx) $  satisfying the above conditions.
Let $ P_{\bfx_1,\cdots,\bfx_L} = \sum_{i = 1}^k\lambda_iP_{\bfx_i}^\tl $ for some $  k\ge \cprk(P_{\bfx_1,\cdots,\bfx_L}) $ be any non-$L$-symmetrizable $\cp$-decomposition of $ P_{\bfx_1,\cdots,\bfx_L} $. 
We introduce a time-sharing variable $ \bfu $ supported on $ [k] $ with probability mass vector $ P_\bfu = [\lambda_1,\cdots,\lambda_k] $. 
By the construction of $ \bfu $, $ P_\bfu $ has no zero atoms. 
The $ \cp $-decomposition can be alternatively written as $ P_{\bfx_1,\cdots,\bfx_L} = \sqrbrkt{ P_{\bfu}P_{\bfx|\bfu}^\tl }_{\bfx_1,\cdots,\bfx_L} $ where $ P_{\bfx|\bfu = i} = P_{\bfx_i} $ for any $ i\in[k] $. 

Let $ \cU\coloneqq[k] $. 
For a list of messages $ \cL = \curbrkt{i_1,\cdots,i_L},1\le i_1<\cdots<i_L $, we sometimes use $ \vv_{\cL} $ to denote the sequence of vectors $ \vv_{i_1},\cdots,\vv_{i_L} $.
Similar notation is used for random variables $ \bfv_\cL $.

\subsection{Codebook design}
\label{sec:codebook_design}
\begin{lemma}[Codeword selection]
\label{lem:cw_select}
Fix any $ L\in\bZ_{>0} $, $ \eps>0 $, sufficiently large $n \in\bZ_{>0} $, $ M\ge L2^{n\eps} $, $ \vu\in\cU^n $ of type $ P_\bfu\in\Delta(\cU) $ and conditional type $ P_{\bfx|\bfu}\in\Delta(\cX|\cU) $.
Let  $ R = \frac{1}{n}\log\frac{M}{L} $. 
Then there exist codewords $ \vx_1,\cdots,\vx_M\in\cX^n $, each of joint type $ P_{\bfu,\bfx} = P_\bfu P_{\bfx|\bfu} $ with $\vu$, satisfying that for every $ \vx\in\cX^n,\vs\in\cS^n $, joint type $ P_{\bfu,\bfx,\bfx_1,\cdots,\bfx_L,\bfs}\in\Delta(\cU\times\cX\times\cX^L\times\cS) $ and $ k\in[L] $,
\begin{enumerate}
	\item if $ I(\bfx;\bfs|\bfu)\ge\eps $, then 
	\begin{align}
	&\frac{1}{M}\card{\curbrkt{ i\in[M]\colon \tau_{\vu,\vx_i,\vs} = P_{\bfu,\bfx,\bfs} }}\le2^{-n\eps/2};
	\label{eqn:property_1}
	\end{align}
	\item if $ I(\bfx;\bfx_k,\bfs|\bfu)\ge\sqrbrkt{ R - I(\bfx_k;\bfs|\bfu) }^++\eps $, then
	\begin{align}
	&\frac{1}{M}\card{\curbrkt{ i\in[M]\colon\exists j\in[M]\backslash\curbrkt{i},\; \tau_{\vu,\vx_i,\vx_j,\vs} = P_{\bfu,\bfx,\bfx_k,\bfs} }}\le2^{-n\eps/2}; 
	\label{eqn:property_2}
	\end{align}
	\item and 
	\begin{align}
	&\card{\curbrkt{ j\in[M]\colon \tau_{\vu,\vx,\vx_j,\vs} = P_{\bfu,\bfx,\bfx_k,\bfs} }}\le2^{ n\paren{ \sqrbrkt{R - I(\bfx_k;\bfx,\bfs|\bfu)}^++\eps } }.
	\label{eqn:property_3}
	\end{align}
\end{enumerate}
Moreover, if $ R<I(\bfx_k;\bfs|\bfu) $ for all $ k\in[L] $, then $ \vx_1,\cdots,\vx_M $ can be selected to further satisfy that
\begin{enumerate}
	\item[2')] if $ I(\bfx;\bfx_{[L]},\bfs|\bfu)\ge\eps $, then 
	\begin{align}
	&\frac{1}{M}\card{\curbrkt{i\in[M]\colon\exists\cL\in\binom{[M]\backslash\curbrkt{i}}{L},\; \tau_{\vu,\vx_i,\vx_\cL,\vs} = P_{\bfu,\bfx,\bfx_{[L]},\bfs}}}\le2^{-n\eps/2}; 
	\label{eqn:property_22}
	\end{align}
	\item[3')] and
	\begin{align}
	\card{\curbrkt{\cL\in\binom{[M]}{L}\colon \tau_{\vu,\vx,\vx_\cL,\vs} = P_{\bfu,\bfx,\bfx_{[L]},\bfs}}}\le2^{n\eps}.
	\label{eqn:property_33}
	\end{align}
\end{enumerate}
\end{lemma}

\subsection{Decoding rules}
\label{sec:dec_rules}
Define the set of joint distributions that $\eta$-approximately (w.r.t. KL-divergence) obey the channel structure:
\begin{align}
\cP_\eta \coloneqq& \curbrkt{P_{\bfu,\bfx,\bfs,\bfy}\in\Delta(\cU\times\cX\times\cS\times\cY)\colon D(P_{\bfu,\bfx,\bfs,\bfy}\|P_\bfu P_{\bfx|\bfu}P_\bfs W_{\bfy|\bfx,\bfs} )\le\eta,\;P_\bfs\in\lambda_\bfs }. \label{eqn:def_p_eta} 
\end{align}

Given the codebook $ \cC = \curbrkt{\vx_{i}}_{i\in[M]} $, the time-sharing sequence $ \vu\in\cU^n $ and Bob's observation $ \vy\in\cY^n $, the decoder $\psi$ outputs all messages $ i\in[M] $ satisfying: 
there is a jamming vector  $ \vs\in\Lambda_\bfs $ such that
\begin{enumerate}
	\item\label{itm:dec_step1}   $\tau_{ \vu, \vx_i,\vs,\vy }\in\cP_\eta $; 
	\item\label{itm:dec_step2} for all $ \cL' = \curbrkt{ \vx_{i_1},\cdots,\vx_{i_L} },1\le i_1<\cdots<i_L\le M $ satisfying that
	\begin{itemize}
		\item there is a collection of jamming vectors $ \curbrkt{\vs_j}_{j\in[L]}\subset\Lambda_\bfs $ such that   $ \tau_{\vu,\vx_{i_j},\vs_j,\vy }\in\cP_\eta $ for each $ j\in[L] $,
	\end{itemize}
	it holds that $ I(\bfx,\bfy;\bfx_1,\dots,\bfx_L|\bfu,\bfs)\le\eta $
	where the mutual information is evaluated w.r.t $ P_{\bfu,\bfx,\bfx_1,\cdots,\bfx_L,\bfs,\bfy} = \tau_{\vu,\vx_i,\vx_{i_1},\cdots,\vx_{i_L},\vs,\vy} $. 
\end{enumerate}

The first step of the decoding rules can be viewed as a list-decoding step.
Bob can use the typicality condition in \Cref{itm:dec_step1} to list-decode to a list $ \cL_0 $ of at most $ \cO(\poly(n)) $ messages. 
The second step can be viewed as a tournament. 
Bob examines all pairs $ (\cL_1,\cL_2)\subseteq\cL_0\times\cL_0 $ of distinct $L$-lists of messages in the big list $ \cL_0 $ and determines if $ \cL_1 $ or $ \cL_2 $ wins. 
If there is a unique champion in this competition, then the decoder outputs that list. 

It is clear that the correct message $i$ will be output with high probability.
We also need to argue that such a decoder is a valid list-decoder such that  $ |\psi(\vy)|\le L $ for any channel output $ \vy\in\cY^n $.

\subsection{Unambiguity of decoding rules}\label{sec:unambiguity_dec_rules}
\label{sec:unambiguity}
\begin{lemma}\label{lem:unambiguity}
There is no joint distribution $ P_{\bfu,\bfx_{[L+1]},\bfs_{[L+1]},\bfy} $ that simultaneously satisfies
\begin{align}
P_{\bfu,\bfx_i} = P_{\bfu,\bfx},\; P_{\bfu,\bfx_i,\bfs_i,\bfy}\in\cP_\eta,\;I(\bfx_i,\bfy;\bfx_{{[L+1]}\setminus\curbrkt{i}}|\bfu,\bfs_i) \le \eta, \label{eqn:unambiguity_cond}
\end{align}
for all $ i\in{[L+1]} $. 
\end{lemma}

\subsection{Analysis of average error probability}\label{sec:analysis_error_prob}
\label{sec:ach_analysis_error_prob}
Fix a non-$L$-symmetrizable $P_\bfx $ and a corresponding $ P_{\bfx_1,\cdots,\bfx_L}\in\cp^\tl(P_\bfx) $ which  induces $ P_\bfu $ and $ P_{\bfx|\bfu} $. 
For each $j \in[\beta] $, define the minimum power that James has to spend to symmetrize the channel, 
\begin{align}
\Lambda_j^*(P_\bfu, P_{\bfx|\bfu}) \coloneqq& \min_{U_{\bfs|\bfu,\bfx_{[L]}}\in\cU_{\obli,\symml{L}} } \sum_{(u,x_{[L]},s)\in\cU\times\cX^L\times\cS}P_\bfu(u)P_{\bfx|\bfu}^\tl(x_{[L]}|u)U_{\bfs|\bfu,\bfx_{[L]}}(s|u,x_{[L]})B_j(s). \notag 
\end{align}
By non-$L$-symmetrizability, we know that there must exist a $ j\in[\beta] $ such that $ \Lambda_j^*(P_\bfu, P_{\bfx|\bfu})>\Lambda_j $. 
Focus on that $j$ from now on. 
Fix an arbitrary $ \vs\in\Lambda_\bfs $. 

Define the set of \emph{bad messages}
\begin{align}
\cM(\vu)\coloneqq& \curbrkt{m\in[M]\colon I(\bfx;\bfs|\bfu)>\eps }, \label{def:bad_msg}
\end{align}
where the mutual information is evaluated w.r.t. the joint distribution $P_{\bfu,\bfx,\bfs} = \tau_{\vu,\vx_m,\vs}$. 

According to the definition of the list-decoder (\Cref{sec:dec_rules}), for any fixed $ \vs\in\Lambda_\bfs $,  a list-decoding error occurs only if:
codeword $ \vx_i $ corresponding to message $i$ was sent by Alice and $ \vy $ was received by Bob, but
\begin{enumerate}
	\item \textbf{either} $ (\vu,\vx_i,\vs,\vy) $ are not jointly typical in the sense that $ \tau_{\vu,\vx_i,\vs,\vy} \notin\cP_\eta $;
	\item \textbf{or} a spoofing list of size $L$ that does not contain message $i$ confuses Bob, i.e.,  there is an $L$-list $ \cL'\in\binom{[M]\setminus\curbrkt{i}}{L} $ such that the joint type $ P_{\bfu,\bfx_i,\bfx_{\cL'},\bfs,\bfy}\in\Delta^{(n)}(\cU\times\cX^{L+1}\times\cS\times\cY) $ of the tuple $ (\vu,\vx_i,\vx_{\cL'},\vs,\vy) $ satisfies the following conditions:
	\begin{enumerate}
		\item $ (\vu,\vx_i,\vs,\vy) $ are jointly typical: $ P_{\bfu,\bfx_i,\bfs,\bfy} \in \cP_\eta $;
		\item each codeword corresponding to a candidate message $i'$ in the competing list $ \cL' $ is also typical: for each $ {i'}\in\cL' $, there exists an $ \vs_{i'}\in\Lambda_\bfs $ such that $ (\vu,\vx_{i'},\vs_{i'},\vy) $ has type $ P_{\bfu,\bfx_{i'},\bfs_{i'},\bfy}\in\cP_\eta $;
		\item and finally $ I(\bfx_i,\bfy;\bfx_{\cL'}|\bfu,\bfs)>\eta $. 
	\end{enumerate}
\end{enumerate}
Motivated by the above conditions, we define the set of \emph{bad distributions}
\begin{align}
\cD \coloneqq\curbrkt{ P_{\bfu,\bfx ,\bfx_{[L]},\bfs,\bfy}\in\Delta(\cU\times\cX\times\cX^L \times\cS\times\cY)\colon \begin{array}{rl}
P_{\bfu,\bfx ,\bfs,\bfy}\in&\cP_\eta \\
\forall i\in\cL, \exists \bfs_i,\;P_{\bfu,\bfx_i,\bfs_i,\bfy}\in&\cP_\eta \\
I(\bfx ,\bfy;\bfx_{[L]}|\bfu,\bfs)>&\eta
\end{array} }.
\label{eqn:def_d} 
\end{align}

The average error probability under $\vs$ is 
\begin{align}
P_{\e,\avg}(\vs) =& \frac{1}{M}\sum_{m\in[M]} P_{\e}(m,\vs) \le \frac{1}{M}\sum_{m\in\cM(\vu)}P_\e(m,\vs) + \frac{1}{M}\sum_{m\in[M]\setminus\cM(\vu)}P_\e(m,\vs) \notag \\
\le& \frac{1}{M}\card{\curbrkt{m\in[M]\colon I(\bfx;\bfs|\bfu)>\eps}} \label[term]{ineq:term1} \\
&+ \frac{1}{M}\sum_{m\in[M]\setminus\cM(\vu)}\sum_{\vy\in\cY^n\colon\tau_{\vu,\vx_m,\vs,\vy}\notin\cP_\eta} W\paren{\vy\condon\vx_m,\vs} \label[term]{eqn:term2} \\
&+ \frac{1}{M}\sum_{m\in[M]\setminus\cM(\vu)} \sum_{\tau_{\bfu,\bfx ,\bfx_{[L]},\bfs,\bfy}\in\cD} P_{\e}(m,\vs,\tau_{\bfu,\bfx ,\bfx_{[L]},\bfs,\bfy}), \label[term]{eqn:term3}
\end{align}
where 
\begin{align}
P_\e(m,\vs,\tau_{\bfu,\bfx ,\bfx_{[L]},\bfs,\bfy})\coloneqq& \sum_{\vy\in\cY^n\colon \exists \cL'\in\binom{[M]\setminus\curbrkt{m}}{L},\; \tau_{\vu,\vx_m,\vx_{\cL'},\vs,\vy} = \tau_{\bfu,\bfx ,\bfx_{[L]},\bfs,\bfy}}W\paren{\vy\condon\vx_m,\vs}. \notag
\end{align}
By \Cref{eqn:property_1}, \Cref{ineq:term1} is at most 
\begin{align}
 \frac{1}{M}\card{\curbrkt{m\in[M]\colon I(\bfx;\bfs|\bfu)>\eps}} =& \frac{1}{M}\sum_{ \substack{\tau_{\bfu,\bfx,\bfs}\in\Delta^{(n)}(\cU\times\cX\times\cS)\colon \\ \tau_{\bfu,\bfx} = P_{\bfu}P_{\bfx|\bfu},\tau_{\bfs}=\tau_\vs} }\card{\curbrkt{m\in[M]\colon \tau_{\vu,\vx_m,\vs} = P_{\bfu,\bfx,\bfs}}}\indicator{I(\bfx;\bfs|\bfu)>\eps}\dotle2^{-n\eps/2}. \notag
\end{align}
By Sanov's theorem (\Cref{lem:sanov}), \Cref{eqn:term2} is dot less than
\begin{align}
\sup_{P_{\bfu,\bfx,\bfs,\bfy}\notin\cP_\eta\colon I(\bfx;\bfs|\bfu)>\eps} 2^{-n\kl{P_{\bfu,\bfx,\bfs,\bfy}}{P_{\bfu,\bfx,\bfs}W_{\bfy|\bfx,\bfs}}} =& \sup_{P_{\bfu,\bfx,\bfs,\bfy}\notin\cP_\eta\colon I(\bfx;\bfs|\bfu)>\eps} 2^{-n\paren{\kl{P_{\bfu,\bfx,\bfs,\bfy}}{P_\bfu P_{\bfx|\bfu}P_\bfs W_{\bfy|\bfx,\bfs}} - I(\bfx;\bfs|\bfu)}} \le 2^{-n(\eta-\eps)}. \notag 
\end{align}
The inequality follows from the definition of $ \cP_\eta $ (\Cref{eqn:def_p_eta}) and the condition $ I(\bfx;\bfs|\bfu)>\eps $. 

In what follows, we focus on \Cref{eqn:term3}. 
We further define the following sets of  distributions, 
\begin{align}
\cD_1\coloneqq&\curbrkt{ P_{\bfu,\bfx ,\bfx_{[L]},\bfs,\bfy}\in\Delta(\cU\times\cX\times\cX^L \times\cS\times\cY)\colon R<\min_{i\in[L]} I(\bfx_i;\bfs|\bfu) }, \label{eqn:def_d1} \\
\cD_2\coloneqq&\curbrkt{ P_{\bfu,\bfx ,\bfx_{[L]},\bfs,\bfy}\in\Delta(\cU\times\cX\times\cX^L \times\cS\times\cY)\colon I(\bfx ;\bfx_{[L]},\bfs|\bfu)>\eps }. \label{eqn:def_d2} 
\end{align}
Now \Cref{eqn:term3} can be decomposed as
\begin{align}
& \frac{1}{M}\sum_{m\in[M]\setminus\cM(\vu)} \sum_{\tau_{\bfu,\bfx ,\bfx_{[L]},\bfs,\bfy}\in\cD\cap\cD_1\cap\cD_2} P_{\e}(m,\vs,\tau_{\bfu,\bfx ,\bfx_{[L]},\bfs,\bfy}) \label[term]{eqn:term31} \\
&+\frac{1}{M}\sum_{m\in[M]\setminus\cM(\vu)} \sum_{\tau_{\bfu,\bfx ,\bfx_{[L]},\bfs,\bfy}\in\cD\cap\cD_1\cap\cD_2^c} P_{\e}(m,\vs,\tau_{\bfu,\bfx ,\bfx_{[L]},\bfs,\bfy}) \label[term]{eqn:term32}\\
&+\frac{1}{M}\sum_{m\in[M]\setminus\cM(\vu)} \sum_{\tau_{\bfu,\bfx ,\bfx_{[L]},\bfs,\bfy}\in\cD\cap\cD_1^c} P_{\e}(m,\vs,\tau_{\bfu,\bfx ,\bfx_{[L]},\bfs,\bfy}). \label[term]{eqn:term33}
\end{align}
\Cref{eqn:term31} can be bounded as follows 
\begin{align}
&\frac{1}{M}\sum_{m\in[M]\setminus\cM(\vu)} \sum_{\tau_{\bfu,\bfx ,\bfx_{[L]},\bfs,\bfy}\in\cD\cap\cD_1\cap\cD_2}\sum_{\vy\in\cY^n}W\paren{\vy\condon\vx_m,\vs}\indicator{ \exists \cL'\in\binom{[M]\setminus\curbrkt{m}}{L},\; \tau_{\vu,\vx_m,\vx_{\cL'},\vs,\vy} = \tau_{\bfu,\bfx ,\bfx_{[L]},\bfs,\bfy}} \notag \\
\le& \frac{1}{M}\sum_{\tau_{\bfu,\bfx ,\bfx_{[L]},\bfs,\bfy}\in\cD_1\cap\cD_2} \sum_{m\in[M]}\sum_{\vy\in\cY^n}W\paren{\vy\condon\vx_m,\vs}\indicator{\exists \cL'\in\binom{[M]\setminus\curbrkt{m}}{L},\; \tau_{\vu,\vx_m,\vx_{\cL'},\vs} = \tau_{\bfu,\bfx ,\bfx_{[L]},\bfs}} \notag \\
=&  \frac{1}{M}\sum_{\tau_{\bfu,\bfx ,\bfx_{[L]},\bfs,\bfy}\in\cD_1\cap\cD_2} \sum_{m\in[M]}\indicator{\exists \cL'\in\binom{[M]\setminus\curbrkt{m}}{L},\; \tau_{\vu,\vx_m,\vx_{\cL'},\vs} = \tau_{\bfu,\bfx ,\bfx_{[L]},\bfs}} \notag \\
\le& \frac{1}{M}\sum_{\tau_{\bfu,\bfx ,\bfx_{[L]},\bfs,\bfy}\in\cD_1\cap\cD_2}\card{\curbrkt{m\in[M]\colon \exists \cL'\in\binom{[M]\setminus\curbrkt{m}}{L},\; \tau_{\vu,\vx_m,\vx_{\cL'},\vs} = \tau_{\bfu,\bfx ,\bfx_{[L]},\bfs}}}\notag \\
\dotle& 2^{-n\eps/2}. \label[ineq]{eqn:bound_term31}
\end{align}
\Cref{eqn:bound_term31} follows from \Cref{eqn:property_22}. 

\Cref{eqn:term32} can be bounded as follows.
\begin{align}
&\frac{1}{M}\sum_{m\in[M]\setminus\cM(\vu)} \sum_{\tau_{\bfu,\bfx ,\bfx_{[L]},\bfs,\bfy}\in\cD\cap\cD_1\cap\cD_2^c}\sum_{\vy\in\cY^n}W\paren{\vy\condon\vx_m,\vs} \indicator{ \exists \cL'\in\binom{[M]\setminus\curbrkt{m}}{L},\; \tau_{\vu,\vx_m,\vx_{\cL'},\vs,\vy} = \tau_{\bfu,\bfx ,\bfx_{[L]},\bfs,\bfy}} \notag \\
\le& \frac{1}{M}\sum_{m\in[M]} \sum_{\tau_{\bfu,\bfx ,\bfx_{[L]},\bfs,\bfy}\in\cD\cap\cD_1\cap\cD_2^c}\sum_{\cL'\in\binom{[M]\setminus\curbrkt{m}}{L}\colon \tau_{\vu,\vx_m,\vx_{\cL'},\vs} = \tau_{\bfu,\bfx ,\bfx_{[L]},\bfs}}\sum_{\vy\in\cY^n}W\paren{\vy\condon\vx_m,\vs} \indicator{ \tau_{\vu,\vx_m,\vx_{\cL'},\vs,\vy} = \tau_{\bfu,\bfx ,\bfx_{[L]},\bfs,\bfy}} \notag \\
\le& \frac{1}{M}\sum_{m\in[M]} \sum_{\tau_{\bfu,\bfx ,\bfx_{[L]},\bfs,\bfy}\in\cD\cap\cD_1\cap\cD_2^c} \card{\curbrkt{\cL'\in\binom{[M]\setminus\curbrkt{m}}{L}\colon \tau_{\vu,\vx_m,\vx_{\cL'},\vs} = \tau_{\bfu,\bfx ,\bfx_{[L]},\bfs}}} 2^{-n\kl{P_{\bfu,\bfx,\bfx_{[L]},\bfs,\bfy}}{P_{\bfu,\bfx,\bfx_{[L]},\bfs}W_{\bfy|\bfx,\bfs}}} \label[ineq]{eqn:sanov32} \\
\le& \frac{1}{M}\sum_{m\in[M]}\sum_{\tau_{\bfu,\bfx ,\bfx_{[L]},\bfs,\bfy}\in\cD\cap\cD_1\cap\cD_2^c}\card{\curbrkt{\cL'\in\binom{[M]}{L}\colon \tau_{\vu,\vx_m,\vx_{\cL'},\vs} = \tau_{\bfu,\bfx ,\bfx_{[L]},\bfs}}} 2^{-n\kl{P_{\bfu,\bfx,\bfx_{[L]},\bfs,\bfy}}{P_{\bfu,\bfx,\bfx_{[L]},\bfs}W_{\bfy|\bfx,\bfs}}} \notag \\
\le&\sum_{\tau_{\bfu,\bfx ,\bfx_{[L]},\bfs,\bfy}\in\cD\cap\cD_1\cap\cD_2^c} 2^{n\eps}2^{-nI(\bfy;\bfu,\bfx_{[L]}|\bfx,\bfs)} \label[ineq]{eqn:apply_property_33} \\
=& \sum_{\tau_{\bfu,\bfx ,\bfx_{[L]},\bfs,\bfy}\in\cD\cap\cD_1\cap\cD_2^c}2^{n\eps}2^{-n(I(\bfy;\bfu|\bfx,\bfs) + I(\bfy;\bfx_{[L]}|\bfu,\bfx,\bfs))}\notag \\
\le& \sum_{\tau_{\bfu,\bfx ,\bfx_{[L]},\bfs,\bfy}\in\cD\cap\cD_1\cap\cD_2^c}2^{n\eps}2^{-n(I(\bfx,\bfy;\bfx_{[L]}|\bfs,\bfu) - I(\bfx;\bfx_{[L]}|\bfu,\bfs))}\notag \\
\le& \sum_{\tau_{\bfu,\bfx ,\bfx_{[L]},\bfs,\bfy}\in\cD\cap\cD_1\cap\cD_2^c}2^{n\eps}2^{-n(I(\bfx,\bfy;\bfx_{[L]}|\bfs,\bfu) - (I(\bfx;\bfx_{[L]},\bfs|\bfu) - I(\bfx;\bfs|\bfu)))} \notag \\
\le& \sum_{\tau_{\bfu,\bfx ,\bfx_{[L]},\bfs,\bfy}\in\cD\cap\cD_1\cap\cD_2^c}2^{n\eps}2^{-n(I(\bfx,\bfy;\bfx_{[L]}|\bfs,\bfu) - I(\bfx;\bfx_{[L]},\bfs|\bfu))}\notag \\
\dotle& 2^{n\eps}2^{-n(\eta - \eps)} = 2^{-n(\eta - 2\eps)}. \label[ineq]{eqn:apply_d2} 
\end{align}
\Cref{eqn:sanov32} is by Sanov's theorem (\Cref{lem:sanov}).
\Cref{eqn:apply_property_33} is by \Cref{eqn:property_33} given $ \tau_{\bfu,\bfx,\bfx_{[L]},\bfs,\bfy}\in\cD_1 $ (see \Cref{eqn:def_d1}).
\Cref{eqn:apply_d2} is by the definition of $ \cD $ (\Cref{eqn:def_d}) and $ \cD_2 $ (\Cref{eqn:def_d2}). 

We now turn to the last term \Cref{eqn:term33}. 
For each $ \tau_{\bfu,\bfx,\bfx_{[L]},\bfs,\bfy}\in\cD_1^c $, there must be an $ i_\tau\in[L] $ such that $ R> I(\bfx_{i_\tau};\bfs|\bfu) $. 
Define 
\begin{align}
\cD_3\coloneqq\curbrkt{ \tau_{\bfu,\bfx,\bfx_{[L]},\bfs,\bfy}\in\Delta^{(n)}(\cU\times\cX\times\cX^L \times\cS\times\cY) \colon I(\bfx;\bfx_{i_\tau},\bfs|\bfu) - \sqrbrkt{R - I(\bfx_{i_\tau};\bfs|\bfu)}^+\ge\eps }. \notag 
\end{align}
\Cref{eqn:term33} is at most
\begin{align}
&\frac{1}{M}\sum_{m\in[M]\setminus\cM(\vu)} \sum_{\tau_{\bfu,\bfx ,\bfx_{[L]},\bfs,\bfy}\in\cD\cap\cD_1^c} \sum_{\vy\in\cY^n}W\paren{\vy\condon\vx_m,\vs} \indicator{ \exists \cL'\in\binom{[M]\setminus\curbrkt{m}}{L},\; \tau_{\vu,\vx_m,\vx_{\cL'},\vs,\vy} = \tau_{\bfu,\bfx ,\bfx_{[L]},\bfs,\bfy}} \notag \\
\le& \frac{1}{M}\sum_{m\in[M]\setminus\cM(\vu)} \sum_{\tau_{\bfu,\bfx ,\bfx_{[L]},\bfs,\bfy}\in\cD\cap\cD_1^c} \sum_{\vy\in\cY^n}W\paren{\vy\condon\vx_m,\vs} \indicator{ \exists m'\in[M]\setminus\curbrkt{m},\; \tau_{\vu,\vx_m,\vx_{m'},\vs,\vy} = \tau_{\bfu,\bfx ,\bfx_{i_\tau},\bfs,\bfy}} \notag \\
\le&\frac{1}{M}\sum_{m\in[M]\setminus\cM(\vu)} \sum_{\tau_{\bfu,\bfx ,\bfx_{[L]},\bfs,\bfy}\in\cD\cap\cD_1^c\cap\cD_3} \sum_{\vy\in\cY^n}W\paren{\vy\condon\vx_m,\vs} \indicator{ \exists m'\in[M]\setminus\curbrkt{m},\; \tau_{\vu,\vx_m,\vx_{m'},\vs,\vy} = \tau_{\bfu,\bfx ,\bfx_{i_\tau},\bfs,\bfy}} \label[term]{eqn:term331} \\
&+\frac{1}{M}\sum_{m\in[M]\setminus\cM(\vu)} \sum_{\tau_{\bfu,\bfx ,\bfx_{[L]},\bfs,\bfy}\in\cD\cap\cD_1^c\cap\cD_3^c} \sum_{\vy\in\cY^n}W\paren{\vy\condon\vx_m,\vs} \indicator{ \exists m'\in[M]\setminus\curbrkt{m},\; \tau_{\vu,\vx_m,\vx_{m'},\vs,\vy} = \tau_{\bfu,\bfx ,\bfx_{i_\tau},\bfs,\bfy}}. \label[term]{eqn:term332}
\end{align}
\Cref{eqn:term331} is at most
\begin{align}
&\frac{1}{M}\sum_{m\in[M]\setminus\cM(\vu)} \sum_{\tau_{\bfu,\bfx ,\bfx_{[L]},\bfs,\bfy}\in\cD\cap\cD_1^c\cap\cD_3} \sum_{\vy\in\cY^n}W\paren{\vy\condon\vx_m,\vs} \indicator{ \exists m'\in[M]\setminus\curbrkt{m},\; \tau_{\vu,\vx_m,\vx_{m'},\vs} = \tau_{\bfu,\bfx ,\bfx_{i_\tau},\bfs}}\notag \\
\le&\frac{1}{M}\sum_{\tau_{\bfu,\bfx ,\bfx_{[L]},\bfs,\bfy}\in\cD\cap\cD_1^c\cap\cD_3}\sum_{m\in[M]}\indicator{ \exists m'\in[M]\setminus\curbrkt{m},\; \tau_{\vu,\vx_m,\vx_{m'},\vs} = \tau_{\bfu,\bfx ,\bfx_{i_\tau},\bfs}} \notag \\
=& \frac{1}{M}\sum_{\tau_{\bfu,\bfx ,\bfx_{[L]},\bfs,\bfy}\in\cD\cap\cD_1^c\cap\cD_3}\card{\curbrkt{m\in[M]\colon \exists m'\in[M]\setminus\curbrkt{m},\; \tau_{\vu,\vx_m,\vx_{m'},\vs} = \tau_{\bfu,\bfx ,\bfx_{i_\tau},\bfs}}} \notag \\
\dotle& 2^{-n\eps/2}, \label[ineq]{eqn:apply_property_22}
\end{align}
where \Cref{eqn:apply_property_22}  follows from \Cref{eqn:property_2}. 

\Cref{eqn:term332} is at most 
\begin{align}
&\frac{1}{M}\sum_{m\in[M]\setminus\cM(\vu)} \sum_{\tau_{\bfu,\bfx ,\bfx_{[L]},\bfs,\bfy}\in\cD\cap\cD_1^c\cap\cD_3^c} \sum_{m'\in[M]\setminus\curbrkt{m}\colon \tau_{\vu,\vx_m,\vx_{m'},\vs} = \tau_{\bfu,\bfx,\bfx_{i_\tau},\bfs}} \sum_{\vy\in\cY^n}W\paren{\vy\condon\vx_m,\vs} \indicator{ \tau_{\vu,\vx_m,\vx_{m'},\vs,\vy} = \tau_{\bfu,\bfx ,\bfx_{i_\tau},\bfs,\bfy}} \notag \\
\le& \frac{1}{M}\sum_{m\in[M]} \sum_{\tau_{\bfu,\bfx ,\bfx_{[L]},\bfs,\bfy}\in\cD\cap\cD_1^c\cap\cD_3^c} \sum_{m'\in[M]\colon \tau_{\vu,\vx_m,\vx_{m'},\vs} = \tau_{\bfu,\bfx,\bfx_{i_\tau},\bfs}} 2^{-nI(\bfy;\bfx_{i_\tau}|\bfu,\bfx,\bfs)} \notag \\
=& \frac{1}{M}\sum_{m\in[M]}\sum_{\tau_{\bfu,\bfx ,\bfx_{[L]},\bfs,\bfy}\in\cD\cap\cD_1^c\cap\cD_3^c} \card{\curbrkt{m'\in[M]\colon \tau_{\vu,\vx_m,\vx_{m'},\vs} = \tau_{\bfu,\bfx,\bfx_{i_\tau},\bfs}}} 2^{-nI(\bfy;\bfx_{i_\tau}|\bfu,\bfx,\bfs)} \notag \\
\le& \sum_{\tau_{\bfu,\bfx ,\bfx_{[L]},\bfs,\bfy}\in\cD\cap\cD_1^c\cap\cD_3^c} 2^{n\paren{\sqrbrkt{R - I(\bfx_{i_\tau};\bfx,\bfs|\bfu)}^++\eps}}2^{-nI(\bfy;\bfx_{i_\tau}|\bfu,\bfx,\bfs)} .\label[ineq]{eqn:to_bd_r_minus_i} 
\end{align}
\Cref{eqn:to_bd_r_minus_i} follows from \Cref{eqn:property_3}. 

Since $ \tau_{\bfu,\bfx ,\bfx_{[L]},\bfs,\bfy}\in\cD_1^c\cap\cD_3^c $,
\begin{align}
\eps>&I(\bfx;\bfx_{i_\tau},\bfs|\bfu) - \sqrbrkt{R - I(\bfx_{i_\tau};\bfs|\bfu)}^+ = I(\bfx;\bfx_{i_\tau},\bfs|\bfu) - \paren{R - I(\bfx_{i_\tau};\bfs|\bfu)}. \notag 
\end{align}
Hence
\begin{align}
&&&R> I(\bfx;\bfx_{i_\tau},\bfs|\bfu) + I(\bfx_{i_\tau};\bfs|\bfu) - \eps = I(\bfx;\bfx_{i_\tau}|\bfu,\bfs) + I(\bfx;\bfs|\bfu) + I(\bfx_{i_\tau};\bfs|\bfu)-\eps \ge I(\bfx_{i_\tau};\bfx,\bfs|\bfu)-\eps. \notag \\
&\implies && R - I(\bfx_{i_\tau};\bfx,\bfs|\bfu) + \eps>0. \notag
\end{align}
Consequently,
\begin{align}
\sqrbrkt{R - I(\bfx_{i_\tau};\bfx,\bfs|\bfu)}^+ =& \begin{cases}
0, & R< I(\bfx_{i_\tau};\bfx,\bfs|\bfu) \\
R - I(\bfx_{i_\tau};\bfx,\bfs|\bfu), & R\ge I(\bfx_{i_\tau};\bfx,\bfs|\bfu)
\end{cases} \notag \\
\le& \begin{cases}
R - I(\bfx_{i_\tau};\bfx,\bfs|\bfu) + \eps,& R<I(\bfx_{i_\tau};\bfx,\bfs|\bfu)\\
R - I(\bfx_{i_\tau};\bfx,\bfs|\bfu), & R\ge I(\bfx_{i_\tau};\bfx,\bfs|\bfu) 
\end{cases} \notag \\
\le& R - I(\bfx_{i_\tau};\bfx,\bfs|\bfu) + \eps. \label[ineq]{eqn:r_minus_i}
\end{align}
Substituting \Cref{eqn:r_minus_i} to \Cref{eqn:to_bd_r_minus_i}, we continue the calculation.
\begin{align}
 \sum_{\tau_{\bfu,\bfx ,\bfx_{[L]},\bfs,\bfy}\in\cD\cap\cD_1^c\cap\cD_3^c} 2^{n\paren{ \sqrbrkt{R - I(\bfx_{i_\tau};\bfx,\bfs|\bfu)}^+ + \eps} - nI(\bfy;\bfx_{i_\tau}|\bfu,\bfx,\bfs)} \le& \sum_{\tau_{\bfu,\bfx ,\bfx_{[L]},\bfs,\bfy}\in\cD\cap\cD_1^c\cap\cD_3^c} 2^{n\paren{ R - I(\bfx_{i_\tau};\bfx,\bfs|\bfu) + 2\eps } - nI(\bfy;\bfx_{i_\tau}|\bfu,\bfx,\bfs)} \notag \\
=& \sum_{\tau_{\bfu,\bfx ,\bfx_{[L]},\bfs,\bfy}\in\cD\cap\cD_1^c\cap\cD_3^c} 2^{-nI(\bfx_{i_\tau};\bfx,\bfs,\bfy|\bfu) + nR + 2n\eps} \notag \\
=& \sum_{\tau_{\bfu,\bfx ,\bfx_{[L]},\bfs,\bfy}\in\cD\cap\cD_1^c\cap\cD_3^c} 2^{-n\paren{I(\bfx_{i_\tau};\bfy|\bfu) - R - 2\eps}}
\end{align}

\section{Converse: capacity upper bound}\label{sec:converse_rate_ub}
In this section, we prove \Cref{thm:converse_rate_ub}.

Fix any constant $ \delta>0 $.
Fix any code $ \cC $ of rate $ C_L(\cA_\obli) + \delta $. 
Let $ M\coloneqq |\cC|  $. 
Without loss of generality (\Cref{lem:apx_cc_reduction}), assume that $ \cC $ is $ (\lambda,\wh P_\bfx) $-constant-composition for some constant $\lambda>0$ to be determined later and for some distribution $ P_\bfx\in\Delta(\cX) $. 

By \Cref{lem:apply_turan}, there exists a constant $ \nu>0 $ such that a $ \nu $ fraction of types $ \tau_{\vx_\cL} $ of ordered $L$-lists $ \vx_\cL $, where $ \cL\in\binom{[M]}{L} $, are $ \eps $-close to $ \cp^\tl(\wh P_\bfx) $. 
By quantizing $ \paren{\cp^{\otimes L }(\wh P_\bfx)}_\eps\cap\Delta(\cX^L) $, there is a smaller positive constant $ 0< c \le\nu $ such that a $c$ fraction of types of ordered $L$-lists in $\cC$ are $\eta$-close to $ \wh P_{\bfx_1,\cdots,\bfx_L} $ for some $ \wh P_{\bfx_1,\cdots,\bfx_L}\in\paren{\cp^{\otimes L }(\wh P_\bfx)}_\eps\cap\Delta(\cX^L) $. 
Projecting $ \wh P_{\bfx_1,\cdots,\bfx} $ back to $ \cp^\tl(\wh P_\bfx) $, we finally get  a $c$ fraction of order-$ L $ types in $\cC$ that are $ (\eta+\eps) $-close to $ \wt P_{\bfx_1,\cdots,\bfx_L} $ for some $ \wt P_{\bfx_1,\cdots,\bfx_L}\in\cp^\tl(\wh P_\bfx) $.

\begin{conjecture}\label{conj:comb_conj}
Let $ \cC\subseteq\cX^n $ be any $ (\lambda,\wh P_\bfx) $-constant-composition code of size $ M $ where $ M\ge|\cX|^{nR} $ for some constant $ R\in(0,1] $. 
Let $ L\in\bZ_{\ge1} $. 
We know from the above argument that for any sufficiently small constant $ \eps>0 $, there must exist a  constant $ c>0 $ and a distribution $ \wt P_{\bfx_1,\cdots,\bfx_L}\in\cp^\tl(\wh P_\bfx) $ such that a $c$ fraction of $ \tau_{\vx_\cL} $ ($ \cL\in\binom{[M]}{L} $) satisfy $ d\paren{\tau_{\vx_\cL},\wt P_{\bfx_1,\cdots,\bfx_L}}\le\eps $. 
Assume that the $\cp$-distribution can be decomposed as $ \wt P_{\bfx_1,\cdots,\bfx_L} = \sqrbrkt{\wt P_\bfu\wt P_{\bfx|\bfu}}_{\bfx_1,\cdots,\bfx_L} $ for some $ k\in\bZ_{\ge1} $,  $ P_\bfu\in\Delta(\cU) $ ($ \cU = [k] $) and $ P_{\bfx|\bfu}\in\Delta(\cX|\cU) $. 
It is conjectured that for any sufficiently small constant $ \gamma>0 $, there exist a constant $ \mu>0 $, a subcode $ \cC'\subset\cC $ of size $ |\cC'| = M'\ge\mu M $ and a \emph{universal time-sharing sequence} $ \vu\in\cU^n $ for $ \cC' $ such that $ d\paren{\tau_{\vu,\vx}, \wt P_\bfu \wt P_{\bfx|\bfu}}\le\gamma $ for every $ \vx\in\cC' $. 
\end{conjecture}

\subsection{James's jamming strategy}
\label{sec:rate_ub_strategy}
Let $ \cC = \curbrkt{\vx_i}_{i \in[M]} $ be a $ (\lambda,\wh P_\bfx) $-constant-composition code of size $M$. 
Let $ \delta_L,\gamma $ be positive constants to be determined later.
\begin{enumerate}
	\item Construct a $\delta_L $-net $\cN_L $ of $ \cp^\tl(\wh P_\bfx) $. 
	\item Draw all types $ \tau_{\vx_\cL} $ ($ \cL\in\binom{[M]}{L} $) on the space $ \Delta(\cX^L) $.
	(Note that the types may not fall within $ \cp^\tl(\wh P_\bfx) $.)
	Find the Voronoi cell induced by $\cN_L $ that contains the largest number of types of ordered $L$-tuples in $\cC$. 
	Let $ \wt P_{\bfx_1,\cdots,\bfx_L} $ denote the representative of that cell.
	Let $ (\wt P_\bfu,\wt P_{\bfx|\bfu}) $ be  a $\cp$-decomposition of $ \wt P_{\bfx_1,\cdots,\bfx_L} $. 
	Let $ \cC'\subseteq\cC $ denote the set of all codewords in the Voronoi cell corresponding to $ \wt P_{\bfx_1,\cdots,\bfx_L} $. 
	\item Find the largest subsubcode $ \cC''\subseteq\cC' $ such that there is a sequence $ \vu\in\cU^n $ satisfying $ d\paren{\tau_{\vu,\vx}, \wt P_\bfu,\wt P_{\bfx|\bfu}}\le\gamma $ for all $ \vx\in\cC'' $. 
	\item Sample $ \vbfs\in\cS^n $ according to the following distribution
	\begin{align}
	\vbfs\sim\prod_{j = 1}^n U_{\bfs|\bfu = \vu(j)} \eqqcolon U_{\vbfs|\vu}. \label{eqn:converse_distr} 
	\end{align}
\end{enumerate}

\subsection{Analysis of average error probability}
\label{sec:rate_ub_analysis}
Let $ \psi\colon\cY^n\to[M] $ be the $L$-list-decoder equipped with $\cC$.
By \Cref{def:enc_list_dec}, $ |\psi(\vy)|\le L $ for any $ \vy\in\cY^n $. 
Using \Cref{def:avg_error_prob}, we compute the expected average error probability over the random  jamming sequence generation defined by \Cref{eqn:converse_distr}.
\begin{align}
\exptover{\vbfs\sim U_{\vbfs|\vu}}{P_{\e,\avg}(\vbfs)} =& \frac{1}{M}\sum_{i\in[M]}\exptover{\vbfs\sim U_{\vbfs|\vu}}{P_{\e}(i,\vbfs)} \notag \\
=& \frac{1}{M}\sum_{i\in[M]}\sum_{\vy\in\cY^n\colon \psi(\vy)\not\ni i}\exptover{\vbfs\sim U_{\vbfs|\vu}} {W^{\tn}_{\bfy|\bfx,\bfs}(\vy|\vx_i,\vbfs)} \notag \\
=& \frac{1}{M}\sum_{i\in[M]}\sum_{\vy\in\cY^n\colon \psi(\vy)\not\ni i}\prod_{j = 1}^n\exptover{\vbfs(j)\sim U_{\bfs|\bfu = \vu(j)}}{W_{\bfy|\bfx,\bfs}(\vy(j)|\vx_i(j),\vbfs(j))}. \label{eqn:avg_error_prob_original}
\end{align}
We observe that \Cref{eqn:avg_error_prob_original} can be viewed as the average error probability of a \emph{time-varying} channel $ \wt W_{\bfy|\bfx,\bfu} $ defined according to
\begin{align}
\wt W_{\bfy|\bfx,\bfu}(y|x,u) \coloneqq& \exptover{\bfs\sim U_{\bfs|\bfu = u}}{W_{\bfy|\bfx,\bfs}(y|x,s)} = \sum_{s\in\cS}U_{\bfs|\bfu}(s|u)W_{\bfy|\bfx,\bfs}(y|x,s), \notag 
\end{align}
for all $ y\in\cY,x\in\cX,u\in\cU $. 
Under this definition, \Cref{eqn:avg_error_prob_original} can be written as
\begin{align}
\Cref{eqn:avg_error_prob_original} =& \frac{1}{M}\sum_{i\in[M]}\sum_{\vy\in\cY^n\colon \psi(\vy)\not\ni i}\prod_{j = 1}^n \wt W_{\bfy|\bfx,\bfu}(\vy(j)|\vx_i(j),\vu(j)) \notag \\
=& \frac{1}{M}\sum_{i\in[M]}\sum_{\vy\in\cY^n\colon \psi(\vy)\not\ni i}\wt W_{\bfy|\bfx,\bfu}^\tn(\vy|\vx_i,\vu) = \wt P_{\e,\avg}(\cC), \label{eqn:to_be_massaged} 
\end{align}
where $ \wt P_{\e,\avg} $ is the average error probability of $ \cC $ used over $ \wt W_{\bfy|\bfx,\bfu}^\tn $. 
To apply strong converse for fading channels with approximate constant-composition codes and list-decoding, we  massage \Cref{eqn:to_be_massaged} as follows.
\begin{align}
\Cref{eqn:to_be_massaged} \ge& \frac{1}{M}\sum_{i\in[M]\colon \vx_i\in\cC''}\sum_{\vy\in\cY^n\colon\psi(\vy)\not\ni i}\wt W_{\bfy|\bfx,\bfu}^\tn(\vy|\vx_i,\vu) \notag \\
\ge& \frac{1}{M}\sum_{i\in[M]\colon \vx_i\in\cC''}\sum_{\vy\in\cY^n\colon\psi_\opt''(\vy)\not\ni i}\wt W_{\bfy|\bfx,\bfu}^\tn(\vy|\vx_i,\vu) \label{eqn:opt_list_dec_subsubcode} \\
=& \frac{M''}{M}\frac{1}{M''}\sum_{i\in[M]\colon\vx_i\in\cC''}\sum_{\vy\in\cY^n\colon \psi_\opt''(\vy)\not\ni i}\wt W_{\bfy|\bfx,\bfu}^\tn(\vy|\vx_i,\vu), \label{eqn:notation_size_subsubcode}
\end{align}
where in \Cref{eqn:notation_size_subsubcode}, $ |\cC''| $ is denoted by $M''$ and in \Cref{eqn:opt_list_dec_subsubcode} we use $ \psi_\opt''\colon\cY^n\to\binom{[M'']}{\le L} $ to denote the optimal $L$-list-decoder for $ \cC'' $ used over $ \wt W^\tn_{\bfy|\bfx,\bfu} $. 

By \Cref{eqn:cp_tuple_extraction}, $ M'/M\ge c $.
By \Cref{conj:comb_conj}, $ M''/M'\ge \mu $. 
Therefore we have
\begin{align}
\frac{M''}{M} =& \frac{M'}{M}\frac{M''}{M'}\ge c\mu>0, \label{eqn:subcode_bound} 
\end{align}
which is a positive constant. 
Note that
\begin{align}
\wt P_{\e,\avg}(\cC'') \coloneqq& \frac{1}{M''}\sum_{i\in[M]\colon\vx_i\in\cC''}\sum_{\vy\in\cY^n\colon \psi_\opt''(\vy)\not\ni i}\wt W_{\bfy|\bfx,\bfu}^\tn(\vy|\vx_i,\vu) \notag
\end{align}
is the average error probability of $ \cC'' $ used over $ \wt W_{\bfy|\bfx,\bfu}^\tn $ under the optimal $L$-list-decoder. 
We are going to apply the strong converse bound on $ \wt P_{\e,\avg}(\cC'') $.
To this end, first recall the capacity of fading channels. 
\begin{theorem}[Capacity of fading DMCs, Chap.  23 \cite{elgamal-kim}]\label{thm:cap_fading}
Let $ W_{\bfy|\bfx,\bfu} $ be a fading channel with fast block fading according to a fixed sequence $ \vu\in\cU^n $ (where $ |\cU| $ is a constant) of type $P_\bfu $.\footnote{That is, the fading process $ \curbrkt{ \bfu_i}_{i\in[n]} $ is a constant $ u\in\cU $  in each coherence time interval $ \curbrkt{i\in[n]\colon \vu(i) = u} $. Since $ |\cU|\ll n $ is a constant, this model is called \emph{fast fading}. }
The capacity of  $ W_{\bfy|\bfx,\bfu} $  is given by
\begin{align}
C(W_{\bfy|\bfx,\bfu})=& \max_{P_{\bfx|\bfu}\in\Delta(\cX|\cU)} I(\bfx;\bfy|\bfu), \label{eqn:cap_fading}
\end{align}
where the mutual information is evaluated according to $ P_\bfu P_{\bfx|\bfu}W_{\bfy|\bfx,\bfu} $. 
\end{theorem}

The strong converse result for fading channels under list-decoding with approximate constant-composition codes follows from standard techniques.
We state the result below and prove it in \Cref{app:strong_conv_fading_pf}.
\begin{theorem}[Strong converse for fading DMCs with approximate constant-composition codes and list-decoding]\label{thm:strong_conv_fading}
Let $ W_{\bfy|\bfx,\bfu} $ be a fading DMC with fast fading according to a sequence $ \vu\in\cU^n $ of type $ P_\bfu\in\Delta(\cU) $ where $ \cardU $ is a constant. 
Fix any constants $ \delta>0 $ and $ 0<\lambda\ll\delta $.
Let  $ P_{\bfx|\bfu}\in\Delta(\cX|\cU) $. 
For any  list-size $ L\in\bZ_{\ge1} $ and any $(\lambda,\vu, P_{\bfx|\bfu})$-constant-composition code $\cC$ of rate $ R = \frac{1}{n}\log \frac{|\cC|}{L} \ge C(W_{\bfy|\bfx,\bfu})+\delta $ (where $ C(W_{\bfy|\bfx,\bfu}) $ is defined in \Cref{eqn:cap_fading}), the average error probability of $ \cC $ used over $ W_{\bfy|\bfx,\bfu}^\tn $ is approaching 1. 
More precisely 
\begin{align}
P_{\e,\avg}(\cC) \ge 1- 2^{-nf(\delta,\lambda)} , \label{eqn:strong_conv_fading_bound}
\end{align}
for some function $ f(\delta,\lambda)>0 $ satisfying $ f(\delta,\lambda)\xrightarrow{\delta,\lambda\to0}0 $. 
\end{theorem}
By \Cref{thm:strong_conv_fading},  
\begin{align}
\wt P_{\e,\avg}(\cC'') \ge1 - 2^{-nf(\delta,\lambda)}. \label{eqn:pe_fading_bound}
\end{align}

Finally, plugging \Cref{eqn:subcode_bound} and \Cref{eqn:pe_fading_bound} into \Cref{eqn:notation_size_subsubcode}, we get
\begin{align}
\exptover{\vbfs\sim U_{\vbfs|\vu}}{P_{\e,\avg}(\vbfs)} \ge& c\mu(1-2^{-nf(\delta,\lambda)})\xrightarrow{n\to\infty} c\mu >0, \notag
\end{align}
which completes the error analysis.

\subsection{A subcode extraction procedure towards \Cref{conj:comb_conj}}
\label{sec:subcode_extraction_towards_conj}
Towards a tight capacity characterization, we propose a natural subcode extraction procedure. 
However, we do not know how to {prove} that it gives rise to a matching capacity upper bound.
The procedure is described below.
Given any $ (\lambda,P_\bfx) $-constant-composition code $ \cC\subseteq\cX^n $ of size $ M $. 
We extract a subcode $ \cC'\subseteq\cC $ via the following steps. 
Let $ \lambda'>0 $ be a  constant to be determined. 
\begin{enumerate}
	\item Quantize $ \Delta(\cX) $ using a $\zeta$-net $ \cN = \curbrkt{P_{\bfx_i}}_{i\in[N]} $ (where $ P_{\bfx_i}\in\Delta(\cX) $ for each $ i $) of size $ N $ at most $ N(\zeta,\cardX) $ (by \Cref{lem:quant}). 
	\item Cluster the \emph{columns} of $ \cC\in\cX^{M\times n} $ according their types.
	This naturally induces a time-sharing sequence $ \vu\in\cU^n $ defined as follows. 
	The alphabet $ \cU $ has size at most $ N $. 
	The $j$-th ($j\in[n]$) component $ \vu(j) $ is set $ i\in[N] $ if the $j$-th column of $ \cC $ has type $ \zeta $-close to $ P_{\bfx_i}\in\cN $. 
	\item Find the largest subcode $ \cC'\subseteq\cC $  with the following properties. 
	\begin{enumerate}
		\item\label{itm:proposed_subcode_cond1}
		The subcode $ \cC' $ has size lat least $ \theta M $ for some constant $ \theta>0 $. 
		\item\label{itm:proposed_subcode_cond2} For each $ i\in\cU $, the shorter subcode $ \cC'|_i\subseteq\cX^{n\tau_\vu(i)} $ is  $ (\lambda', P_{\bfx_i}') $-constant-composition for some  distribution $ P_{\bfx_i}'\in\Delta(\cX) $\footnote{Unless other operation is done otherwise, for each $ i\in\cU $, the \emph{column} composition $ P_{\bfx_i} $ may not be preserved in each punctured subcode $ \cC'|_i $. 
		This is why we only claim $ (\lambda',P_{\bfx_i}') $-constant-composition for another distribution $ P_{\bfx_i}' $ that can be different from $ P_{\bfx_i} $ a priori.}.
		Here  $ \cC'|_i\subseteq\cX^{ n\tau_\vu(i) } $ denotes the puncturing of $ \cC' $ onto the components $j$'s ($j\in[n]$) such that $ \vu(j) = i $. 
		\item The subcode $ \cC' $ is the one that minimizes the induced conditional  mutual information 
		$$ \max_{U_{\bfs|\bfu}\in\Delta(\cS|\cU)\colon \sqrbrkt{P_\bfu U_{\bfs|\bfu}}_\bfs\in\lambda_\bfs} I(\bfx;\bfy|\bfu) $$
		among all subcodes satisfying \Cref{itm:proposed_subcode_cond1,itm:proposed_subcode_cond2},
		where the conditional mutual information  is evaluated according to $ P_{\bfu,\bfx,\bfy}(u,x,y) = \tau_\vu(u) P'_{\bfx_u}(x) U_{\bfs|\bfx}(s|x)W_{\bfy|\bfx,\bfs}(y|x,s)  $. 
	\end{enumerate}
\end{enumerate}
We claim that subcodes satisfying \Cref{itm:proposed_subcode_cond1,itm:proposed_subcode_cond2} do exist for sufficiently small $ \theta $. 
Indeed, they can be constructed as follows. 
By approximate-constant-composition reduction (\Cref{lem:apx_cc_reduction}), for $ i = 1 $, we can find a  subcode $ \cC'_1 \subseteq\cC $ of size $ M_1 \ge \theta_0 M $ where $ \theta_0 \ge \paren{\frac{\cardX}{2\lambda'} + 1}^{-\cardX} $ such that $ \cC'_1|_1 \subseteq\cX^{n\tau_\vu(1)} $ is $ (\lambda',P'_{\bfx_1}) $-constant-composition for some $ P_{\bfx_1}'\in\Delta(\cX) $. 
For $ i = 2 $, there  exists a subcode $ \cC_2' \subseteq\cC_1' $ of size $ M_2 \ge \theta_0M_1 \ge\theta_0^2M $ such that $ \cC_2'|_2\subseteq\cX^{n\tau_\vu(2)} $ is $ (\lambda',P_{\bfx_2}') $-constant-composition for some $ P_{\bfx_2}'\in\Delta(\cX) $. 
Note that $ \cC_2'|_1 \subseteq\cX^{n\tau_\vu(1)} $ is still $ (\lambda',P_{\bfx_1}') $-constant-composition since $ \cC_2'\subseteq\cC_1' $. 
Iteratively doing this for each $ i = 1,2,\cdots,\cardU $, we end up with a subcode $ \cC' = \cC_\cardU' $ of size at least $ M_\cardU \ge \theta_0^\cardU M $ such that for each $ i\in\cU $, $ \cC'|_i $ is $ (\lambda',P_{\bfx_i}') $-constant-composition.
See \Cref{fig:proposed_subcode_extraction}. 
\begin{figure}[htbp]
	\centering
	\includegraphics[width=0.45\textwidth]{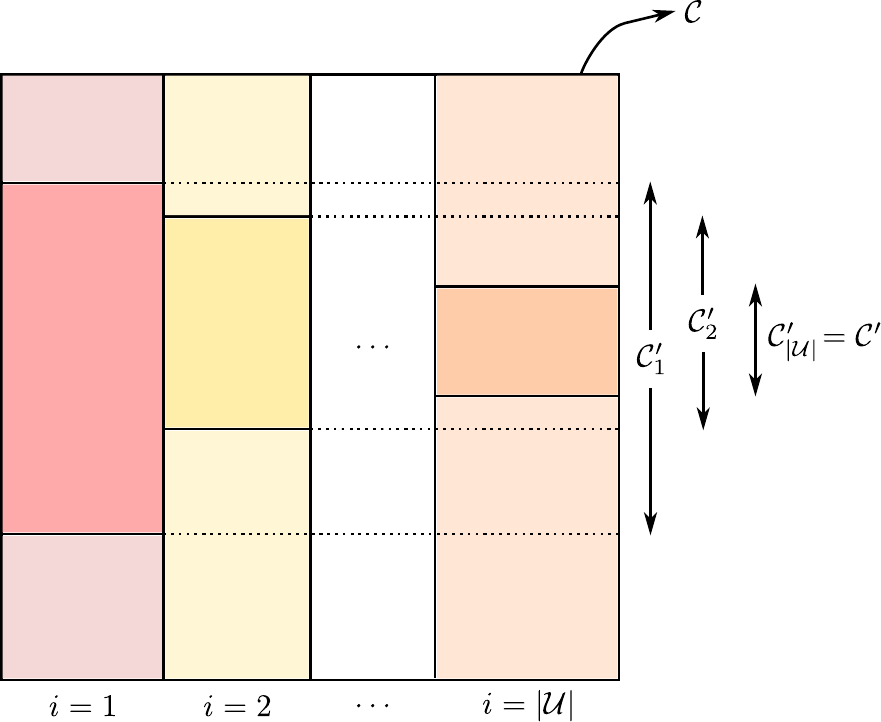}
	\caption{A proposed procedure for extracting large subcodes that is chunk-wise constant-composition w.r.t. a time-sharing sequence $ \vu\in\cU^n $.}
	\label{fig:proposed_subcode_extraction}
\end{figure}
Therefore, as long as $ \theta\le \paren{\frac{\cardX}{2\lambda'} + 1}^{-\cardX N} \le \theta_0^\cardU  $, subcodes satisfying \Cref{itm:proposed_subcode_cond1,itm:proposed_subcode_cond2} can be found. 
The goal is to find the one with the minimum conditional mutual information.

By construction, a subcode $ \cC' $ obtained as above is indeed constant-composition w.r.t. a \emph{universal} time-sharing sequence $\vu$ as desired in \Cref{conj:comb_conj}. 
It is not hard to see that if a code $ \cC $ was truly sampled from the ensemble defined in \Cref{app:cw_select}, then the above procedure will with high probability detect the correct underlying time-sharing sequence $ \vu $ and the distributions $ \curbrkt{P_{\bfx_u}}_{u\in\cU} $ that $ \cC $ was sampled from.
However, one crucial gap in the above construction that prevents us from {proving} a matching capacity upper bound is as follows. 
Recall that in the capacity expression (\Cref{eqn:cap_obli_avc_list_dec}), the maximization is over \emph{non-$L$-symmetrizable} $\cp$-distributions $ P_{\bfx_{[L]}} = \sum_{u\in\cU}P_\bfu(u)P_{\bfx_u}^\tl $. 
Unfortunately, we do not know how to show the non-$L$-symmetrizability of the $\cp$-distribution $ \sum_{u\in\cU} \tau_\vu(u)P_{\bfx_u}'^\tl $ induced by the above construction.

\section{Technical remarks}\label{sec:technical_rmk}
In the first two subsections (\Cref{sec:myop_obli_symm,sec:myop_omni_symm}) of this section, 
we discuss different notions of symmetrizability for different channel models and their relations. 
See \Cref{fig:symm_notioni} for an overview. 
\begin{figure}[htbp]
	\centering
	\includegraphics[width=\textwidth]{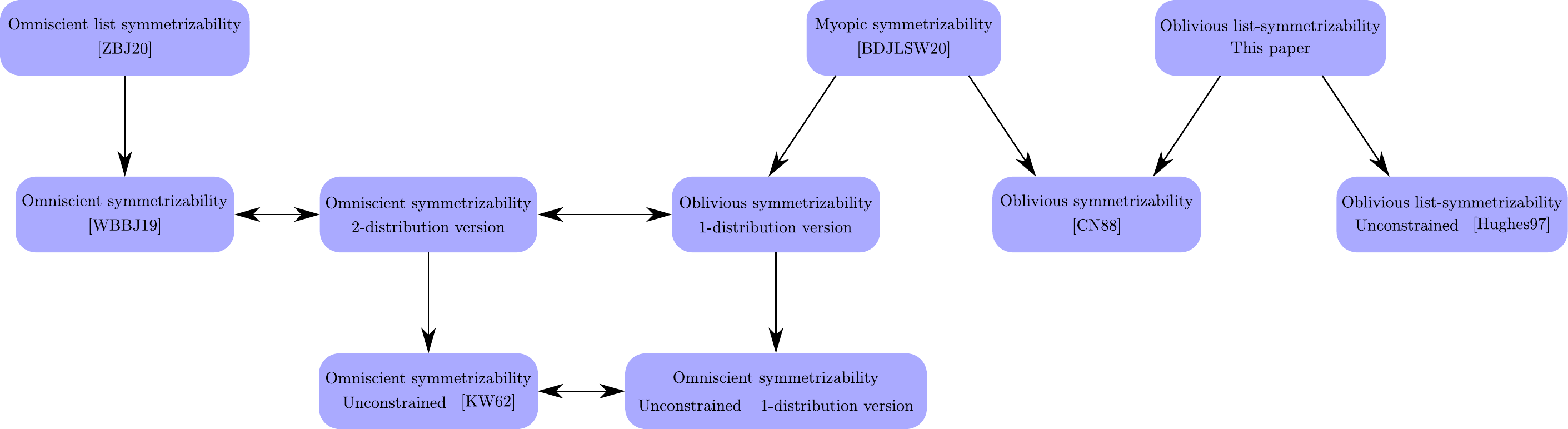}
	\caption{Different notions of symmetrizability for different channel models. If there is an arrow $ \rightarrow $ from Definition $A$ to Definition $B$, i.e., $ A\to B $, then Definition $B$ is a special case of Definition $A$. If the arrow is bidirectional, it means that the definitions on two endpoints of the arrow are equivalent.}
	\label{fig:symm_notioni}
\end{figure}
In the last subsection (\Cref{sec:seemingly_natural}), we discuss a seemingly more natural jamming strategy mentioned in \Cref{rk:seemingly_natural} for obtaining capacity upper bounds. 

\subsection{Myopic/oblivious symmetrizability}
\label{sec:myop_obli_symm}
By the unique-decoding version of the generalized Plotkin bound \cite{wbbj-2019-gen_plotkin}, in any sufficiently large approximately constant-composition code, there is a constant $ c>0 $ such that a $ c $ fraction of types $ \tau_{\vx_i,\vx_j} $ of ordered pairs $ i<j $ of codewords is approximately $\cp$. 
Curious readers may wonder why \csiszar--Narayan in their seminal paper \cite{csiszar-narayan-it1988-obliviousavc} did \emph{not} resort to the heavy  machinery in  \cite{wbbj-2019-gen_plotkin} to prove a tight characterization of the capacity of input-constrained state-constrained oblivious AVCs. 
Before explaining the reason, we would like to first take a diversion and discuss different notions of symmetrizability/confusability in the AVC/adversarial channel literature.

Let us introduce \emph{myopic AVCs} and \emph{omniscient AVCs}.
In oblivious AVCs, the adversary does not know which particular codeword from the codebook was transmitted, though he does know the codebook (which is known to everyone).
Whereas in omniscient AVCs, the adversary knows precisely which codeword was transmitted.
\begin{definition}[Omniscient AVCs]
\label{def:omni_avc}
An omniscient AVC $ \cA_\omni = (\cX,\cS,\cY,\lambda_\bfx,\lambda_\bfs,W_{\bfy|\bfx,\bfs}) $ consists of three alphabets $ \cX,\cS,\cY $ for the input, jamming and output sequences, respectively; 
input constraints $ \lambda_\bfx\subseteq\Delta(\cX) $ and state constraints $ \lambda_\bfs\subseteq\Delta(\cS) $; 
an adversarial channel $ W_{\bfy|\bfx,\bfs} $ from Alice to Bob governed by James.
Knowing the codebook $ \cC $, Alice's encoder $ \phi $, Bob's decoder $\psi$, receiving Alice's transmitted sequence $ \vbfx $, James generates a jamming sequence $ \vbfs $ and sends it through the channel $ W_{\bfy|\bfx,\bfs} $. 
See \Cref{fig:diag_omni_listdec} for a block diagram of the omniscient AVC model under list-decoding.
WHen $L=1$, it degenerates to the unique-decoding case.
\end{definition}
\begin{figure}[htbp]
	\centering
	\includegraphics[width=0.7\textwidth]{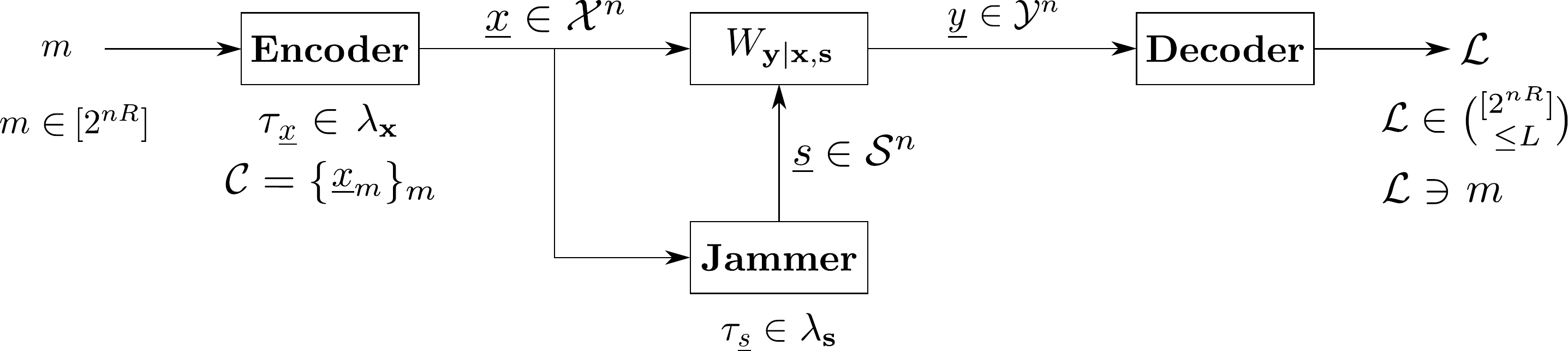}
	\caption{Block diagram of an omniscient AVC under list-decoding.}
	\label{fig:diag_omni_listdec}
\end{figure}

The myopic AVCs are a natural information-theoretic bridge between the oblivious and omniscient models.
In the myopic setting, the adversary observes a \emph{noisy} version of the transmitted sequence through a \emph{DMC}.
\begin{definition}[Myopic AVCs]
\label{def:myop_avc}
A myopic AVC $ \cA_\myop = (\cX,\cS,\cY,\lambda_\bfx,\lambda_\bfs,W_{\bfz|\bfx},W_{\bfy|\bfx,\bfs}) $ consists of three alphabets $ \cX,\cS,\cY $ for the input, jamming and output sequences, respectively; 
input constraints $ \lambda_\bfx\subseteq\Delta(\cX) $ and state constraints $ \lambda_\bfs\subseteq\Delta(\cS) $; 
a DMC $ W_{\bfz|\bfx} $ from Alice to James;
an adversarial channel $ W_{\bfy|\bfx,\bfs} $ from Alice to Bob governed by James.
Knowing the codebook $ \cC $, Alice's encoder $ \phi $, Bob's decoder $ \psi $, receiving a noisy version $ \vbfz $ of Alice's transmitted sequence $ \vbfx $, James  generates a jamming sequence $ \vbfs $ and sends it through the channel $ W_{\bfy|\bfx,\bfs} $. 
See \Cref{fig:diag_myop} for a block diagram of the myopic AVC model. 
\end{definition}
\begin{figure}[htbp]
	\centering
	\includegraphics[width=0.7\textwidth]{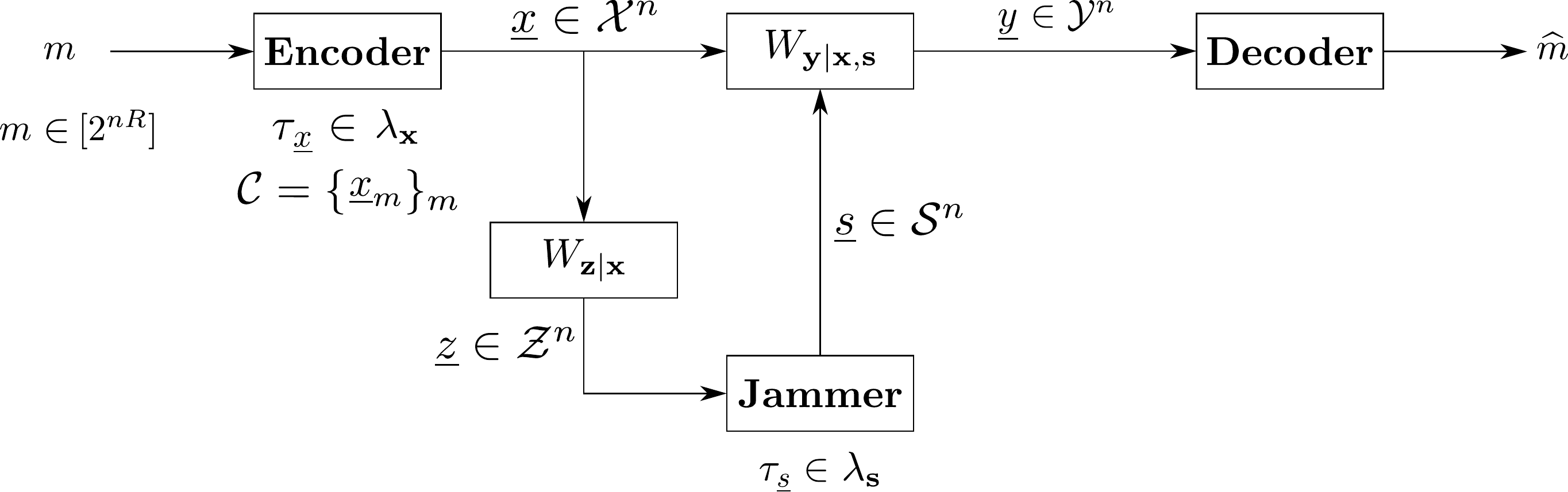}
	\caption{Block diagram of a myopic AVC.}
	\label{fig:diag_myop}
\end{figure}

In a recent work \cite{bdjlsw-2020-myopic_symm}, a new notion of \emph{myopic symmetrizability} was introduced.
\begin{definition}[Myopically symmetrizing distributions]
\label{def:myop_symm_distr_set}
Define the set of \emph{myopically symmetrizing distributions} as 
\begin{align}
\cU_{\myop,\symm}\coloneqq& \curbrkt{ U_{\bfs|\bfz,\bfx'}\in\Delta(\cS|\cZ\times\cX)\colon 
\begin{array}{rl}
&\forall (x,x',y)\in\cX^2\times\cY,\\
&\displaystyle\sum_{z\in\cZ,s\in\cS} W_{\bfz|\bfx}(z|x)U_{\bfs|\bfz,\bfx'}(s|z,x')W_{\bfy|\bfx,\bfs}(y|x,s)\\
=&\displaystyle\sum_{z\in\cZ,s\in\cS} W_{\bfz|\bfx}(z|x')U_{\bfs|\bfz,\bfx'}(s|z,x)W_{\bfy|\bfx,\bfs}(y|x',s)
\end{array}
 }. \notag
\end{align}
\end{definition}
\begin{definition}[Myopic symmetrizability, \cite{bdjlsw-2020-myopic_symm}]
\label{def:myop_symm}
Consider a myopic AVC $ \cA_\myop $. 
An input distribution $ P_\bfx\in\lambda_\bfx $ is called \emph{myopically symmetrizable} if for all $ P_{\bfx,\bfx'}\in\cp(P_\bfx) $ there is a $ U_{\bfs|\bfz,\bfx'}\in\cU_{\myop,\symm} $ such that $ \sqrbrkt{ P_{\bfx,\bfx'} W_{\bfz|\bfx}U_{\bfs|\bfz,\bfx'} }_\bfs\in\lambda_\bfs $. 
\end{definition}

It was shown in \cite{bdjlsw-2020-myopic_symm} that myopic symmetrizability is a sufficient (though not known to be necessary) condition for a myopic AVC to have zero rate.
\begin{theorem}[Capacity positivity of myopic AVCs, , \cite{bdjlsw-2020-myopic_symm}]
\label{thm:myop_positivity}
For any myopic AVC $ \cA_\myop $, the capacity is zero if every $ P_\bfx\in\lambda_\bfx $ is myopically symmetrizable. 
\end{theorem}

We claim that \csiszar--Narayan's notion of oblivious symmetrizability (for unique-decoding) is a special case of \Cref{def:myop_symm}.
Indeed, the myopic model collapses to the oblivious model if $ \bfz \indep \bfx $. 
That is, the channel $ W_{\bfz|\bfx} $ from Alice to James is infinitely noisy: $ W_{\bfz|\bfx} = W_\bfz $. 
Taking $ U_{\bfs|\bfx'}(s|x') \coloneqq \sum_{z\in\cZ}W_\bfz(z)U_{\bfs|\bfz,\bfx'}(s|z,x') $, the definition of $ \cU_{\myop,\symm} $ becomes:
\begin{definition}[Obliviously symmetrizing distributions]
\label{def:obli_symm_distr_set}
The set of \emph{obliviously symmetrizing distributions} is defined as
\begin{align}
\cU_{\obli,\symm}\coloneqq&\curbrkt{ U_{\bfs|\bfx'}\in\Delta(\cS|\cX)\colon 
\forall (x,x',y)\in\cX^2\times\cY, \;
 \sum_{s\in\cS}U_{\bfs|\bfx'}(s|x') W_{\bfy|\bfx,\bfs}(y|x,s) = \sum_{s\in\cS} U_{\bfs|\bfx}(s|x) W_{\bfy|\bfx,\bfs}(y|x',s)
}. \notag 
\end{align}
\end{definition}

Note also that 
\begin{align}
\sqrbrkt{P_{\bfx,\bfx'} W_\bfz U_{\bfs|\bfz,\bfx'} }_\bfs(s) =& \sum_{(x,x')\in\cX^2} \sum_{z\in\cZ} P_{\bfx,\bfx'}(x,x') W_\bfz(z) U_{\bfs|\bfz,\bfx'}(s|z,x) \notag \\
=& \sum_{(x,x')\in\cX^2} P_{\bfx,\bfx'} (x,x') U_{\bfs|\bfx'} (s|x) \notag \\
=& \sum_{x\in\cX} P_{\bfx'} (x) U_{\bfs|\bfx'}(s|x) \notag \\
=& \sqrbrkt{ P_{\bfx} U_{\bfs|\bfx'} }_{\bfs}(s), \label{eqn:obli_symm_selfcoupl}
\end{align}
where \Cref{eqn:obli_symm_selfcoupl} follows since $ \sqrbrkt{P_{\bfx,\bfx'}}_\bfx = \sqrbrkt{P_{\bfx,\bfx'}}_{\bfx'} = P_\bfx $.

Therefore, we get the following notion of \emph{oblivious symmetrizability} which precisely matches the one in \cite{csiszar-narayan-it1988-obliviousavc}.
\begin{definition}[Oblivious symmetrizability, \cite{csiszar-narayan-it1988-obliviousavc}]
\label{def:cn_obli_symm}
Given an oblivious AVC $ \cA_\obli $, an input distribution $ P_\bfx\in\lambda_\bfx $ is called \emph{obliviously symmetrizable} if there is a $ U_{\bfs|\bfx'}\in\cU_{\obli,\symm} $ such that $ \sqrbrkt{P_\bfx U_{\bfs|\bfx'}}_\bfs\in\lambda_\bfs $. 
\end{definition}
It was shown in \cite{csiszar-narayan-it1988-obliviousavc} that:
\begin{theorem}[Capacity positivity of oblivious AVCs, \cite{csiszar-narayan-it1988-obliviousavc}]
\label{thm:oblivious_positivity}
An oblivious AVC $ \cA_{\obli} $ has zero capacity  if and only if every $ P_\bfx\in\lambda_\bfx $ is obliviously symmetrizable. 
\end{theorem}

We are now ready to answer the question at the beginning of this section.
Due to the oblivious nature,  the joint distribution $ P_{\bfx,\bfx'} $ is marginalized when we evaluate the jamming distribution $ P_\bfs $.\footnote{Note that the evaluation of $ P_\bfs $ is the only place in the definition of symmetrizability where $ P_{\bfx,\bfx'} $ plays a role. The definition of $ \cU_{\symm} $ has nothing to do with $ P_{\bfx,\bfx'} $.} 
As a result, only the marginal input distribution $ P_\bfx $ rather than the completely positive joint distribution $ P_{\bfx,\bfx'} $ matters in the definition of oblivious symmetrizability.
It is as if Alice's transmitted signal $ \bfx $ and James's spoofing signal $ \bfx' $ are independent rather than jointly $ \cp $\footnote{Of course, product distributions are also $\cp$. The point here is that requiring \emph{generic} $\cp$ correlation turns out to be unnecessary.}, i.e.,  $ P_{\bfx,\bfx'} = P_\bfx P_\bfx^\top $. 
Indeed, it is \emph{provably} known \cite{csiszar-narayan-it1988-obliviousavc} that the optimal symmetrization strategy samples a spoofing codeword $ \vbfx' $ uniformly from the codebook, independent of the transmitted codeword $ \vbfx $.

Similarly, for oblivious AVCs under $L$-list-decoding, the $L$-oblivious symmetrizability only depends on the size-$L$ marginal distribution $ P_{\bfx_1,\cdots,\bfx_L}\in\cp^\tl(P_\bfx) $ of a spoofing $L$-list,  rather than the full joint distribution $ P_{\bfx,\bfx_1,\cdots,\bfx_L} $ of the transmitted signal and the spoofing list. 
It is easy to see that $1$-oblivious symmetrizability (\Cref{def:cp_symm}) is equivalent to \csiszar--Narayan's notion of oblivious symmetrizability (\Cref{def:cn_obli_symm}). 

\subsection{Myopic/omniscient symmetrizability}
\label{sec:myop_omni_symm}
In this section, we  discuss the relation between  symmetrizability for myopic channels and  \emph{confusability} for \emph{omniscient} channels. 

The omniscient channel, as mentioned in the previous section (\Cref{sec:myop_obli_symm}), is a more coding-theoretic model of communication. 
In the absence of  constraints, i.e., $ \lambda_\bfx = \Delta(\cX), \lambda_\bfs = \Delta\cS $, the right notion of {confusability}\footnote{Due to historical reasons, the condition for zero capacity is referred to as confusability. We follow this convention in zero-error information theory. However, in essence, one can view it as a notion of symmetrizability for omniscient channels. }
was proposed by Kiefer--Wolfowitz \cite{kiefer-wolfowitz-1962-omniscient_confusability_unconstrained}. 
\begin{definition}[Unconstrained omniscient symmetrizability,  \cite{kiefer-wolfowitz-1962-omniscient_confusability_unconstrained}]
\label{def:omni_symm_unconstr}
An unconstrained omniscient AVC $$\cA_\omni = (\cX,\cS,\cY,\Delta(\cX),\Delta(\cS), W_{\bfy|\bfx,\bfs})$$ is called \emph{omnisciently symmetrizable} if for all $ (x,x')\in\cX^2 $, there are distributions $ U_{\bfs|x,x'},U_{\bfs'|x,x'}\in\Delta(\cS) $\footnote{Since the symmetrizing distributions may depend on $ (x,x') $, we emphasize this dependence using the notation for conditional distribution.} such that for every $ y\in\cY $,
\begin{align}
\sum_{s\in\cS}U_{\bfs|x,x'}(s)W_{\bfy|\bfx,\bfs}(y|x,s) = \sum_{s\in\cS}  U_{\bfs'|x,x'}(s)W_{\bfy|\bfx,\bfs}(y|x',s). \label{eqn:omni_symm_id_two_distr} 
\end{align}
\end{definition}
Before proceeding to state Kiefer--Wolfowitz's characterization of capacity positivity, we remark that one can slightly simplify \Cref{def:omni_symm_unconstr} by restricting to the case where $ U_{\bfs|x,x'} = U_{\bfs|x,x'} $. 
Indeed, once we have (potentially different) symmetrizing distributions $ U_{\bfs|x,x'} $ and $ U_{\bfs|x,x'} $, the distribution $ V_{\bfs|x,x'} = \frac{1}{2}U_{\bfs|x,x'} + \frac{1}{2}U_{\bfs|x,x'}' $ is also symmetrizing. 
This follows trivially from \Cref{eqn:omni_symm_id_two_distr}. 
Note also that associated to each pair $ (x,x') $ there is a distribution $ U_{\bfs|x,x'} $.
All these distributions $ \curbrkt{U_{\bfs|x,x'}}_{(x,x')\in\cX^2} $ together give us a conditional distribution $ U_{\bfs|\bfx,\bfx'}\in\Delta(\cS|\cX^2) $. 
Therefore, \Cref{def:omni_symm_unconstr} can be simplified as follows.
\begin{definition}[Omnisciently symmetrizing distributions]
\label{def:omni_symm_distr_set}
Define the set of \emph{omnisciently symmetrizing distributions} as 
\begin{align}
\cU_{\omni,\symm} \coloneqq& \curbrkt{ U_{\bfs|\bfx,\bfx'}\in\Delta(\cS|\cX^2)\colon
\begin{array}{rl}
&\forall (x,x',y)\in\cX^2\times\cY,\\
&\displaystyle \sum_{s\in\cS}U_{\bfs|\bfx,\bfx'}(s|x,x')W_{\bfy|\bfx,\bfs}(y|x,s) = \sum_{s\in\cS}U_{\bfs|\bfx,\bfx'}(s|x,x')W_{\bfy|\bfx,\bfs}(y|x',s)
\end{array}
 }. \notag 
\end{align}
\end{definition}
\begin{definition}[Unconstrained omniscient symmetrizability, one-distribution version]
An unconstrained omniscient AVC $ \cA_\omni  $ is called \emph{omnisciently symmetrizable} if $ \cU_{\omni,\symm}\ne\emptyset $. 
\end{definition}
\begin{theorem}[Capacity positivity of unconstrained omniscient AVCs, \cite{kiefer-wolfowitz-1962-omniscient_confusability_unconstrained}]
An unconstrained omniscient AVC $ \cA_\omni  $ has zero capacity if and only if it is omnisciently symmetrizable. 
\end{theorem}

In the presence of (input and state) constraints, characterizing the condition for capacity positivity turns out to be significantly more challenging. 
Incorporating state constraints is one of the major obstacles and ideas that are fundamentally different from \cite{kiefer-wolfowitz-1962-omniscient_confusability_unconstrained} are required. 
Indeed, a sufficient and necessary condition is not available in the literature until very recently \cite{wbbj-2019-gen_plotkin}.
The proof, which is of combinatorial nature, is a significant generalization of the Plotkin bound in classical coding theory. 
\begin{definition}[Omniscient symmetrizability, \cite{wbbj-2019-gen_plotkin}]
\label{def:omni_symm_constr}
Given an omniscient AVC $ \cA_\omni $, an input distribution $ P_\bfx\in\lambda_\bfx $ is called \emph{omnisciently symmetrizable} if for every $ P_{\bfx,\bfx'}\in\cp(P_\bfx) $, there is a joint distribution $ P\coloneqq P_{\bfx,\bfx',\bfs,\bfs',\bfy}\in\Delta(\cX^2\times\cS^2\times\cY) $ such that 
\begin{enumerate}
	\item\label{itm:omni_symm_cond1} $ \sqrbrkt{P}_\bfx = \sqrbrkt{P}_{\bfx'} = P_\bfx,\sqrbrkt{P}_{\bfx,\bfs,\bfy} = \sqrbrkt{P}_{\bfx,\bfs}W_{\bfy|\bfx,\bfs},\sqrbrkt{P}_{\bfx',\bfs',\bfy} = \sqrbrkt{P}_{\bfx',\bfs'}W_{\bfy|\bfx,\bfs} $;
	\item\label{itm:itm:omni_symm_cond2} $ \sqrbrkt{P}_\bfs,\sqrbrkt{P}_{\bfs'}\in\lambda_\bfs $. 
\end{enumerate}
\end{definition}
It was shown in \cite{wbbj-2019-gen_plotkin} that:
\begin{theorem}[Capacity positivity of omniscient AVCs, \cite{wbbj-2019-gen_plotkin}]
\label{thm:omniscient_positivity}
An omniscient AVC $ \cA_\omni $ has zero capacity if and only if every $ P_\bfx\in\lambda_\bfx $ is omnisciently symmetrizable. 
\end{theorem}

Though having a seemingly different form, \Cref{def:omni_symm_constr} actually has an intimate relation with \Cref{def:myop_symm}.
We will massage \Cref{def:omni_symm_constr} and argue that it is a special case of \Cref{def:myop_symm}. 
The condition in \Cref{itm:omni_symm_cond1} can be rewritten as: for all $ (x,x',s,s',y)\in\cX^2\times\cS^2\times\cY $,
\begin{align}
P_{\bfx,\bfx',\bfs,\bfs',\bfy}(x,x',s,s',y) =& P_{\bfx,\bfx'}(x,x')P_{\bfs,\bfs'|\bfx,\bfx'}(s,s'|x,x') W_{\bfy|\bfx,\bfs}(y|x,s) \label{eqn:omni_symm_lhs} \\
=& P_{\bfx,\bfx'}(x,x') P_{\bfs,\bfs'|\bfx,\bfx'}(s,s'|x,x')W_{\bfy|\bfx,\bfs}(y|x',s'). \label{eqn:omni_symm_rhs} 
\end{align}
For any $ (x,x')\in\cX^2 $,
\Cref{eqn:omni_symm_lhs} and \Cref{eqn:omni_symm_rhs} are equal when marginalized out $ \bfs,\bfs' $, i.e.,
\begin{align}
\sum_{s\in\cS}P_{\bfs|\bfx,\bfx'}(s|x,x')W_{\bfy|\bfx,\bfs}(y|x,s) =& \sum_{(s,s')\in\cS^2} P_{\bfs,\bfs'|\bfx,\bfx'}(s,s'|x,x') W_{\bfy|\bfx,\bfs}(y|x,s) \notag \\
=& \sum_{(s,s')\in\cS^2} P_{\bfs,\bfs'|\bfx,\bfx'}(s,s'|x,x') W_{\bfy|\bfx,\bfs}(y|x,s') = \sum_{s'\in\cS}P_{\bfs'|\bfx,\bfx'}(s'|x,x') W_{\bfy|\bfy,\bfs}(y|x',s'). \notag 
\end{align}
It is worth noting that the above equation is precisely the one in Kiefer--Wolfowitz's notion of omniscient symmetrizability (\Cref{def:omni_symm_unconstr}).
Note also that \Cref{itm:itm:omni_symm_cond2} is equivalent to $ \sqrbrkt{ P_{\bfx,\bfx'}P_{\bfs,\bfs'|\bfx,\bfx'} }_\bfs = \sqrbrkt{ P_{\bfx,\bfx'}P_{\bfs|\bfx,\bfx'} }_\bfs\in\lambda_\bfs $ and $ \sqrbrkt{ P_{\bfx,\bfx'}P_{\bfs,\bfs'|\bfx,\bfx'} }_{\bfs'} = \sqrbrkt{ P_{\bfx,\bfx'}P_{\bfs'|\bfx,\bfx'} }_{\bfs'}\in\lambda_\bfs $.
Therefore, taking $ U_{\bfs|\bfx,\bfx'}\coloneqq \sqrbrkt{P_{\bfs,\bfs'|\bfx,\bfx'}}_\bfs $ and $ U_{\bfs'|\bfx,\bfx'}\coloneqq\sqrbrkt{P_{\bfs,\bfs'|\bfx,\bfx'}}_{\bfs'} $, we reduce \Cref{def:omni_symm_constr} to the following definition. 
\begin{definition}[Omniscient symmetrizability, two-distribution version]
\label{def:omni_symm_aux}
Given an omniscient AVC $ \cA_\omni $, an input distribution $ P_\bfx\in\lambda_\bfx $ is called \emph{omnisciently symmetrizable} if for every $ P_{\bfx,\bfx'}\in\cp(P_\bfx) $, there are $ U_{\bfs|\bfx,\bfx'},U_{\bfs'|\bfx,\bfx'} \in\Delta(\cS|\cX^2) $ satisfying \Cref{eqn:omni_symm_id_two_distr} such that $ \sqrbrkt{P_{\bfx,\bfx'}U_{\bfs|\bfx,\bfx'}}_\bfs\in\lambda_\bfs $ and $ \sqrbrkt{P_{\bfx,\bfx'}U_{\bfs'|\bfx,\bfx'}}_{\bfs'}\in\lambda_\bfs $. 
\end{definition}
However, by the observation right after \Cref{def:omni_symm_unconstr}, without loss of generality, we may as well restrict to a single symmetrizing distribution. 
The one-distribution version of \Cref{eqn:omni_symm_id_two_distr} is precisely the definition of $ \cU_{\omni,\symm} $ (\Cref{def:omni_symm_unconstr}).
We hence arrive at the following equivalent definition.
\begin{definition}[Omniscient symmetrizability, one-distribution version]
\label{def:omni_symm_myop_special}
Given an omniscient AVC $ \cA_\omni $, an input distribution $ P_\bfs\in\lambda_\bfx $ is called \emph{omnisciently symmetrizable} if for every $ P_{\bfx,\bfx'}\in\cp(P_\bfx) $, there is a $ U_{\bfs|\bfx,\bfx'}\in\cU_{\omni,\symm} $ such that $ \sqrbrkt{P_{\bfx,\bfx'}U_{\bfs|\bfx,\bfx'}}_\bfs\in\lambda_\bfs $. 
\end{definition}
It is now easy to see that \Cref{def:omni_symm_myop_special} (which is an equivalent formulation of \Cref{def:omni_symm_constr}) is a special case of \Cref{def:myop_symm}. 
Indeed, under the omniscient assumption $ \bfz = \bfx $, the channel $ W_{\bfz|\bfx} $ from Alice to James is noiseless, i.e., $ W_{\bfz|\bfx}(z|x) = \indicator{z = x} $, and \Cref{def:myop_symm_distr_set} collapses to \Cref{def:omni_symm_distr_set}.

\subsection{A seemingly more natural jamming strategy for capacity upper bound}
\label{sec:seemingly_natural}
Given the definition of symmetrizing distributions (\Cref{def:def_symm_distr}), careful readers might expect it more natural to use distributions of the form $ U_{\bfs|\bfu,\bfx_1,\cdots,\bfx_L} $ (rather than $ U_{\bfs|\bfu} $) in \Cref{sec:converse_rate_ub}. 
Specifically, they maybe would like to equip James with the following jamming strategy which seems more ``compatible'' with the definition of $ \cU_{\obli,\symml{L}} $. 
James uses the same strategy as in \Cref{sec:rate_ub_strategy} except that in Step 2) he further picks arbitrarily an $L$-list of codewords $ \vbfx_{i_1},\cdots,\vbfx_{i_L} $ ($ i_1<\cdots<i_L $) whose  type falls in the Voronoi cell of $ \wt P_{\bfx_1,\cdots,\bfx_L} $. 
Then in Step 4) he samples $ \vbfs $ from the following distribution which takes more information into account: $ U_{\vbfs|\vu, \vx_{i_1},\cdots,\vx_{i_L}} \coloneqq \prod_{j = 1}^n U_{\bfs|\bfu = \vu(j), \bfx_1 = \vx_{i_1}(j),\cdots,\bfx_L = \vx_{i_L}(j)} $. 

We now argue that the above seemingly more judicious jamming strategy will never result in a stronger (i.e., smaller) upper bound on the $L$-list-decoding capacity.
Hence the original jamming strategy we provided in \Cref{sec:rate_ub_strategy} is optimal. 
Via  analysis (which we omit) similar to that in \Cref{sec:rate_ub_analysis}, the strategy described in the last paragraph gives the following capacity upper bound:
\begin{align}
C_L(\cA_\obli)\le& \max_{ \substack{P_\bfx\in\lambda_\bfx\colon \\ P_\bfx\;\mathrm{non}\hyphen L\hyphen \mathrm{symmetrizable}} } \min_{ \substack{U_{\bfs|\bfu,\bfx_1,\cdots,\bfx_L}\in\Delta(\cS|\cU\times\cX^L)\colon \\ \sqrbrkt{ P_\bfu P_{\bfx|\bfu}^\tl U_{\bfs|\bfu,\bfx_1,\cdots,\bfx_L} }_\bfs\in\lambda_\bfs} } I(\bfx;\bfy|\bfu,\bfx_1,\cdots,\bfx_L). \label{eqn:cap_obli_avc_list_dec_seemingly} 
\end{align}
By elementary information inequalities, the above expression (\Cref{eqn:cap_obli_avc_list_dec_seemingly}) is no larger than \Cref{eqn:cap_obli_avc_list_dec}. 
Indeed,
\begin{align}
I(\bfx;\bfy|\bfu,\bfx_1,\cdots,\bfx_L) =& 
H(\bfx|\bfu,\bfx_1,\cdots,\bfx_L) - H(\bfx|\bfu,\bfx_1,\cdots,\bfx_L,\bfy) \label{eqn:def_mutual_info} \\
=& H(\bfx|\bfu) - H(\bfx|\bfu,\bfx_1,\cdots,\bfx_L,\bfy) \label{eqn:cond_indep} \\
=& I(\bfx;\bfx_1,\cdots,\bfx_L,\bfy|\bfu) \notag \\
=& I(\bfx;\bfy|\bfu) + I(\bfx;\bfx_1,\cdots,\bfx_L|\bfu,\bfy) \label{eqn:chain_rule} \\
\ge& I(\bfx;\bfy|\bfu). \label{eqn:nonnegative}
\end{align}
\Cref{eqn:def_mutual_info} is by the definition of mutual information (\Cref{def:info_mes}).
In \Cref{eqn:cond_indep}, we used the assumption  $ P_{\bfx,\bfu, \bfx_1,\cdots,\bfx_L} = P_\bfx P_\bfu P_{\bfx|\bfu}^\tl $ which implies that $ \bfx $ and $ \bfx_1,\cdots,\bfx_L $ are conditionally independent given $ \bfu $. 
More specifically, as explained in the last section (\Cref{sec:myop_obli_symm}), in the oblivious setting, we do not assume  complete positivity of the joint distribution of the transmitted signal $\bfx$ and the spoofing list $ \bfx_1,\cdots,\bfx_L $. 
In the definition of $L$-symmetrizability (\Cref{def:cp_symm}) and also the list-decoding capacity expression (\Cref{eqn:cap_obli_avc_list_dec}),  they are effectively independent. 
\Cref{eqn:chain_rule} follows from chain rule for mutual information. 
\Cref{eqn:nonnegative} is by nonnegativity of mutual information.

\section{Canonical constructions of oblivious AVCs}
\label{sec:canonical_constr}
In this section, we introduce a machinery, which we call the \emph{canonical constructions} of oblivious AVCs, for generating concrete oblivious channels with prescribed sets of $L$-symmetrizable input distributions. 
Besides being interesting in its own right, this will  help us generate examples for which $ L_\strong^* $, $ L_\cp^* $ and $ L_\weak^* $ are \emph{strictly} different. 
To warm up,  we first introduce canonical constructions of oblivious AVCs under {unique-decoding} in \Cref{sec:canonical_constr_unique_dec}. 
They will then be generalized to list-decoding in \Cref{sec:canonical_constr_list_dec}. 
Finally, we use these constructions with carefully chosen components to separate different notions of list-symmetrizability due to \cite{sarwate-gastpar-2012-list-dec-avc-state-constr} and us. 

\subsection{Canonical constructions of oblivious AVCs}
\label{sec:canonical_constr_unique_dec}
Let $ \cP\subseteq\Delta(\cX) $ be a convex subset of distributions.
Our goal is to construct an oblivious AVC $ \obliavc $ whose \emph{symmetrizability set} is precisely $ \cP $.
Here the symmetrizability set $\cK_\obli(\cA_\obli)$ of $\cA_\obli$ is defined as the set of (obliviously) symmetrizable input distributions: $$ \cK_\obli(\cA_\obli) \coloneqq \curbrkt{P_\bfx\in\lambda_\bfx\colon P_\bfx\mathrm{\ is\ symmetrizable}} .$$
Here we use \csiszar--Narayan's \cite{csiszar-narayan-it1988-obliviousavc} notion of symmetrizability (\Cref{def:cn_obli_symm}).

The construction (which we refer to as the \emph{canonical construction}) of $ \cA_{\obli,\cano} $ is as follows.
\begin{enumerate}
	\item $ \cX $ is the same as the alphabet of $\cP$.
	\item $ \cS = \cX $.
	\item $ \cY = \curbrkt{y_{x_1,x_2}}_{(x_1,x_2)\in\cX^2} $ where $ (x_1,x_2) $ is an \emph{unordered}\footnote{That is, $ y_{x_1,x_2} = y_{x_2,x_1} $} (possibly repeated) pair of input symbols. 
	Note that $ \cardY = \cardX + \binom{\cardX}{2} = \frac{1}{2}\cardX\paren{\cardX+1} = \binom{\cardX+1}{2} $. 
	\item $ \lambda_\bfx = \Delta(\cX) $\footnote{In fact any $ \lambda_\bfx $ satisfying $ \cP\subseteq\lambda_\bfx\subseteq\Delta(\cX) $ works.}.
	\item $ \lambda_\bfs = \cP $.
	\item $ W_{\bfy|\bfx,\bfs} $ is a deterministic channel which we write as a (deterministic) function $ W\colon \cX\times\cS\to\cY $ defined as $ W(x,s) = y_{x,s} $. 
\end{enumerate}

We now compute $ \cK_\obli(\cA_{\obli,\cano}) $. 
To this end, we first compute $ \cU_{\obli,\symm}$.
For $ x=x' $, the symmetrization identity in the definition of $ \cU_{\obli,\symm}$ (\Cref{def:obli_symm_distr_set}) apparently holds.
Fix $ x\ne x' $. 
The symmetrization identity gives us the following constraints on $ U_{\bfs|\bfx} $. 
\begin{enumerate}
	\item If $ y = y_{x,x} $: $ \lhs = U_{\bfs|\bfx}(x|x') = 0 = \rhs $. 
	Similarly the case where $ y = y_{x',x'} $ gives $  0 = U_{\bfs|\bfx}(x'|x)  $. 
	\item If $ y = y_{x,s^*} $ for some $ s^*\ne x,x' $: $ \lhs = U_{\bfs|\bfx}(s^*|x') = 0 = \rhs $. 
	Similarly the case where $ y = y_{x',s^*} $ for some $ s^* \ne x,x' $ gives $  0 = U_{\bfs|\bfx}(s^*|x)  $. 
	\item If $ y = y_{x,x'} $: $ \lhs = U_{\bfs|\bfx}(x'|x') = U_{\bfs|\bfx}(x|x) = \rhs $. 
\end{enumerate}
Since the above conditions hold for all $ x\ne x' $, we have: if  $ U_{\bfs|\bfx} $ is written as an $ \cardX\times\cardS $ matrix, then all off-diagonal entries are 0 and all diagonal entries are equal.
Therefore, $ U_{\bfs|\bfx} $ has to be an identity matrix (denoted by $ \id_{\bfs|\bfu} $, defined as $ \id_{\bfs|\bfx}(s|x) = \indicator{ x = s } $)
and $ \cU_{\obli,\symm}$ is a singleton set: $ \cU_{\obli,\symm}= \curbrkt{\id_{\bfs|\bfx}} $. 

We now show $ \cK_\obli(\cA_{\obli,\cano}) = \cP $. 

The direction $ \cK_\obli(\cA_{\obli,\cano}) \supseteq\cP $ is easy. 
Take any $ P_\bfx\in\cP $. 
We want to show that $ P_\bfx $ is symmetrizable. 
Since $ \cU_{\obli,\symm}$ only contains the identity distribution, we just need to evaluate $ \sqrbrkt{P_\bfx U_{\bfs|\bfx}}_{\bfs} = \paren{\diag(P_\bfx)\bfI_{\cardX\times\cardS}}^\top\vone_{\cardS} = P_\bfx\in\cP = \lambda_\bfs $, where $ \vone_\cardS $ denotes the all-one vector of length-$ \cardS $. 
Therefore, $P_\bfs $ is symmetrized by $ U_{\bfs|\bfx} = \bfI_{\cardX\times\cardS} $. 

We then show the other direction $ \cK_\obli(\cA_{\obli,\cano})\subseteq\cP $ which is in fact similar.
Take any symmetrizable input distribution $ P_\bfx $.
We want to show  $ P_\bfx\in\cP $. 
Since $ P_\bfx $ is symmetrizable and there is only one possible symmetrizing distribution, we have
$ \sqrbrkt{ P_{\bfx}U_{\bfs|\bfx} }_\bfs = P_\bfx\in\lambda_\bfs = \cP $. 

\begin{remark}
Operationally, the identity symmetrizing distribution corresponds to the following jamming strategy.
James simply uniformly samples a codeword  $ \vbfx' $ from the codebook and transmits it. 
\end{remark}

\begin{remark}
The above construction works for \emph{any} subset $\cP$ of input distributions, even for a \emph{non-convex} subset. 
However, in this paper, we only focus on AVCs whose $ \lambda_\bfs $ is convex to exclude  exotic behaviours. 
\end{remark}

\subsection{Canonical construction of oblivious AVCs under list-decoding}
\label{sec:canonical_constr_list_dec}
The above construction can be generalized to $L$-list-decoding in a natural way.
For notational brevity, we take the $L=2$ case as an example. 
Given an arbitrary subset $ \cP\subseteq\Delta(\cX) $, the goal is to construct an oblivious AVC $ \obliavc $ whose \emph{2-symmetrizability set} is equal to $\cP$. 
The  2-symmetrizability set of $\cA$ is defined as: $$ \cK_{\obli,2}(\cA_\obli) \coloneqq\curbrkt{ P_\bfx\in\lambda_\bfx\colon P_\bfx\mathrm{\ is\ 2\text{-}symmetrizable} } .$$ 
Recall that  2-symmetrizability is defined as follows, which is a special case of \Cref{def:def_symm_distr}.
\begin{definition}
The set of \emph{2-symmetrizing distributions} is defined as
\begin{align}
\cU_{\obli,\symml{2}} \coloneqq& \curbrkt{ U_{\bfs|\bfu,\bfx_1,\bfx_2}\in\Delta(\cS|\cU\times\cX^2)\colon 
\begin{array}{rl}
&\forall u\in\cU,\forall (x_0,x_1,x_2,y)\in\cX^3\times\cY,\forall\pi\in S_3, \\
&\displaystyle \sum_{s\in\cS}U_{\bfs|\bfu = u,\bfx_1,\bfx_2}(s|x_1,x_2)W_{\bfy|\bfx,\bfs}(y|x_0,s) \\
=&\displaystyle \sum_{s\in\cS}U_{\bfs|\bfu = u,\bfx_1,\bfx_2}(s|x_{\pi(1)},x_{\pi(2)})W_{\bfy|\bfx,\bfs}(y|x_{\pi(0)},s)
\end{array}
 }. \notag
\end{align}
Note that $  \cU_{\obli,\symml{1}} = \cU_{\obli,\symm}$.
\end{definition}

The following definition is a special case of \Cref{def:cp_symm}. 
We list it below for the readers' convenience. 
\begin{definition}
Given an oblivious AVC $\cA_\obli$, an input distribution $ P_\bfx\in\lambda_\bfx $ is called \emph{2-symmetrizable} if for every $ P_{\bfx_1,\bfx_2}\in\cp(P_\bfx) $ and every  $\cp$-decomposition $ (P_\bfu,P_{\bfx|\bfu}) $ of $ P_{\bfx_1,\bfx_2} $, there is a $ U_{\bfs|\bfu,\bfx_1,\bfx_2}\in\cU_{\obli,\symml{2}} $ such that $ \sqrbrkt{ P_\bfu P_{\bfx|\bfu}^{\otimes2}U_{\bfs|\bfu,\bfx_1,\bfx_2} }_\bfs\in\lambda_\bfs $. 
\end{definition}

The canonical construction $ \cA_{\obli,\canol{2}} $ is as follows.
\begin{enumerate}
	\item $ \cX $ is the same as the alphabet of $ \cP $.
	\item $ \cS = \cX^2 $.
	\item $ \cY = \curbrkt{ y_{x_0,x_1,x_2} }_{(x_0,x_1,x_2)\in\cX^3} $ where $ (x_0,x_1,x_2) $ is an unordered\footnote{That is, $ y_{x_0,x_1,x_2} = y_{x_{\pi(0)},x_{\pi(1)},x_{\pi(2)}} $ for all $ \pi\in S_3 $.} (possibly not all distinct) triple of input symbols.
	Note that $ \cardY = \cardX + 2\binom{\cardX}{2} + \binom{\cardX}{3} $.
	\item $ \lambda_\bfx = \Delta(\cX) $\footnote{Again, any $\lambda_\bfx $ such that $ \cP\subseteq\lambda_\bfx\subseteq\Delta(\cX) $ works. Indeed, We will use canonical constructions with a different $ \lambda_\bfx $ to separate different notions of list-symmetrizability.}. 
	\item\label{itm:def_lambda_s_list_dec} $ \lambda_\bfs = \bigsqcup_{P_\bfx\in\cP}\cp(P_\bfx) $.
	Note that $ \lambda_\bfs $ is convex if $\cP$ is convex. 
	To see this, take a $ \cp $ $ P $-self-coupling $ P_{\bfx_1,\bfx_2} = \sum_{i = 1}^k\lambda_i P_iP_i^\top $ and a $\cp$ $Q$-self-coupling $ P_{\bfx_1,\bfx_2}' = \sum_{j = 1}^\ell\mu_jQ_jQ_j^\top $ for some $ P,Q\in\cP $, respectively. 
	By self-coupledness, $ \sqrbrkt{P_{\bfx_1,\bfx_2}}_{\bfx_1} = \sqrbrkt{P_{\bfx_1,\bfx_2}}_{\bfx_2} = \sum_{i = 1}^k\lambda_iP_i = P $ and 
	$ \sqrbrkt{P_{\bfx_1,\bfx_2}'}_{\bfx_1} = \sqrbrkt{P_{\bfx_1,\bfx_2}'}_{\bfx_2} = \sum_{j = 1}^\ell\mu_jQ_j =Q $. 
	We would like to show that for any $ \alpha\in[0,1] $, $ R_{\bfx_1,\bfx_2}\coloneqq\alpha P_{\bfx_1,\bfx_2} + (1-\alpha)P_{\bfx_1,\bfx_2}' $ is in $ \lambda_\bfs $.
	First, $ R_{\bfx_1,\bfx_2} = \sum_{i = 1}^k\alpha\lambda_iP_iP_i^\top + \sum_{j = 1}^\ell(1-\alpha)\mu_jQ_jQ_j^\top $ is apparently $\cp$.
	Second, 
	\begin{align}
	\sqrbrkt{R_{\bfx_1,\bfx_2}}_{\bfx_1}(x_1) =& \sum_{x_2\in\cX} \paren{ \alpha\sum_{i = 1}^k\lambda_iP_i(x_1)P_i(x_2) + (1-\alpha)\sum_{j = 1}^\ell\mu_jQ_j(x_1)Q_j(x_2) } \notag \\
	=& \alpha\sum_{i = 1}^k\lambda_iP_i(x_1) + (1-\alpha)\sum_{j=1}^\ell\mu_jQ_j(x_1)\notag \\
	=& \alpha P(x_1) + (1-\alpha)Q(x_1). \notag 
	\end{align}
	The above identity also holds for the $ \bfx_2 $ marginal.
	Therefore $ R_{\bfx_1,\bfx_2} $ is an $ \paren{\alpha P+(1-\alpha)Q} $-self-coupling. 
	By convexity of $ \cP $, $ \paren{\alpha P+(1-\alpha)Q}\in\cP $. 
	Combined with complete positivity, we showed that $ R_{\bfx_1,\bfx_2}\in\lambda_\bfs $. 
	\item $ W_{\bfy|\bfx,\bfs} $ is a deterministic channel: $ W(x_0, (x_1,x_2)) = y_{x_0,x_1,x_2} $. 
\end{enumerate}

We claim that for any $ \cU $ (the alphabet of the time-sharing variable $\bfu$), $ \cU_{\obli,\symml{2}} = \curbrkt{ \id_{\bfs|\bfu,\bfx_1,\bfx_2} } $ is a singleton set.
Here $\id_{\bfs|\bfu,\bfx_1,\bfx_2}$ is defined according to: for any $ u\in\cU $, $ \id_{\bfs|\bfu = u, \bfx_1,\bfx_2}(s|x_1,x_2) = \indicator{ s = (x_1,x_2) } $. 
It is obvious that $ \id_{\bfs|\bfu,\bfx_1,\bfx_2}\in\cU_{\obli,\symml{2}} $.
We now show that any  $ U_{\bfs|\bfu,\bfx_1,\bfx_2}$ in $\cU_{\obli,\symml{2}} $ must be $ \id_{\bfs|\bfu,\bfx_1,\bfx_2} $. 
Fix any $ u\in\cU $ and $(x_1,x_2)\in\cX^{2}$ ($x_1,x_2$ can be the same). 
It suffices to show that $ U_{\bfs|\bfu = u,\bfx_1,\bfx_2}((x_1',x_2')|x_1,x_2) = 0 $ for all $ (x_1',x_2')\in\cS $ such that $ \curbrkt{x_1',x_2'} \ne\curbrkt{x_1,x_2} $. 
The condition $ \curbrkt{x_1',x_2'} \ne\curbrkt{x_1,x_2} $ implies $ x_1',x_2'\ne x_1 $ or $ x_1',x_2'\ne x_2 $.
This is more easily seen by looking at the contrapositive:
\begin{align}
& \paren{\curbrkt{x_1' = x_1\mathrm{\ or\ } x_2' = x_1}\cap\curbrkt{x_1' = x_2\mathrm{\ or\ }x_2' = x_2}} \notag \\
\implies& \paren{\curbrkt{ x_1' = x_1,x_2' = x_2 }\cup\curbrkt{x_1' = x_2,x_2' = x_1}} \equiv \paren{\curbrkt{x_1',x_2'} = \curbrkt{x_1,x_2}}. \notag 
\end{align}
Assume $ x_1',x_2'\ne x_1 $.
The other case where $ x_1',x_2' \ne x_2 $ is similar and we omit it.
Consider $ x_0,x_1,x_2\in\cX^3,y_{x_0,x_1',x_2'}\in\cY $  and $ \pi = \paren{\begin{matrix}
0&1&2 \\
1&0&2
\end{matrix}}\in S_3 $.
The symmetrization identity is specialized to: $$ \lhs = U_{\bfs|
\bfu = u,\bfx_1,\bfx_2}((x_1',x_2')|x_1,x_2 ) = 0 = \rhs .$$ 
This finishes the proof of the claim. 

We now show that $ \cK_{\obli,2}(\cA) = \cP $.

First, we show $ \cK_{\obli,2}(\cA_{\obli,\canol{2}}) \supseteq \cP $.
Take any $ P_\bfx\in\cP $. 
We claim that it is 2-symmetrizable. 
Indeed, for any $ P_{\bfx_1,\bfx_2}\in\cp(P_\bfx) $ and any  $\cp$-decomposition $ (P_\bfu,P_{\bfx|\bfu}) $ of $ P_{\bfx_1,\bfx_2} $, $ \sqrbrkt{ P_\bfu P_{\bfx|\bfu}^{\otimes2}\id_{\bfs|\bfu,\bfx_1,\bfx_2} }_\bfs = P_{\bfx_1,\bfx_2} \in\cp(P_\bfx)\subseteq \lambda_\bfs $.
The  equality is because: for any $ s = (x_1^*, x_2^*)\in\cS $, 
\begin{align}
P_\bfs((x_1^*,x_2^*) ) =& \sum_{(u,x_1,x_2 )\in\cU\times\cX^2} P_\bfu(u)P_{\bfx|\bfu}(x_1|u)P_{\bfx|\bfu}(x_2|u) \id_{\bfs|\bfu,\bfx_1,\bfx_2}(s|u,x_1,x_2) \notag \\
=& \sum_{u\in\cU} P_\bfu(u) P_{\bfx|\bfu}(x_1^*|u)P_{\bfx|\bfu}(x_2^*|u) \notag \\
=& P_{\bfx_1,\bfx_2}(x_1^*,x_2^*). \notag
\end{align}
Therefore, $ P_\bfx $ is 2-symmetrized by $ \id_{\bfs|\bfu,\bfx_1,\bfx_2}\in\cU_{\obli,\symml{2}} $. 

Second, we show $ \cK_{\obli,2}(\cA_{\obli,\canol{2}})\subseteq\cP $. 
Take any 2-symmetrizable input distribution $ P_\bfx $. 
We want to show $ P_\bfx\in\cP $. 
By 2-symmetrizability, there is a $\cp$ extension $ P_{\bfx_1,\bfx_2}\in\cp(P_\bfx) $ and a $\cp$-decomposition $ (P_\bfu,P_{\bfx|\bfu}) $ of $ P_{\bfx_1,\bfx_2} $ such that $ \sqrbrkt{ P_\bfu P_{\bfx|\bfu}^{\otimes2}\id_{\bfs|\bfu,\bfx_1,\bfx_2} }_\bfs = P_{\bfx_1,\bfx_2}\in\lambda_\bfs $. 
Since $ P_{\bfx_1,\bfx_2}\in\cp(P_\bfx) $, by the definition of $ \lambda_\bfs $ (\Cref{itm:def_lambda_s_list_dec}), we have  $ P_\bfx\in\cP $. 

\subsection{Separating $ L_\strong^*, L_\cp^*, L_\weak^* $ via canonical constructions}
\label{sec:separation}
Fix any convex set $ \cP\subseteq\Delta(\cX) $. 
Take a canonical construction $ \cA_{\obli,\canol{L}} $ with $ \lambda_\bfx = \cP $. 
We will argue that for this example: $ L_\strong^*<L_\cp^* = L_\weak^* = L $. 

We claim that $ L_\cp^* = L $. 
It is easy to check that $ L_\cp^*\ge L $ since all input distributions $ P_\bfx\in\lambda_\bfx = \cP $ are $\cp$-$L$-symmetrizable. 
Furthermore, $ \cA_{\obli,\canol{L}} $ is \emph{not} $\cp$-$(L+1)$-symmetrizable.
Indeed, $ \cU_{\obli,\symml{(L+1)}} $ is empty. 
To see this, one  evaluates the symmetrization identity 
and can prove that for any $ u\in\cU $, $$ U_{\bfs|\bfu = u,\bfx_1,\cdots,\bfx_L,\bfx_{L+1}}((x_1',\cdots,x_L')|x_0,x_1,\cdots,x_L ) = 0 $$ for all $ (x_1',\cdots,x_L')\in\cS = \cX^L $. 
We give a proof for the $ L=2 $ case and the general case is similar.
Fix any $ u\in\cU $. 
In the $L=2$ case, we want to show $ U_{\bfs|\bfu = u,\bfx_1,\bfx_2,\bfx_3}((x_1',x_2')|x_1,x_2,x_3 ) = 0 $ for all $ (x_1',x_2')\in\cS = \cX^2 $ and $ (x_1,x_2,x_3)\in\cX^3 $. 
For the symmetrization identity to be nontrivial, we had better take $ x_0,x_1,x_2,x_3 $ to be not all the same. 
Take $ (x_1,x_2,x_3)\in\cX^3 $ and $ x_0\in\cX\setminus\curbrkt{x_1,x_2,x_3} $. 
Take  $ y_{x_0,x_1',x_2'}\in\cY $ for any $ (x_1',x_2')\in\cS = \cX^2 $ and $\pi = \paren{\begin{matrix}
0&1&2&3 \\
1&0&2&3
\end{matrix}}\in S_4 $. 
Note that at least one of the following cases hold: 
$ x_1\notin\curbrkt{x_1',x_2'}$ or $x_2\notin\curbrkt{x_1',x_2'}$ or $x_3\notin\curbrkt{x_1',x_2'} $.
Assume WLOG $ x_1\notin\curbrkt{x_1',x_2'} $. 
Other two cases are similar.
Since $ x_0\ne x_1 $, we have $ \lhs = U_{\bfs|\bfu = u,\bfx_1,\bfx_2,\bfx_3}((x_1',x_2')|x_1,x_2,x_3 ) = 0 = \rhs $. 
Therefore, $ \cU_{\obli,\symml{3}} = \emptyset $.

We claim that $ L_\strong^*<L $. 
Apparently, by definition, $ L^*_\strong \le L_\cp^* $. 
It is left to show that $ \cA_{\obli,\canol{L}} $ is not strongly-$L$-symmetrizable.
In fact, \emph{no} input distribution $ P_\bfx\in\lambda_\bfx $ is strongly-$L$-symmetrizable. 
This is because: for any $ P_\bfx\in\lambda_\bfx $ and any $ P_{\bfx_1,\bfx_2}\in\cJ(P_\bfx)\setminus\cp(P_\bfx) $\footnote{Note that $ \cJ(P_\bfx)\setminus\cp(P_\bfx)\ne\emptyset $.}, 
$ \sqrbrkt{ P_{\bfx_1,\bfx_2}\id_{\bfs|\bfx_1,\bfx_2} }_\bfs = P_{\bfx_1,\bfx_2}\notin\cp(P_\bfx) $ and hence $ \sqrbrkt{ P_{\bfx_1,\bfx_2}\id_{\bfs|\bfx_1,\bfx_2} }_\bfs\notin\lambda_\bfs = \bigsqcup_{P_\bfx\in\cP}\cp(P_\bfx) $. 

We claim that $ L_\weak^* = L $. 
By definition, $ L_\weak^*\ge L_\cp^* = L $.
Furthermore, since $ \cU_{\obli,\symml{(L+1)}} = \emptyset $, $ L_\weak^*<L+1 $. 
This proves the claim. 

We now give an example for which $ L_\cp^*<L_\weak^* = L $. 
Let $ \cP = \curbrkt{P_\bfx} $ be a singleton set for some $ P_\bfx\in\Delta(\cX) $. 
Let $ \cA_{\obli,\canol{L}} $ be a canonical construction of oblivious AVC with $ \lambda_\bfx = \curbrkt{P_\bfx} $ and $ \lambda_\bfs = \curbrkt{P_\bfx^\tl} $. 
By the same argument as before, $ L_\weak^* = L $. 
However, for non-product $\cp$ distributions $ P_{\bfx_1,\cdots,\bfx_L}\in\cp^\tl(P_\bfx)\setminus\curbrkt{P_\bfx^\tl} $\footnote{Note that $ \cp^\tl(P_\bfx)\setminus\curbrkt{P_\bfx^\tl}\ne\emptyset $. }, no matter which $\cp$-decomposition we take,  $ \sqrbrkt{P_{\bfu,\bfx_1,\cdots,\bfx_L}\id_{\bfs|\bfu,\bfx_1,\cdots,\bfx_L}}_\bfs = P_{\bfx_1,\cdots,\bfx_L} \notin\curbrkt{P_\bfx^\tl} = \lambda_\bfs $. 
Therefore $ L_\cp^*<L $.

\section{Concluding remarks and open problems}
In this paper, we revisited the classical problem of list-decoding for oblivious Arbitrarily Varying Channels.
We pinned down the exact threshold of list-size that determines the positivity of list-decoding capacity.
The proof utilized a machinery recently developed in \cite{wbbj-2019-gen_plotkin} and \cite{zbj-2019-generalized-ld}. 
A natural capacity lower bound was proved.
However, our capacity upper bound is conditioned on \Cref{conj:comb_conj}. 

We list several open problems for future research.
\begin{enumerate}
	\item The most obvious open question is to obtain a tight characterization of the list-decoding capacity.
	This amounts to proving/disproving/bypassing \Cref{conj:comb_conj}. 
	We do believe that \Cref{conj:comb_conj} is true and our lower bound (\Cref{eqn:cap_obli_avc_list_dec}) is tight, though proving it requires new ideas.
	Otherwise, one could try to develop other jamming strategies whose analysis bypasses this conjecture.
	\item 
	In our capacity bounds (\Cref{eqn:cap_obli_avc_list_dec}), the maximization (implicitly) includes searching over \emph{all} $\cp$-decompositions $ (P_\bfu,P_{\bfx|\bfu}) $ of a $\cp$-distribution $ P_{\bfx_{[L]}}\in\cp^\tl(P_\bfx) $.
	As mentioned in \Cref{rk:cp-decomp-non-unique}, a $\cp$-distribution can have multiple $\cp$-decompositions. 
	Among all decompositions, the smallest induced $ \cardU $  is called the $\cp$-rank. 
	We do not know whether decompositions with $ \cardU>\cprk(P_{\bfx_{[L]}}) $ will ever be maximizers.
	Also, there are $\cp$-distributions whose $\cp$-decompositions can have arbitrarily large $ \cardU $. 
	For the purpose of maximizing the mutual information in \Cref{eqn:cap_obli_avc_list_dec}, we  do not have a cardinality bound on $ \cU $ of the $\cp$-decomposition. 
	Therefore, \Cref{eqn:cap_obli_avc_list_dec} is not computable for general AVCs. 
	On the other hand, it may be the case that for some channels, to approach the maximum of \Cref{eqn:cap_obli_avc_list_dec}, we cannot put an upper bound on $ \cardU $. 
	We leave this issue for future exploration. 
	\item We only found (via the canonical constructions) examples of oblivious AVCs  for which $ L_\strong^*<L_\cp^* = L_\weak^* $ and other examples of oblivious AVCs for which $ L_\cp^*<L_\weak^* $. 
	We are curious to see a single example for which $ L_\strong^*<L_\cp^*<L_\weak^* $. 
	\item It is possible to extend our results to multiuser channels.
	Arguably the most well-understood multiuser channels would be the Multiple Access Channels (MACs). 
	A two-user MAC consists of two transmitters who transmit their encodings simultaneously and a single receiver who wants to decode \emph{both} messages. 
	In the presence of an oblivious adversary, the fundamental limits of MACs are not fully understood until recently \cite{jahn-1981-avmac,ahlswedecai-1999-obli-avmac-no-constr,pereg-steinberg-2019-mac}. 
	For the list-decoding variant, Cai \cite{cai-2016-list-dec-obli-avmac} gave the right notion of list-symmetrizability for MACs and  completely determined the list-decoding capacity of \emph{unconstrained} MACs.
	To complete the picture along this line of research, it is desirable to extend Cai's results to the \emph{constrained} case.
	If one follows the techniques in this paper, this will likely require us to first prove a Plotkin-type converse for \emph{omniscient} MACs, which itself is highly nontrivial and intriguing. 
	Given the obstacles even in the point-to-point case, a complete characterization of the list-decoding capacity of MACs seems challenging. 
	\item Guruswami and Smith \cite{guruswami-smith-2016-explicit-bounded} constructed explicit codes equipped with (stochastic) encoders and decoders that run in \emph{polynomial} time for the oblivious bitflip channels. 
	Towards explicit fast encodable/decodable code constructions for general oblivious AVCs, it is worth exploring the limits to Guruswami--Smith's techniques. 
\end{enumerate}

\bibliographystyle{alpha}
\bibliography{IEEEabrv,ref}

\appendices

\section{Robust generalized Plotkin bound}
\label[app]{app:robust_plotkin}
In this section, we give a proof sketch of \Cref{thm:robust_plotkin}.
It is essentially a corollary of \cite{zbj-2019-generalized-ld} with a twist that the code can be \emph{approximately} $ \wh P_\bfx $-constant-composition rather than \emph{exactly} $ \wh P_\bfx $-constant-composition.
We prove that the generalized Plotkin bound in \cite{zbj-2019-generalized-ld} still holds for approximate constant-composition codes. 
The proof is basically the same with more slack factors to take care of. 

We will give a proof sketch to a stronger theorem.
To state the stronger theorem, we need several definitions.

\begin{definition}[Approximate self-coupling]
\label{def:apx_self_coupling}
Let $ \lambda>0 $ be a constant and $ \wh P_\bfx\in\Delta(\cX) $. 
A joint distribution $ P_{\bfx_1,\cdots,\bfx_{ L }}\in\Delta(\cX^{ L }) $ is called a \emph{$ (\lambda,\wh P_\bfx) $-self-coupling} if
\begin{align}
d\paren{ \sqrbrkt{ P_{\bfx_1,\cdots,\bfx_{ L }}}_{\bfx_i}, \wh P_\bfx }\le\lambda \notag
\end{align}
for every $ 1\le i\le  L  $. When $\lambda=0$, we say that $ P_{\bfx_1,\cdots,\bfx_{ L }} $ is a $ \wh P_\bfx $-self-coupling which agrees with  \Cref{def:def_self_coupling}.
The set of $ (\lambda,\wh P_\bfx) $-self-couplings of order-$L$ is denoted by $ \cJ_\lambda^\tl(\wh P_\bfx) $. 
When $ \lambda = 0 $, $ \cJ_0^\tl(\wh P_\bfx) = \cJ^\tl(\wh P_\bfx) $.
\end{definition}

For every approximate self-coupling, there is always an exact self-coupling close by. 
\begin{lemma}[Distribution approximation \cite{sarwate-gastpar-2012-list-dec-avc-state-constr}]\label{lem:distr_apx}
Suppose a joint distribution $ \wh P_{\bfx_1,\cdots,\bfx_{ L }} $ is a $ (\lambda,\wh P_\bfx) $-self-coupling for some  $ \wh P_\bfx\in\Delta(\cX) $. 
Then there exists a $ \wh P_\bfx $-self-coupling $ \wt P_{\bfx_1,\cdots,\bfx_{ L }} $ such that 
$ d\paren{\wt P_{\bfx_1,\cdots,\bfx_L}, \wh P_{\bfx_1,\cdots,\bfx_L}}\le f(\lambda) $ for some $ f(\lambda)>0 $ such that $ f(\lambda)\xrightarrow{\lambda\to0}0 $. 
\end{lemma}

\begin{remark}
In the original paper \cite{sarwate-gastpar-2012-list-dec-avc-state-constr}, Sarwate and Gastpar chose to use the $ \ell^\infty $-norm, i.e., $ d_{\ell^\infty}\paren{P,Q}\coloneqq\max_{x\in\cX}\abs{P(x) - Q(x)} $ for $ P,Q\in\Delta(\cX) $. 
Despite their choice, $ \ell^1 $-norm (which we use in this paper) only differs from $ \ell^\infty $-norm by a multiplicative constant factor. 
Indeed, for any constant-sized alphabet $ \cX $ and any constant list-size $ L\in\bZ_{\ge1} $, for any $ P_{\bfx_{[L]}},Q_{\bfx_{[L]}}\in\Delta(\cX^L) $, $ d_{\ell^\infty}\paren{P_{\bfx_{[L]}},Q_{\bfx_{[L]}}} \le d_{\ell^1}\paren{P_{\bfx_{[L]}},Q_{\bfx_{[L]}}}\le \cardX^L d_{\ell^\infty}\paren{P_{\bfx_{[L]}},Q_{\bfx_{[L]}}} $. 
Therefore, up to multiplicative constant factors, \Cref{lem:distr_apx} holds under $ d_{\ell^1} $ as well. 
\end{remark}

\begin{definition}[Robust confusability set]
Let $ \omniavc $ be an omniscient channel with a 0-1 channel transition distribution $ W_{\bfy|\bfx,\bfs} $ and convex input \& state constraints $ \lambda_\bfx\subseteq\Delta(\cX) $ and $ \lambda_\bfs\subseteq\Delta(\cS) $, respectively. 
Let $ \lambda>0 $ be a constant and $ P_\bfx\in\lambda_\bfx $. 
Define the \emph{robust confusability set} associated to this channel as
\begin{align}
\cK_\lambda^\tl(P_\bfx)\coloneqq& \curbrkt{ P_{\bfx_{[L]}}\in\cJ_\lambda^\tl(P_\bfx)  \colon\begin{array}{rl}
\exists P_{\bfx_{[L]}, \bfs_{[L]}, \bfy}\in\Delta(\cX^L\times\cS^L\times\cY),\;&\suchthat\\
\sqrbrkt{P_{\bfx_{[L]}, \bfs_{[L]}, \bfy}}_{\bfx_{[L]}} =&P_{\bfx_{[L]}}\\
\forall i\in[L],\;\sqrbrkt{P_{\bfx_{[L]}, \bfs_{[L]}, \bfy}}_{\bfs_i}\in& \lambda_\bfs \\
\forall i\in[L],\;\sqrbrkt{P_{\bfx_{[L]}, \bfs_{[L]}, \bfy}}_{\bfx_i,\bfs_i,\bfy} =& P_{\bfx_i,\bfs_i} W_{\bfy|\bfx,\bfs}
\end{array}  }. \notag 
\end{align}
When $ \lambda = 0 $, we write $ \cK^\tl(P_\bfx) $ for $ \cK_\lambda^\tl(P_\bfx) $. 
\end{definition}

\begin{theorem}[Robust generalized Plotkin bound, general form]
\label{thm:robust_plotkin_stronger}
Assume that
\begin{align}
d\paren{ \cJ^\tl(P_\bfx)\setminus\cK^\tl(P_\bfx), \cp^\tl(P_\bfx) }\ge\eps
\label{eqn:cond_robust_plotkin_stronger}
\end{align}
 for some constant $ \eps>0 $.
Let $ \cC $ be a $ (\lambda,P_\bfx) $-constant-composition code satisfying
 $\Gamma^\tl(\cC)\subset \cK_\lambda^\tl(P_\bfx) $. 
Then there is a constant $K = K(\eps, \lambda,L,\cardX)>0$ such that $ \cardC\le K $.
\end{theorem}
\begin{remark}
Note that $K$ does \emph{not} depend on $n$ --  the blocklength of $\cC$.
\end{remark}
\begin{remark}
Apparently, the condition given by  \Cref{eqn:cond_robust_plotkin_stronger} implies the condition given by \Cref{eqn:cond_robust_plotkin} in \Cref{thm:robust_plotkin}. 
Hence \Cref{thm:robust_plotkin_stronger} implies \Cref{thm:robust_plotkin}. 
\end{remark}

\begin{proof}[Proof sketch of \Cref{thm:robust_plotkin_stronger}.]
Take an $ \eta $-net $\cN$ of $ \cJ_\lambda^\tl(P_\bfx)\setminus\cK_\lambda^\tl(P_\bfx) $. 
Note that $ \card{\cN} $ is a constant. 
Use hypergraph Ramsey's theorem to find a subcode $ \cC'\subset\cC $ such that $ d\paren{\Gamma^\tl(\cC'),  P_{\bfx_{[L]}}'}\le \eta $ for some $ P_{\bfx_{[L]}}'\in\cN $. 
The subcode is guaranteed to satisfy $ \card{\cC'}\xrightarrow{\cardC\to\infty}\infty $. 
Note that $ P_{\bfx_{[L]}}'\in\cN\subset\cJ_\lambda^\tl(P_\bfx) $ may  not be exactly a $ P_\bfx $-self-coupling.
However, by \Cref{lem:distr_apx}, there is a constant $ \lambda' = \lambda'(\lambda)>0 $ and a $ P_\bfx $-self-coupling $ P_{\bfx_{[L]}}\in\cJ^\tl(P_\bfx) $ such that $ d\paren{ P_{\bfx_{[L]}}, P_{\bfx_{[L]}}' }\le\lambda' $. 
We point out that $ P_{\bfx_{[L]}} $ may be inside $ \cK_\lambda^\tl(P_\bfx) $.
Nevertheless, as long as  $ \lambda\ll \eps $, $ P_{\bfx_{[L]}} $ is still outside $ \cp^\tl(P_\bfx) $ since $ \cp^\tl(P_\bfx) $ is a proper subset of $ \cK^\tl(P_\bfx) $ and they are $\eps$-separated. 
Indeed, 
\begin{align}
d\paren{ P_{\bfx_{[L]}}, \cp^\tl(P_\bfx) } \ge&  d\paren{ P_{\bfx_{[L]}}', \cp^\tl(P_\bfx) } - d\paren{ P_{\bfx_{[L]}}',P_{\bfx_{[L]}} } \ge d\paren{ \cJ^\tl(P_\bfx)\setminus\cK^\tl(P_\bfx), \cp^\tl(P_\bfx) } -  \lambda' \ge \eps - \lambda'. \notag 
\end{align}
See \Cref{fig:robust_plotkin} for the geometry of various sets of distributions that show up in the proof. 
\begin{figure}[htbp]
	\centering
	\includegraphics[width=0.5\textwidth]{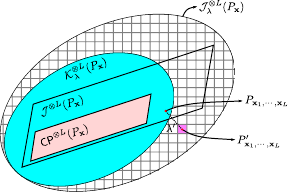}
	\caption{The geometry of various sets of distributions in the proof of the robust generalized Plotkin bound (\Cref{thm:robust_plotkin_stronger}).}
	\label{fig:robust_plotkin}
\end{figure}

Define the \emph{asymmetry} of a distribution $ P_{\bfx_{[L]}}\in\Delta(\cX^L) $ as 
\begin{align}
\asymm(P_{\bfx_{[L]}}) \coloneqq& \max_{\pi\in S_L} \max_{x_{[L]}\in\cX^L} \abs{ P_{\bfx_{[L]}}(x_{[L]}) - P_{\bfx_{[L]}}(x_{\pi([L])}) }. \notag 
\end{align}
Given $ P_{\bfx_{[L]}}\in\cJ^\tl(P_\bfx) $, we now consider two cases where  $ \asymm(P_{\bfx_{[L]}})\ge\alpha $ and $ \asymm(P_{\bfx_{[L]}})<\alpha $. 
Let $ M'\coloneqq\card{\cC'} $ and $ M\coloneqq\cardC $. 
Write $ \cC'=\curbrkt{\vx_i}_{i\in[M']} $. 
\paragraph{Asymmetric case}

If $ P_{\bfx_{[L]}} $ has asymmetry at least $ \alpha $, the code $ \cC' $ has to be small by a result of \cite{komlos-1990-strange_pigeonhole} and its list-decoding extension by \cite{bondaschi-dalai-2019-revisiting}. 
Their results do not have to do with the self-coupledness of $ P_{\bfx_{[L]}} $  and hence are directly reusable. 
\paragraph{Symmetric case}
In this case we assume $ P_{\bfx_{[L]}} $ has asymmetry less than $ \alpha $. 
We first project $ P_{\bfx_{[L]}} $ to $ \sym^\tl(P_\bfx) $ to get an \emph{exactly} symmetric distribution. 
The projection is nothing but the symmetrization of $ P_{\bfx_{[L]}} $: 
\begin{align}
\ol{P}_{\bfx_{[L]}} \coloneqq& \frac{1}{L!} \sum_{\pi\in S_L} P_{\bfx_{\pi([L])}}, \notag 
\end{align}
which is apparently symmetric. 
Using equicoupledness of $ P_{\bfx_{[L]}} $, one can easily check the equicoupledness of $ \ol{P}_{\bfx_{[L]}} $. 
Indeed, for any $ i\in[L] $, 
\begin{align}
\sqrbrkt{\ol{P}_{\bfx_{[L]}}}_{\bfx_i} =& \sqrbrkt{ \frac{1}{L!} \sum_{\pi\in S_L} P_{\bfx_{\pi([L])}} }_{\bfx_i} = \frac{1}{L!} \sum_{\pi\in S_L} \sqrbrkt{P_{\bfx_{\pi([L])}}}_{\bfx_i}= \frac{1}{L!} \sum_{\pi\in S_L} \sqrbrkt{P_{\bfx_{[L]}}}_{\bfx_{\pi^{-1}(i)}} = \frac{1}{L!} \sum_{\pi\in S_L} P_\bfx= P_\bfx. \notag 
\end{align}
Therefore $ \ol{P}_{\bfx_{[L]}} \in\sym^\tl(P_\bfx) $. 
We also note that $ \ol{P}_{\bfx_{[L]}} $ is not too far from $ P_{\bfx_{[L]}} $:
\begin{align}
d\paren{P_{\bfx_{[L]}}, \ol{P}_{\bfx_{[L]}} } =& \sum_{x_{[L]}\in\cX^L} \abs{ P_{\bfx_{[L]}}(x_{[L]}) - \ol{P}_{\bfx_{[L]}}(x_{[L]}) } \notag \\
=& \sum_{x_{[L]}\in\cX^L}\abs{ \frac{1}{L!} \sum_{\pi\in S_L}\paren{ P_{\bfx_{[L]}}(x_{[L]}) - P_{\bfx_{\pi([L])}}(x_{[L]}) } }\notag \\
\le& \frac{1}{L!} \sum_{\pi\in S_L} \sum_{x_{[L]}\in\cX^L} \abs{ P_{\bfx_{[L]}}(x_{[L]}) - P_{\bfx_{\pi([L])}}(x_{[L]}) } \notag \\
\le & \paren{\cardX^L - \cardX}\alpha \eqqcolon \alpha'. \notag 
\end{align}
By taking $ \alpha' $ to be sufficiently small, we can ensure $ \ol{P}_{\bfx_{[L]}}\notin\cp(P_\bfx) $. 
Indeed, 
\begin{align}
d\paren{ \ol{P}_{\bfx_{[L]}}, \cp^\tl(P_\bfx) } \ge& d\paren{ P_{\bfx_{[L]}}, \cp^\tl(P_\bfx) } - d\paren{P_{\bfx_{[L]}}, \ol{P}_{\bfx_{[L]}}} \ge \eps - \lambda' - \alpha'. \notag 
\end{align}
Since $ \ol{P}_{\bfx_{[L]}} \in\sym^\tl(P_\bfx)\setminus\cp^\tl(P_\bfx) $,
the duality between completely positively tensor cone and \emph{copositive} tensor cone (denoted by $ \cop^\tl(P_\bfx) $)\footnote{A distribution $ P_{\bfx_{[L]}}\in\sym^\tl(P_\bfx) $ is called \emph{$ P_\bfx $-copositive} if $ \inprod{P_{\bfx_{[L]}}}{Q_\bfx^\tl}\ge0 $ for all $ Q_\bfx\in\Delta(\cX) $. The set of copositive distributions form a convex cone, denoted by $ \cop^\tl(P_\bfx) $. It turns out that $ \cp^\tl $ and $ \cop^\tl $ are dual cones of each other.} guarantees the existence of a witness $ Q_{\bfx_{[L]}}\in\cop^\tl(P_\bfx) $ of non-complete positivity.
The witness satisfies that  $ \inprod{\ol{P}_{\bfx_{[L]}}}{Q_{\bfx_{[L]}}}\le-\eps' $ for some $ \eps' = \eps'(\eps,\lambda')>0 $. 
To get an upper bound on $ M' $ (and hence an upper bound $ M $), we bound the following quantity from above and below:
$\sum_{ \cL\in[M']^L }\inprod{ \tau_{\vx_\cL} }{Q_{\bfx_{[L]}}} $.

On the one hand, via the method of types,  it is easy to show that the above quantity is nonnegative. 
Indeed, it is precisely equal to  
\begin{align}
\frac{M'^L}{n}\sum_{j = 1}^n\inprod{\paren{P_\bfx^{(j)}}^\tl}{Q_{\bfx_{[L]}}} \ge0 , \label[ineq]{eqn:double_counting_lb}
\end{align}
where $ P_{\bfx}^{(j)} $ is the type of the $j$-the \emph{column} of $ \cC'\in\cX^{M'\times n} $.
\Cref{eqn:double_counting_lb} follows since $ \paren{P_{\bfx}^{(j)}}^\tl\in\cp $ for each $ j\in[n] $. 
See \cite{zbj-2019-generalized-ld} for details. 

On the other hand, 
\begin{align}
\sum_{ \cL\in[M']^L }\inprod{ \tau_{\vx_\cL} }{Q_{\bfx_{[L]}}} =& \sum_{\cL\in \binom{[M']}{L}}\inprod{ \tau_{\vx_\cL} }{Q_{\bfx_{[L]}}} + \sum_{\cL\in[M']^L \setminus\binom{[M']}{L}}\inprod{ \tau_{\vx_\cL} }{Q_{\bfx_{[L]}}}. \notag
\end{align}
It turns out that the second term is a lower order term and it suffices to use a trivial bound: 
\begin{align}
\sum_{\cL\in \binom{[M']}{L}}\inprod{ \tau_{\vx_\cL} }{Q_{\bfx_{[L]}}}\le M'^L - \binom{M'}{L}. \label[ineq]{eqn:not_all_distinct}
\end{align}

We then focus on the first term. 
Before proceeding, we first observe that without loss of generality, we can assume $ \cL = \curbrkt{i_1,\cdots,i_L} $ is an ordered list, i.e.,  $ i_1<\cdots<i_L $. 
For a list $\cL$ that is not in ascending order, 
there is a $ \sigma_\cL\in S_L $ such that $ \sigma_\cL(\cL) $ is in ascending order. 
The observation then follows since 
\begin{align}
\inprod{\tau_{\vx_\cL}}{Q_{\bfx_{[L]}}} = \inprod{\tau_{\vx_{\sigma_\cL(\cL)}}}{Q_{\bfx_{\sigma_\cL([L])}}} = \inprod{ \tau_{\vx_{\sigma_\cL(\cL)}} }{Q_{\bfx_{[L]}}}, \notag 
\end{align}
where the last equality is by the symmetry of $ Q_{\bfx_{[L]}}\in\cop^\tl(P_\bfx)\subset \sym^\tl(P_\bfx) $. 

We now bound the first term.
\begin{align}
\sum_{\cL\in \binom{[M']}{L}}\inprod{ \tau_{\vx_\cL} }{Q_{\bfx_{[L]}}} =& \sum_{\cL\in \binom{[M']}{L} } \inprod{ \tau_{\vx_\cL} - \ol{P}_{\bfx_{[L]}} }{Q_{\bfx_{[L]}}} + \sum_{\cL\in \binom{[M']}{L}}\inprod{  \ol{P}_{\bfx_{[L]}} }{Q_{\bfx_{[L]}}} \notag \\
\le& \sum_{\cL\in \binom{[M']}{L} }  \normone{\tau_{\vx_\cL} - \ol{P}_{\bfx_{[L]}}} \norminf{Q_{\bfx_{[L]}}} + \sum_{\cL\in \binom{[M']}{L} } (-\eps') \notag \\
\le& \sum_{\cL\in \binom{[M']}{L} } {  \normone{\tau_{\vx_\cL} - \ol{P}_{\bfx_{[L]}}} } - \binom{M'}{L} \eps' \notag \\
\le& \sum_{\cL\in \binom{[M']}{L} } \paren{ d\paren{\tau_{\vx_\cL}, P_{\bfx_{[L]}}'} + d\paren{ P_{\bfx_{[L]}}', P_{\bfx_{[L]}}} + d\paren{P_{\bfx_{[L]}} , \ol{P}_{\bfx_{[L]}}} } - \binom{M'}{L} \eps'\notag \\
\le& \binom{M'}{L}\paren{ \eta + \lambda' + \alpha' - \eps' }. \label[ineq]{eqn:all_distinct}
\end{align}
Combining \Cref{eqn:double_counting_lb,eqn:not_all_distinct,eqn:all_distinct} allows us to establish \Cref{thm:robust_plotkin_stronger}. 
\end{proof}
\begin{remark}
Though the robust version of the generalized Plotkin bound allows the input distributions to be slightly perturbed from a fixed composition, in the actual double counting argument, we still made the effort to preprocess $ P_{\bfx_{[L]}}'\in\cJ_\lambda^\tl(P_\bfx) $ to get $ P_{\bfx_{[L]}}\in\cJ^\tl(P_\bfx) $ and then $ \ol P_{\bfx_{[L]}}\in\sym^\tl(P_\bfx) $. 
This process created several  slack factors that ultimately turned out to be negligible.
The reason why we desire a symmetric self-coupling is at its root  that the notion of complete positivity only makes sense for symmetric tensors whose marginals are all the same. 
After all, $\cp$-tensors are defined to take the form $ \sum_i\lambda_iP^\tl $ rather than $ \sum_i\lambda_iP_i^\tl $ where $ P_i $'s can be  different. 
\end{remark}

\section{$\cp$-symmetrization converse}
\label[app]{app:converse_rate_zero}
We assume that $ P_\bfs = \sqrbrkt{ \wt P_\bfu \wt P_{\bfx|\bfu}^\tl U_{\bfs|\bfu,\bfx_1,\cdots,\bfx_{ L }} }_\bfs $ (where $ (\wt P_\bfu, \wt P_{\bfx|\bfu}) $ was constructed in \Cref{sec:cp_symm}) is \emph{strictly} inside $ \lambda_\bfs $ in the sense that $ P_\bfs $ is in the $\vdelta$-interior of $ \lambda_\bfs $ for some entry-wise positive vector $ \vdelta = (\delta_1,\cdots,\delta_\beta)\in\bR_{>0}^{\beta} $. 
More precisely, we assume \Cref{eqn:cost_assumption} (which we recall below) holds: 
$\cost_i((P_\bfu,P_{\bfx|\bfu}), U_{\bfs|\bfu,\bfx_{[L]}} ) \le\Lambda_i - \delta_i$ for all $ i\in[\beta] $. 
See \Cref{fig:state_constr}.
\begin{figure}[htbp]
	\centering
	\includegraphics[width=0.4\textwidth]{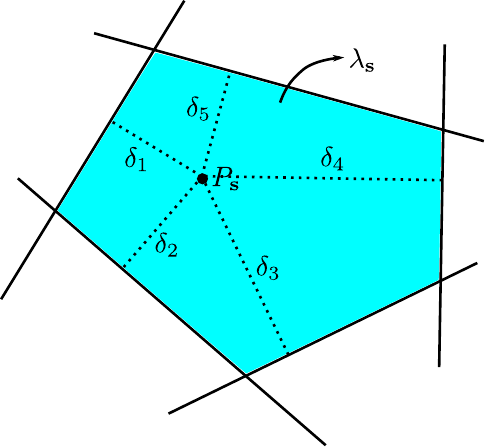}
	\caption{We assume in $\cp$-symmetrization that $ P_\bfs $ is in the $ \vdelta $-interior of $ \lambda_\bfs $.}
	\label{fig:state_constr}
\end{figure}

Under the jamming strategy described in \Cref{sec:cp_symm}, we lower bound the expected average error probability over the list selection and the jamming sequence generation. 

Define a good event that the list chosen by James has a type that is $(\eta + \eps)$-close to a certain $\cp$-distribution: 
\begin{align}
\cG\coloneqq\curbrkt{ d\paren{\tau_{\vbfx_\cL},\wt P_{\bfx_{[L]}}}\le\eta + \eps }. 
\label{eqn:def_g}
\end{align}
Note that $\cG$ only depends on the choice of $\cL$ but not $\vbfs$. 
Then we have
\begin{align}
\exptover{\cL,\vbfs}{P_{\e,\avg}(\vbfs)} =&  \exptover{\cL}{ \exptover{\vbfs}{P_{\e,\avg}(\vbfs)\condon\cL} } 
\ge \exptover{\cL}{ \exptover{\vbfs}{P_{\e,\avg}(\vbfs)\one_{\cG}\condon\cL} } 
=\exptover{\cL}{ \one_{\cG}\exptover{\vbfs}{P_{\e,\avg}(\vbfs)\condon\cL} } \notag \\
=& \probover{\cL}{\cG}\exptover{\cL} {\exptover{\vbfs}{P_{\e,\avg}(\vbfs)\condon\cL}\condon\cG}
= \probover{\cL}{\cG}\exptover{\cL}{\exptover{\vbfs}{P_{\e,\avg}(\vbfs)\condon\cL,\cG}}. \label{eqn:pe_decomp} 
\end{align}
The expectations and probabilities are over $ \cL\sim\binom{[M]}{L} $ and $ \vbfs\sim\prod_{j = 1}^nU_{\bfs|\bfu,\bfx_{[L]} = \vbfx_\cL(j)} $. 

We first argue that $ \prob{\cG} $ is bounded away from zero, $ \prob{\cG}\ge c $.
Let $ \cF\coloneqq\curbrkt{ d\paren{\tau_{\vbfx_\cL}, \wh P_{\bfx_{[L]}}}\le\eta } $.
As we have shown in \Cref{eqn:cp_tuple_extraction}, $ \prob{\cF}\ge c $. 
Note that  by \Cref{eqn:dist_hat_tilde}, $ \cF $ implies
\begin{align}
d\paren{\tau_{\vbfx_\cL},\wt P_{\bfx_{[L]}}}\le d\paren{\tau_{\vx_\cL}, \wh P_{\bfx_{[L]}} } + d\paren{ \wh P_{\bfx_{[L]}}, \wt P_{\bfx_{[L]}}} \le \eta + \eps. \notag
\end{align}
Hence $ \prob{\cG}\ge c $. 

It remains to bound $ \expt{\expt{P_{\e,\avg}(\vbfs)\condon\cL,\cG}} $. 

Fix any realization $ 1\le i_1<\cdots<i_L\le M $ of the random list $\cL$ that satisfies $ \cG $. 
All expectations and probabilities in the following analysis are over $ \vbfs $. 
For any $ i\in[\beta] $, the expected cost of $ \vbfs $ given $ \cL $ is 
\begin{align}
&\expt{B_i(\vbfs)\condon\cL,\cG} \notag \\
=& \frac{1}{n}\sum_{j = 1}^n\expt{B_i(\vbfs(j))\condon\cL,\cG} \notag \\
=& \frac{1}{n}\sum_{j = 1}^n \sum_{s\in\cS}  U_{\bfs|\bfu,\bfx_{[L]}}(s|\vu(j),\vx_{i_1}(j),\cdots,\vx_{i_L}(j))B_i(s)\notag \\
=& \frac{1}{n}\sum_{j = 1}^n\sum_{s\in\cS} \sum_{u\in\cU} \sum_{(x_1,\cdots,x_L)\in\cX^L} \indicator{\vu(j) = u} \indicator{\vx_{i_1}(j) = x_1,\cdots,\vx_{i_L}(j) = x_L}  U_{\bfs|\bfu,\bfx_{[L]}}(s|u,x_1,\cdots,x_L )B_i(s) \notag \\
=& \sum_{s\in\cS}\sum_{u\in\cU} \sum_{(x_1,\cdots,x_L)\in\cX^L}U_{\bfs|\bfu,\bfx_{[L]}}(s|u,x_1,\cdots,x_L )B_i(s)\frac{1}{n}\sum_{j = 1}^n\indicator{\vu(j) = u} \indicator{\vx_{i_1}(j) = x_1,\cdots,\vx_{i_L}(j) = x_L} \notag \\
=& \sum_{s\in\cS}\sum_{u\in\cU}\sum_{(x_1,\cdots,x_L)\in\cX^L}U_{\bfs|\bfu,\bfx_{[L]}}(s|u,x_1,\cdots,x_L )B_i(s)\tau_{\vu, \vx_{i_1},\cdots,\vx_{i_L}}(u,x_1,\cdots,x_L)  \label{eqn:expected_cost} 
\end{align}
The above jamming cost does not differ much from the cost computed using the joint distribution $ \wt P_\bfu \wt P_{\bfx|\bfu} U_{\bfs|\bfu,\bfx_{[L]}} $, given the fact that $ \tau_{\vu,\vx_{i_1},\cdots,\vx_{i_L}} $ is $ \eta $-close to $ \wh P_{\bfx_1,\cdots,\bfx_L} $ which is in turn $ \eps $-close to $ \wt P_{\bfx_1,\cdots,\bfx_L}\in\cp^\tl(\wh P_\bfx) $. 
Indeed, 
\begin{align}
& \abs{\sum_{s\in\cS}\sum_{u\in\cU}\sum_{x_\cL\in\cX^L}U_{\bfs|\bfu,\bfx_{[L]}}(s|u,x_\cL )B_i(s)\tau_{\vu, \vx_{i_1},\cdots,\vx_{i_L}}(u,x_\cL) - \sum_{s\in\cS}\sum_{u\in\cU}\sum_{x_\cL\in\cX^L}U_{\bfs|\bfu,\bfx_{[L]}}(s|u,x_\cL)B_i(s)\wt P_{\bfu}(u) \wt P_{\bfx|\bfu}^\tl(x_\cL|u) } \notag \\
=& \abs{ \sum_{s\in\cS}\sum_{x_\cL\in\cX^L}\tau_{\vx_{i_1},\cdots,\vx_{i_L}}(x_\cL) U_{\bfs|\bfx_{[L]}}(s|x_\cL)B_i(s) - \sum_{s\in\cS}\sum_{x_\cL\in\cX^L}\wt P_{\bfx_{[L]}}(x_\cL)U_{\bfs|\bfx_{[L]}}(s|x_\cL)B_i(s) } \notag \\
=& \sum_{s\in\cS}\sum_{x_\cL\in\cX^L}\abs{U_{\bfs|\bfx_{[L]}}(s|x_\cL)}\abs{B_i(s)} \abs{ \tau_{\vx_{i_1},\cdots,\vx_{i_L}}(x_\cL) - \wt P_{\bfx_{[L]}}(x_\cL) }\notag \\
\le& \sum_{s\in\cS}\sum_{x_\cL\in\cX^L}\abs{B_i(s)} \abs{ \tau_{\vx_{i_1},\cdots,\vx_{i_L}}(x_\cL) - \wt P_{\bfx_{[L]}}(x_\cL) }\notag \\
\le& \card{\cS} B_i^*d\paren{ \tau_{\vx_{i_1},\cdots,\vx_{i_L}} , \wt P_{\bfx_{[L]}} } \label{eqn:def_b_i_star} \\
\le& \card{\cS} B_i^*(\eta + \eps). \label{eqn:cost_apx} 
\end{align}
In \Cref{eqn:def_b_i_star}, $ N_i^* $ is defined as $ B_i^* \coloneqq \max_{s\in\cS}\abs{ B_i(s) } $.

Combining \Cref{eqn:expected_cost}, \Cref{eqn:cost_apx} and \Cref{eqn:cost_assumption}, we get
\begin{align}
\expt{B_i(\vbfs)\condon\cL,\cG} \le& \cost_i((\wt P_\bfu,\wt P_{\bfx|\bfu}) ,U_{\bfs|\bfu,\bfx_{[L]}}) + \card{\cS} B_i^*(\eta + \eps) \le \Lambda_i - \delta_i + \card{\cS} B_i^*(\eta + \eps) \le \Lambda_i - \delta_i/2 < \Lambda_i. \notag 
\end{align}
In the last inequality we assume $ \eta\le\frac{\delta_i}{4\cardS B_i^*} $ and $ \eps\le\frac{\delta_i}{4\cardS B_i^*} $ for all $ i\in[\beta] $.

We then bound the (conditional) variance of $ B_i(\vbfs) $.
\begin{align}
\var{B_i(\vbfs)\condon\cL,\cG} =& \var{ \frac{1}{n}\sum_{j = 1}^nB_i(\vbfs(j))\condon\cL,\cG } \notag \\
=& \frac{1}{n^2}\sum_{j = 1}^n\var{ B_i(\vbfs(j))\condon\cL,\cG } \label{eqn:var_indep}  \\
\le& \frac{1}{n^2}\sum_{j = 1}^n\expt{B_i(\vbfs(j))^2\condon\cL,\cG} \notag \\
\le& \frac{1}{n^2}n(B_{i}^*)^2 \notag \\
=& (B_{i}^*)^2/n, \notag 
\end{align}
where \Cref{eqn:var_indep} is because each component of $\vbfs$ is sampled independently.
Now by Chebyshev's inequality (\Cref{lem:chebyshev}),
\begin{align}
\prob{ B_i(\vbfs)>\Lambda_i\condon\cL,\cG } \le \prob{ \abs{B_i(\vbfs) - \paren{\Lambda_i - \frac{\delta_i}{2}}} > \frac{\delta_i}{2} \condon\cL,\cG } \le \frac{(B_i^*)^2}{n(\delta_i/2)^2} = \frac{4(B_i^*)^2}{n\delta_i^2}. \notag 
\end{align}
Therefore, taking a union bound over all type constraints on the jamming sequence, we have
\begin{align}
\prob{ \tau_\bfs\notin\lambda_\bfs\condon\cL,\cG } = \prob{\exists i\in[\beta],\; B_i(\vbfs)>\Lambda_i\condon\cL,\cG }\le \sum_{i\in[\beta]}\prob{B_i(\vbfs)>\Lambda_i\condon\cL,\cG }\le \sum_{i\in[\beta]}\frac{4(B_i^*)^2}{n\delta_i^*}. \notag 
\end{align}

We then make several observations.
For any $ i\in[M] $ and $ i'\in\cL $, we have
\begin{align}
\exptover{\vbfs\sim\prod_{j=1}^n U_{\bfs|\bfu =\vu(j) ,\bfx_{[L]} = \vbfx_\cL(j)} }{W^{\tn}( \vy|\vx_i,\vbfs )\condon\cL,\cG} 
=& \exptover{ \substack{\vbfs(1)\sim U_{\bfs|\bfu =\vu(1) ,\bfx_{[L]} = \vbfx_{\cL}(1),}\\\cdots ,\\ \vbfs(n)\sim U_{\bfs|\bfu =\vu(n) ,\bfx_{[L]} = \vbfx_{\cL}(n)  } }}{\prod_{j = 1}^n W(\vy(j)|\vx_i(j),\vbfs(j))\condon\cL,\cG} \notag \\
=& \prod_{j = 1}^n\exptover{ \vbfs(j)\sim U_{\bfs|\bfu = \vu(j),\bfx_{[L]} = \vbfx_\cL(j)} }{W(\vy(j)|\vx_i(j),\vbfs(j))\condon\cL,\cG} \label{eqn:s_w_prod} \\
=& \prod_{j = 1}^n\sum_{s\in\cS} U_{\bfs|\bfu,\bfx_{[L]}}(s|\vu(j),\vx_\cL(j))W_{\bfy|\bfx,\bfs}(\vy(j)|\vx_i(j),s) \notag \\
=& \prod_{j= 1}^n \sum_{s\in\cS} U_{\bfs|\bfu,\bfx_{[L]}}(s|\vu(j),\vx_{(\cL\backslash\curbrkt{i'})\cup\curbrkt{i}}(j))W_{\bfy|\bfx,\bfs}(\vy(j)|\vx_{i'}(j),s) \label{eqn:u_symm} \\
=& \exptover{\vbfs\sim\prod_{j=1}^n U_{\bfs|\bfu =\vu(j) ,\bfx_{[L]} = \vbfx_\cL(j)} }{W^{\tn}(\vy|\vx_{i'},\vbfs)\condon(\cL\backslash\curbrkt{i'})\cup\curbrkt{i},\cG}. \label{eqn:rollback} 
\end{align}
In the above chain of equalities, \Cref{eqn:s_w_prod} follows since $ \vbfs $ is sampled from a product distribution and the channel is memoryless $ W_{\vbfy|\vbfx,\vbfs} = W^{\tn}_{\bfy|\bfx,\bfs} $; \Cref{eqn:u_symm}  follows since $ U_{\bfs|\bfu,\bfx_{[L]}} $ is a symmetrizing distribution (\Cref{def:def_symm_distr}); \Cref{eqn:rollback}  follows by rolling back the above chain of equalities.

Similar to $\cG$, define  
\begin{align}
\cG'\coloneqq\curbrkt{ d\paren{\tau_{\vbfx_{\cL'}},\wt P_{\bfx_{[L+1]}}}\le\eta' } \label{eqn:def_gprime}
\end{align}
where $ \cL'\sim\binom{[M]}{L+1} $ and $ \wt P_{\bfx_{[L+1]}} \coloneqq \wt P_{\bfx_1,\cdots,\bfx_{L+1}}\in\cp^{\otimes(L+1)}(\wh P_{\bfx}) $ satisfies 
\begin{align}
&  \sqrbrkt{\wt P_{\bfx_1,\cdots,\bfx_{L+1}}}_{\bfx_1,\cdots,\bfx_L} = \wt P_{\bfx_1,\cdots,\bfx_L} = \wt P_{\bfx_{[L]}}. \notag
\end{align} 
Define the set of good $L$-sized lists and good $(L+1)$-sized lists, respectively, as
\begin{align}
\cH\coloneqq\curbrkt{\cL\in\binom{[M]}{L}\colon \cG\holds}, \quad\cH'\coloneqq\curbrkt{\cL'\in\binom{[M]}{L+1}\colon\cG'\holds}. \notag
\end{align}
Note that $ \cG = \curbrkt{\cL\in\cH} $ and $ \cG'=\curbrkt{\cL'\in\cH'} $. 
Before proceeding, let us first prove the following lemma concerning the relation between $ \cG $ and $ \cG' $.
\begin{lemma}\label{lem:relation_g_gprime}
Suppose that $ \cG' $ is defined w.r.t. a certain distribution $ \wt P_{\bfx_{[L+1]}}\in\cp^{\otimes(L+1)}(\wh P_\bfx) $. 
Let $ \eta'>0 $ in the definition of $ \cG' $ (\Cref{eqn:def_gprime}) be such that $ \eta'\le\eta+\eps $ and $ \probover{\cL'}{\cG'} \ge c'\ge c(L+1) $. 
Then $ \probover{\cL}{\cG}\ge c $.
\end{lemma}
\begin{proof}
First, we claim that, if the distribution of  $ \bfx_{[L+1]} $ has $\cp$-decomposition $ \sum_{i\in[k]}\lambda_i P_{\bfx_i}^{\otimes(L+1)} $, then  the marginal distribution of   $ \bfx_{[L+1]\setminus\curbrkt{j}} $  is the same for every $ j\in[L+1] $ and has $\cp$-decomposition $ \sum_{i\in[k]}\lambda_i P_{\bfx_i}^\tl $.
Indeed, for any $ j\in[L+1] $, 
\begin{align}
\sqrbrkt{ \wt P_{\bfx_{[L+1]}} }_{\bfx_{[L+1]\setminus\curbrkt{j}}} (x_{[L+1]\setminus\curbrkt{j}}) =& \sum_{x_j\in\cX} \sum_{i\in[k]} \lambda_i P_{\bfx_i}^{\otimes(L+1)} (x_{[L+1]}) \notag \\
=& \sum_{i\in[k]}\lambda_i\sum_{x_j\in\cX} P_{\bfx_i}(x_1)\cdots P_{\bfx_i}(x_{[L+1]}) \notag \\
=& \sum_{i\in[k]}\lambda_i P_{\bfx_i}^\tl (x_{[L+1]\setminus\curbrkt{j}})\sum_{x_j\in\cX}P_{\bfx_i}(x_j)  \notag \\
=& \sum_{i\in[k]} \lambda_i P_{\bfx_i}^\tl(x_{[L+1]\setminus\curbrkt{j}})  \notag \\
=& \wt P_{\bfx_{[L]}}(x_{[L+1]\setminus\curbrkt{j}}) . \label{eqn:relation_pl_plprime}
\end{align}

Recall the definition of $ \cG $ and $ \cG' $ in \Cref{eqn:def_g} and \Cref{eqn:def_gprime}, respectively. 
We claim that if $ \cG' $ holds for some $ \cL'\in\binom{[M]}{L+1} $ then for any $ j\in\cL' $, $ \cG $ holds, where $ \cL = \cL'\setminus\curbrkt{j} $. 
Indeed, 
\begin{align}
 d\paren{\tau_{\vx_\cL}, \wt P_{\bfx_{[L]}}} =& \sum_{x_\cL\in\cX^L} \abs{ \tau_{\vx_\cL}(x_\cL) - \wt P_{\bfx_{[L]}}(x_\cL) } \notag \\
=& \sum_{x_\cL\in\cX^L}\abs{ \sum_{x_j\in\cX} \tau_{\vx_{\cL'}}(x_{\cL'}) - \wt P_{\bfx_{[L+1]}}(x_{\cL'}) } \label{eqn:apply_relation_pl_plprime} \\
\le& \sum_{x_{\cL'}\in\cX^{L+1}}\abs{\tau_{\vx_{\cL'}}(x_{\cL'}) - \wt P_{\bfx_{[L+1]}}(x_{\cL'})} = d\paren{ \tau_{\vx_{\cL'}}, \wt P_{\bfx_{[L+1]}} } \le\eta' \notag \\
 \le& \eta+\eps . \label{eqn:apply_assumption_etaprime}
\end{align}
\Cref{eqn:apply_relation_pl_plprime} follows from the previous claim given by \Cref{eqn:relation_pl_plprime}. 
\Cref{eqn:apply_assumption_etaprime} is by one of the assumptions $ \eta'\le\eta+\eps $.

By another assumption,
\begin{align}
\probover{\cL'}{\cG'} =& \probover{ \cL' }{ \cL'\in\cH' } = \frac{\card{\cH'}}{\binom{M}{L+1}} = \frac{1}{\binom{M}{L+1}}{\card{\curbrkt{\cL'\in\binom{[M]}{L+1}\colon d\paren{\tau_{\vx_{\cL'}}, \wt P_{\bfx_{[L+1]}}}\le\eta'}}}\ge c'. \label{eqn:bound_gprime} 
\end{align}
Now we bound $ \probover{\cL}{\cG} $.
\begin{align}
\probover{\cL}{\cG} =& \frac{1}{\binom{M}{L}}\card{\curbrkt{\cL\in\binom{[M]}{L}\colon d\paren{\tau_{\vx_\cL}, \wt P_{\bfx_{[L]}}}\le\eta + \eps }} \notag \\
\ge& \frac{1}{\binom{M}{L}}\card{\curbrkt{\cL\in\binom{[M]}{L}\colon \exists\cL'\in\cH',\;\cL\in\binom{\cL'}{L},\;d\paren{\tau_{\vx_\cL},\wt P_{\bfx_{[L]}}}\le\eta+\eps  }} \notag \\
\ge& \frac{1}{\binom{M}{L}}\card{\curbrkt{\cL\in\binom{[M]}{L}\colon \exists\cL'\in\cH',\;\cL\in\binom{\cL'}{L}  }} \label{eqn:gprime_imiplies_g} \\
\ge& \frac{1}{\binom{M}{L}}\frac{\card{\cH'}}{M-L} \label{eqn:lprime_generate_l} \\
\ge& \frac{c'\binom{M}{L+1}}{\binom{M}{L}(M-L)} \notag \\
=& \frac{c'}{L+1} \label{eqn:apply_bound_gprime} \\
\ge& c \label{eqn:relation_c_cprime}
\end{align}
In \Cref{eqn:gprime_imiplies_g}, given the condition $ \cL'\in\cH' $, we could drop the condition $ d\paren{\tau_{\vx_\cL}, \wt P_{\bfx_{[L]}}}\le\eta+\eps $ since $ \cG' $ for $ \cL' $ implies $ \cG $ for $ \cL\in\binom{\cL'}{L} $, as we argued just now. 
\Cref{eqn:lprime_generate_l} follows from the following fact by taking $ \sL \coloneqq \cH' $. 
\begin{fact}\label{fact:lprime_generate_l}
Let $ \sL'\subseteq\binom{[M]}{L+1} $ be a collection of $ (L+1) $-sized lists.
Let $ \sL\coloneqq\curbrkt{\cL\in\binom{[M]}{L}\colon \exists\cL'\in\sL',\;\cL\in\binom{\cL'}{L}} $.
Then $ \card{\sL}(M-L)\ge\card{\sL'} $. 
\end{fact}
\Cref{eqn:apply_bound_gprime} follows from \Cref{eqn:bound_gprime}.
\Cref{eqn:relation_c_cprime} is by the assumption $ \frac{c'}{L+1}\ge c $. 
\end{proof}

We make another observation. 
For any $ \cL'\in\cH' $ and $ i_0\in\cL' $,  
we have
\begin{align}
\sum_{i\in\cL'}\expt{P_{\e,\avg}(i,\vbfs)\condon\cL'\backslash\curbrkt{i},\cL'\backslash\curbrkt{i}\in\cH} =& \sum_{i\in\cL'}\paren{ 1 - \sum_{\vy\in\cY^n\colon\psi(\vy)\ni i}\expt{ W^{\tn}(\vy|\vx_i,\vbfs)\condon\cL'\backslash\curbrkt{i},\cL'\backslash\curbrkt{i}\in\cH } } \notag \\
=& \sum_{i\in\cL'} \paren{ 1 - \sum_{\vy\in\cY^n\colon\psi(\vy)\ni i}\expt{ W^{\tn}(\vy|\vx_{i_0},\vbfs)\condon\cL'\backslash\curbrkt{{i_0}},\cL'\backslash\curbrkt{{i_0}}\in\cH } } \label{eqn:perm_symm} \\
=& (L+1) - \sum_{i\in\cL'\cap\psi(\vy)}\sum_{\vy\in\cY^n}\expt{ W^{\tn}(\vy|\vx_{i_0},\vbfs)\condon\cL'\backslash\curbrkt{{i_0}},\cL'\backslash\curbrkt{{i_0}}\in\cH } \notag \\
=& (L+1) - \sum_{i\in\cL'\cap\psi(\vy)}1 \label{eqn:inner_one} \\
\ge& (L+1)-L  \label{eqn:list_size_l} \\
=&  1. \label{eqn:sum_over_i_bd}
\end{align}
\Cref{eqn:perm_symm} is by \Cref{eqn:rollback}; \Cref{eqn:inner_one} follows since the inner summation equals 1; \Cref{eqn:list_size_l} follows because $ \card{\psi(\vy)}\le L $ for any $ \vy\in\cY^n $. 

Given the above observations, we can lower bound $ \expt{\expt{P_{\e,\avg}(\vbfs)\condon\cL,\cG}} $ as follows. 
\begin{align}
\expt{\expt{P_{\e,\avg}(\vbfs)}\condon\cL,\cG} =& \frac{1}{\binom{M}{L}}\frac{1}{M}\sum_{\cL\in\binom{[M]}{L}}\sum_{i\in[M]}\expt{P_{\e,\avg}(i,\vbfs)\condon\cL} \indicator{\cL\in\cH} \notag \\
\ge& \frac{1}{\binom{M}{L}}\frac{1}{M}\sum_{\cL\in\binom{[M]}{L}}\sum_{i\in[M]\colon \cL\cup\curbrkt{i}\in\cH'}\expt{P_{\e,\avg}(i,\vbfs)\condon\cL} \indicator{\cL\in\cH} \notag \\
\ge& \frac{1}{\binom{M}{L}}\frac{1}{M}\sum_{\cL'\in\cH'}\sum_{i\in\cL'}\expt{P_{\e,\avg}(i,\vbfs)\condon\cL'\backslash\curbrkt{i}} \indicator{\cL'\backslash\curbrkt{i}\in\cH} \label{eqn:drop_ind} \\
=& \frac{1}{\binom{M}{L}}\frac{1}{M}\sum_{\cL'\in\cH'}\sum_{i\in\cL'}\expt{P_{\e,\avg}(i,\vbfs)\condon\cL'\backslash\curbrkt{i}} \label{eqn:drop_condition_l} \\
\ge& \frac{1}{\binom{M}{L}}\frac{1}{M}\card{\cH'} \label{eqn:apply_sum_over_i_bd} \\
\ge& \frac{c'\binom{M}{L+1}}{\binom{M}{L}M} \label{eqn:apply_bd_on_gprime} \\
=& \frac{M-L}{L+1}\frac{c'}{M} \notag \\
=& c'\paren{\frac{1}{L+1} - \frac{L}{(L+1)M} } \notag \\
\ge& \frac{c'}{2(L+1)}. \label{eqn:bound_pe_given}
\end{align}
Since $ \cL'\backslash\curbrkt{i}\in\cH $ is always true given the choice of $ \cL' $ and $ i $, we can drop the indicator in \Cref{eqn:drop_ind}. 
In \Cref{eqn:drop_condition_l}, we drop the indicator of the event $ \cL'\setminus\curbrkt{i}\in\cH $ since $ \cG' = \curbrkt{\cL'\in\cH'} $ automatically implies $ \cG = \curbrkt{\cL'\setminus\curbrkt{i}\in\cH} $ by \Cref{lem:relation_g_gprime}.
\Cref{eqn:apply_sum_over_i_bd} follows from \Cref{eqn:sum_over_i_bd}. 
\Cref{eqn:apply_bd_on_gprime} follows from a bound similar  to $ \prob{\cG}\ge c $. 
Indeed, by replacing $ L $ with $ L+1 $, we have 
\begin{align}
\probover{\cL'\sim\binom{[M]}{L+1}}{\cG'} = \probover{\cL'\sim\binom{[M]}{L+1}}{\cL'\in\cH'} = \frac{\card{\cH'}}{\binom{M}{L+1}} \ge c' \notag
\end{align}
for some constant $ c'>0 $. 
In \Cref{eqn:bound_pe_given}, we assume $ M\ge2L $.

Finally, putting \Cref{eqn:bound_pe_given} back to \Cref{eqn:pe_decomp}, we have
\begin{align}
\exptover{\cL,\vbfs}{P_{\e,\avg}(\vbfs)} \ge& \frac{cc'}{2(L+1)}. \label{eqn:pe_final_bound}
\end{align}
\Cref{eqn:pe_final_bound} finishes the proof of the symmetrization converse (\Cref{thm:symm_converse}).

\section{Codeword selection}
\label[app]{app:cw_select}
Fix rate $R$ and a time-sharing sequence $ \vu\in\cU^n $ of type $ P_\bfu $.
We sample $ M= L2^{nR} $ codewords $ \vbfx_1,\cdots,\vbfx_M $ independently using the following distribution.
For each $ i\in[M] $, let $ \vbfx_i^{(j)} $ denote the $j$-th ($j\in\cU$) subvector  of $ \vbfx_i $, i.e., 
$\vbfx_i^{(j)} \coloneqq( \vbfx_i(k_1),\cdots,\vbfx_i(k_{nP_\bfu(j)}) )\in\cX^{nP_\bfu(j)} $
where $ (k_1,\cdots,k_{nP_\bfu(j)}) $ satisfies $\vbfu(k_{\ell}) = j $ for all $ \ell\in[nP_\bfu(j)] $. 
We sample $\vbfx_i^{(j)}$  uniformly from all  $ \cX^{nP_\bfu(j)} $-valued vectors of type $ P_{\bfx|\bfu = j} $. 

We will show that a codebook $ \cC $ sampled as above simultaneously satisfies all desired properties in \Cref{lem:cw_select} with high probability.
To this end, we need the following lemma whose proof appeared in \cite{csiszar-narayan-it1988-obliviousavc}.
\begin{lemma}
\label{lem:cn_conc}
Let $ X_1,\cdots,X_M $ be (not necessarily identically distributed, possibly dependent) random variables.
For all $ i\in[M] $, let $ f_i(X_1,\cdots,X_i) $ be  a function such that $ f_i\in[0,1] $.
If $ \expt{ f_i(X_1,\cdots,X_i)\condon X_1,\cdots,X_{i-1} }\le a $ a.s. for all $ i\in[M] $,
then 
\begin{align}
\prob{ \frac{1}{M}\sum_{i = 1}^Mf_i(X_1,\cdots,X_i) > t }\le& 2^{ -M(t - a\log e) }. \notag 
\end{align}
\end{lemma}

\begin{proof}[Proof of \Cref{lem:cw_select}.]
Fix $ \vu\in\cU^n ,\vx\in\cX^n,\vs\in\cS^n,P_{\bfu}\in\Delta(\cU), P_{\bfx|\bfu}\in\Delta(\cX|\cU), P_{\bfu,\bfx,\bfx_{[L]},\bfs}\in\Delta(\cU\times\cX\times\cX^L\times\cS) $.
To avoid trivialities, assume $ P_{\bfu,\bfx,\bfs} = \tau_{\vu,\vx,\vs}, P_{\bfu,\bfx} = P_{\bfu,\bfx_k} = P_\bfu P_{\bfx|\bfu} $ for all $ k\in[L] $. 
Let $ \vbfx_1,\cdots,\vbfx_M $ be i.i.d. random vectors following the distribution specified in the beginning of \Cref{app:cw_select}. 
\paragraph{Proof of \Cref{eqn:property_3}}
Fix any $ k\in[L] $. 
Let $ f_i(\vbfx_1,\cdots,\vbfx_i) = \indicator{ \tau_{\vu,\vx,\vbfx_i,\vs} = P_{\bfu,\bfx,\bfx_k,\bfs} }  ,\;i\in[M] $. 
Note that each $ f_i $ in fact only depends on $ \vbfx_i $. 
\begin{align}
a =& \expt{ f_i(\vbfx_1,\cdots,\vbfx_i)\condon\vbfx_1,\cdots,\vbfx_{i - 1} }\notag \\
=& \expt{ f_i(\vbfx_1,\cdots,\vbfx_i) } \notag \\
=& \prob{ \tau_{\vu,\vx,\vbfx_i,\vs} = P_{\bfu,\bfx,\bfx_k,\bfs} } \notag \\
=& \frac{ \card{\curbrkt{ \vx'\in\cX^n\colon \tau_{\vu,\vx,\vx',\vs} = P_{\bfu,\bfx,\bfx_k,\bfs} }} }{\card{\curbrkt{ \vx'\in\cX^n\colon \tau_{\vu,\vx' = P_{\bfu,\bfx}} }}} \notag \\
\doteq& \frac{2^{nH(\bfx_k|\bfu,\bfx,\bfs)}}{2^{nH(\bfx|\bfu)}} \notag \\
=& 2^{-n(H(\bfx_k|\bfu) -  H(\bfx_k|\bfu,\bfx,\bfs))} \label{eqn:eq_ent} \\
=& 2^{-nI(\bfx_k;\bfx,\bfs|\bfu)}, \notag 
\end{align}
where \Cref{eqn:eq_ent} follows since $ P_{\bfu,\bfx} = P_{\bfu,\bfx_k} $. 

Let $ t = \frac{1}{M}2^{n\paren{\sqrbrkt{R - I(\bfx_k;\bfx,\bfs|\bfu)}^++\eps}} $. 
Under these choices of parameters ($t$ and $a$),
\begin{align}
M(t - a\log e) \doteq& 2^{n\paren{\sqrbrkt{R - I(\bfx_k;\bfx,\bfs|\bfu)}^++\eps}} - L2^{nR}\log e2^{-I(\bfx_k;\bfx,\bfs|\bfu)} \notag \\
=& \begin{cases}
2^{n(R - I(\bfx_k;\bfx,\bfs|\bfu) + \eps)} - 2^{R - I(\bfx_k;\bfx,\bfs|\bfu)}L\log e, & R\ge I(\bfx_k;\bfx,\bfs|\bfu) \\
2^{n\eps} - 2^{n(R - I(\bfx_k;\bfx,\bfs|\bfu))}, & R< I(\bfx_k;\bfx,\bfs|\bfu) 
\end{cases} \notag \\
\dotge& 2^{n\eps} \notag
\end{align}
for sufficiently large $n$. 

Now by \Cref{lem:cn_conc}, 
\begin{align}
\prob{ \sum_{i = 1}^Mf_i(\vbfx_1,\cdots,\vbfx_i)>t }  =&\prob{ \card{\curbrkt{i\in[M]\colon \tau_{\vu,\vx,\vbfx_i,\vs} = P_{\bfu,\bfx,\bfx_k,\bfs}}}>2^{n\paren{\sqrbrkt{R - I(\bfx_k;\bfx,\bfs|\bfu)}^++\eps}} } \notag \\
\le& 2^{-M(a - t\log e)} \notag \\
\le& 2^{-2^{n\eps}}. \label{eqn:property_3_proved} 
\end{align}



\paragraph{Proof of \Cref{eqn:property_1}}
Following the same procedures as in the previous paragraph, we have
\begin{align}
\prob{ \card{\curbrkt{i\in[M]\colon \tau_{\vu,\vbfx_i,\vs} = P_{\bfu,\bfx,\bfs}}} > 2^{n\paren{\sqrbrkt{R - I(\bfx;\bfs|\bfu)}^++\paren{\eps/2 + \frac{1}{n}\log L}}} } \le&
2^{-2^{n\paren{\eps/2 + \frac{1}{n}\log L}}} = 2^{-L2^{n\eps/2}} . \notag 
\end{align}
Recall that $ R\ge \eps $.
Under the assumption that $ I(\bfx;\bfs|\bfu)\ge\eps $, we have
\begin{align}
\sqrbrkt{R - I(\bfx;\bfs|\bfu)}^+=& R - \min\curbrkt{R, I(\bfx;\bfs|\bfu)} \le R-\eps. \label{eqn:positive_part} 
\end{align}
Hence
\begin{align}
2^{-L2^{n\eps/2}}\ge&
\prob{ \card{\curbrkt{i\in[M]\colon \tau_{\vu,\vbfx_i,\vs} = P_{\bfu,\bfx,\bfs}}}>L2^{n\paren{\sqrbrkt{R - I(\bfx;\bfs|\bfu)}^+ + \eps/2} } } \notag \\
\ge& \prob{ \card{\curbrkt{i\in[M]\colon \tau_{\vu,\vbfx_i,\vs} = P_{\bfu,\bfx,\bfs}}}> L2^{n(R - \eps/2 )} } \label{eqn:simp_positive_part} \\
=& \prob{\frac{1}{M} \card{\curbrkt{i\in[M]\colon \tau_{\vu,\vbfx_i,\vs} = P_{\bfu,\bfx,\bfs}}}\dotg 2^{-n\eps/2 } }, \label{eqn:bound_property_1} 
\end{align}
where \Cref{eqn:simp_positive_part} is by \Cref{eqn:positive_part}.

\paragraph{Proof of \Cref{eqn:property_33}}
First observe that 
\begin{align}
\curbrkt{\cL\in\binom{[M]}{L}\colon \tau_{\vu,\vx,\vbfx_\cL,\vs} = P_{\bfu,\bfx,\bfx_{[L]},\bfs}}
\subseteq& \curbrkt{(i_1,\cdots,i_L)\in\binom{[M]}{L}\colon \tau_{\vu,\vx,\vbfx_{i_1},\vs} = P_{\bfu,\bfx,\bfx_1,\bfs},\cdots,\tau_{\vu,\vx,\bfx_{i_L},\bfs} = P_{\bfu,\bfx,\bfx_L,\bfs}} \notag \\
\subseteq& \curbrkt{(i_1,\cdots,i_L)\in[M]^L\colon \tau_{\vu,\vx,\vbfx_{i_1},\vs} = P_{\bfu,\bfx,\bfx_1,\bfs},\cdots,\tau_{\vu,\vx,\bfx_{i_L},\bfs} = P_{\bfu,\bfx,\bfx_L,\bfs}} \notag \\
=& \bigtimes_{k = 1}^L\curbrkt{i\in[M]\colon \tau_{\vu,\vx,\vbfx_{i},\vs} = P_{\bfu,\bfx,\bfx_k,\bfs}}. \notag 
\end{align}
Therefore,
\begin{align}
\prob{ \card{\curbrkt{\cL\in\binom{[M]}{L}\colon \tau_{\vu,\vx,\vbfx_\cL,\vs} = P_{\bfu,\bfx,\bfx_{[L]},\bfs}}}>2^{n\eps} } \le&
\prob{\prod_{k\in[L]}\card{\curbrkt{ i\in[M]\colon \tau_{\vu,\vx,\vbfx_{i},\vs} = P_{\bfu,\bfx,\bfx_k,\bfs} }} >2^{n\eps}} \notag \\
\le& \prob{ \exists k\in[L],\;\card{\curbrkt{ i\in[M]\colon \tau_{\vu,\vx,\vbfx_{i},\vs} = P_{\bfu,\bfx,\bfx_k,\bfs} }} >2^{n\eps/L} } \notag \\
\le& \sum_{k\in[L]}\prob{\card{\curbrkt{ i\in[M]\colon \tau_{\vu,\vx,\vbfx_{i},\vs} = P_{\bfu,\bfx,\bfx_k,\bfs} }} >2^{n\eps/L}} \notag \\
\le& L2^{-2^{n\eps/L}}, \label{eqn:three_prime}
\end{align}
where the last inequality follows from \Cref{eqn:property_3_proved} since $ R<\min_{k\in[L]}I(\bfx_k;\bfs|\bfu)\le\min_{k\in[L]}I(\bfx_k;\bfx,\bfs|\bfu) $.

\paragraph{Proof of \Cref{eqn:property_22}}
We first make several definitions. 
\begin{align}
\sL_i' \coloneqq& \curbrkt{ \cL\in\binom{[i-1]}{L}\colon \tau_{\vu,\vbfx_\cL,\vs} = P_{\bfu,\bfx_{[L]},\bfs} }, \notag \\
\sL_i \coloneqq& \begin{cases}
\sL_i',& |\sL_i'|\le2^{n\eps_1} \\
\emptyset, & |\sL_i'|> 2^{n\eps_1}
\end{cases}, \notag \\
f_i(\vbfx_1,\cdots,\vbfx_i) \coloneqq& \indicator{ \exists\cL\in\sL_i,\;\tau_{\vu,\vbfx_i,\vbfx_\cL,\vs} = P_{\bfu,\bfx,\bfx_{[L]},\bfs} }, \notag 
\end{align}
where $ \eps_1>0 $ is a small constant to be specified later. 

Observe that if
\begin{align}
\max_{i\in[M]}|\sL'_i|\le& \card{\curbrkt{ \cL\in\binom{[M]}{L}\colon\tau_{\vu,\vbfx_\cL,\vs} = P_{\bfu,\bfx_{[L]},\bfs} }}\le 2^{n\eps_1}, \label{eqn:condition_l_lprime}
\end{align}
then $ \sL_i = \sL_i' $ for all $ i\in[M] $.
Furthermore, under the condition in \Cref{eqn:condition_l_lprime}, we have
\begin{align}
\sum_{i\in[M]}f_i(\vbfx_1,\cdots,\vbfx_i) =& \sum_{i\in[M]}\indicator{ \exists\cL\in\sL_i,\;\tau_{\vu,\vbfx_i,\vbfx_\cL,\vs} = P_{\bfu,\bfx,\bfx_{[L]},\bfs} } \notag \\
=& \sum_{i\in[M]}\indicator{ \exists\cL\in\sL_i',\;\tau_{\vu,\vbfx_i,\vbfx_\cL,\vs} = P_{\bfu,\bfx,\bfx_{[L]},\bfs} } \notag \\
=&\sum_{i\in[M]} \indicator{ \exists\cL\in\binom{[i-1]}{L},\;\tau_{\vu,\vbfx_\cL,\vs} = P_{\bfu,\bfx_{[L]},\bfs},\;\tau_{\vu,\vbfx_i,\vbfx_\cL,\vbfs} = P_{\bfu,\bfx,\bfx_{[L]},\bfs} } \notag \\
=& \sum_{i\in[M]}\indicator{ \exists\cL\in\binom{[i-1]}{L},\;\tau_{\vu,\vbfx_i,\vbfx_\cL,\vbfs} = P_{\bfu,\bfx,\bfx_{[L]},\bfs} } \notag \\
=& \card{\curbrkt{ i\in[M]\colon \exists\cL\in\binom{[i-1]}{L},\; \tau_{\vu,\vbfx_i,\vbfx_\cL,\vs} = P_{\bfu,\bfx,\bfx_{[L]},\bfs} }}. \label{eqn:card_equals_sum}
\end{align}

We can bound the probability that both sides of \Cref{eqn:card_equals_sum} differ 
using \Cref{eqn:three_prime}. 
\begin{align}
&\prob{ \card{\curbrkt{ i\in[M]\colon \exists\cL\in\binom{[i-1]}{L},\; \tau_{\vu,\vbfx_i,\vbfx_\cL,\vs} = P_{\bfu,\bfx,\bfx_{[L]},\bfs} }} \ne \sum_{i\in[M]}f_i(\vbfx_1,\cdots,\vbfx_i) } \label{eqn:prob_sum_diff_card_term} \\
=& \prob{\exists i\in[M],\;\sL_i\ne\sL_i'} \notag \\
=&\prob{\exists i\in[M],\;\card{\curbrkt{ \cL\in\binom{[i-1]}{L}\colon \tau_{\vu,\vbfx_\cL,\vs} = P_{\bfu,\bfx_{[L]},\bfs} }}>2^{n\eps_1}} \notag \\
\le&\prob{\exists i\in[M],\;\card{\curbrkt{ \cL\in\binom{[M]}{L}\colon \tau_{\vu,\vbfx_\cL,\vs} = P_{\bfu,\bfx_{[L]},\bfs} }}>2^{n\eps_1}} \notag \\
\le& L2^{-2^{n\eps_1/L}}. \label{eqn:prob_sum_diff_card}
\end{align}

We then concentrate $\sum_{i\in[M]}f_i(\vbfx_1,\cdots,\vbfx_i) $ using \Cref{lem:cn_conc}. 
To this end, let us compute 
\begin{align}
\expt{f_i(\bfx_1,\cdots,\bfx_i)\condon\bfx_1,\cdots,\bfx_{i-1}} =& \prob{ \exists\cL\in\sL_i,\;\tau_{\vu,\vbfx_i,\vbfx_\cL,\vs} = P_{\bfu,\bfx,\bfx_{[L]},\bfs}\condon\vbfx_1,\cdots,\vbfx_{i-1} } \notag \\
\stackrel{\cdot}{\le}& |\sL_i|2^{nH(\bfx|\bfu,\bfx_{[L]},\bfs)}/2^{nH(\bfx|\bfu)}\notag \\
\le& 2^{-n(I(\bfx;\bfx_{[L]},\bfs|\bfu) - \eps_1)} \notag \\
\le& 2^{-n(\eps - \eps/4)} = 2^{-\frac{3}{4}n\eps}, \label{eqn:fix_eps1}
\end{align}
where \Cref{eqn:fix_eps1} follows from the assumption $ I(\bfx;\bfx_{[L]},\bfs|\bfu)\ge\eps $ and the choice of parameter $ \eps_1 = \eps/4 $. 
By \Cref{lem:cn_conc}, we have 
\begin{align}
\prob{\frac{1}{M}\sum_{i\in[M]}f_i(\bfx_1,\cdots,\bfx_i)> 2^{-\frac{2}{3}n\eps}} \le& 2^{-M\paren{ 2^{-\frac{2}{3}n\eps} - 2^{-\frac{3}{4}n\eps}\log e }}\le 2^{-\frac{M}{2}2^{-\frac{2}{3}n\eps}}\le 2^{-\frac{L}{2}2^{n\eps}2^{-\frac{2}{3}n\eps}}= 2^{-\frac{L}{2}2^{n\eps/3}}, \label{eqn:sum_bound} 
\end{align}
where we used the assumption $R = \frac{1}{n}\log\frac{M}{L}\ge\eps$. 

Combining \Cref{eqn:prob_sum_diff_card} and \Cref{eqn:sum_bound}, we get
\begin{align}
\prob{\frac{1}{M}\card{\curbrkt{i\in[M]\colon\exists\cL\in\binom{[i-1]}{L},\;\tau_{\vu,\vbfx_i,\vbfx_\cL,\vs} = P_{\bfu,\bfx,\bfx_{[L]},\bfs}}} > 2^{-\frac{2}{3}n\eps}}
\le&  \Cref{eqn:prob_sum_diff_card_term} + \mathrm{LHS\ of\ } \Cref{eqn:sum_bound} \notag \\
\le& L2^{-2^{\frac{n\eps}{4L}}} +   2^{-\frac{L}{2}2^{n\eps/3}} \le (L+1)2^{-2^{\frac{n\eps}{4L}}}. \label{eqn:bound_unperm}
\end{align}

Take any $ \pi\in S_M $. 
\Cref{eqn:bound_unperm} remains true when $ \binom{[i-1]}{L} $ is replaced by
\begin{align}
\pi\binom{[i-1]}{L}\coloneqq& \binom{\pi^{-1}[\pi(i) - 1]}{L} =\curbrkt{L\in\binom{[M]\setminus\curbrkt{i}}{L}\colon\forall j\in\cL,\;\pi(j)\le\pi(i) - 1}. \notag 
\end{align}
Indeed, the proof follows by replacing $ \vbfx_i $ with $ \vbfx_{\pi(i)} $. 
Let $ \pi_1,\cdots,\pi_K $ be permutations given by the following lemma.
\begin{lemma}[Lemma A2, \cite{hughes-1997-list-avc}]
For  $n\ge\log(2L)$, there exist $ K $ ($ K\le n(L+1)^2(\log\cardX+1) $) permutations $ \pi_1,\cdots,\pi_K\in S_M $ ($ M = L2^{nR} $) such that for all $i\in[M]$,
\begin{align}
\binom{[M]\setminus\curbrkt{i}}{L} =& \bigcup_{k\in[K]} \pi_k\binom{[i-1]}{L}. \notag 
\end{align}
\end{lemma}

Note that for sufficiently large $n$, $ K\le2^{n\eps/6} $. 
We are finally ready to prove \Cref{eqn:property_22}.
\begin{align}
& \prob{\frac{1}{M}\card{\curbrkt{i\in[M]\colon \exists\cL\in\binom{[M]\setminus\curbrkt{i}}{L},\;\tau_{\vu,\vbfx_i,\vbfx_\cL,\vs} = P_{\bfu,\bfx,\bfx_{[L]},\bfs}}}> 2^{-n\eps/2}} \notag \\
=& \prob{\frac{1}{M}\card{\curbrkt{i\in[M]\colon \exists\cL\in\bigcup_{k\in[K]}\pi_k\binom{[i-1]}{L},\;\tau_{\vu,\vbfx_i,\vbfx_\cL,\vs} = P_{\bfu,\bfx,\bfx_{[L]},\bfs}}}> 2^{-n\eps/2}} \notag \\
=& \prob{\frac{1}{M}\card{\curbrkt{i\in[M]\colon\exists k\in[K],\; \exists\cL\in\pi_k\binom{[i-1]}{L},\;\tau_{\vu,\vbfx_i,\vbfx_\cL,\vs} = P_{\bfu,\bfx,\bfx_{[L]},\bfs}}}> 2^{-n\eps/2}} \notag \\
=&\prob{\frac{1}{M}\card{\bigcup_{k\in[K]}\curbrkt{i\in[M]\colon \exists\cL\in\pi_k\binom{[i-1]}{L},\;\tau_{\vu,\vbfx_i,\vbfx_\cL,\vs} = P_{\bfu,\bfx,\bfx_{[L]},\bfs}}}> 2^{-n\eps/2}} \notag \\
\le& \prob{\frac{1}{M}\sum_{k\in[K]}\card{\curbrkt{i\in[M]\colon \exists\cL\in\pi_k\binom{[i-1]}{L},\;\tau_{\vu,\vbfx_i,\vbfx_\cL,\vs} = P_{\bfu,\bfx,\bfx_{[L]},\bfs}}}> 2^{-n\eps/2}} \notag \\
\le& \sum_{k\in[K]}\prob{\frac{1}{M}\card{\curbrkt{i\in[M]\colon \exists\cL\in\pi_k\binom{[i-1]}{L},\;\tau_{\vu,\vbfx_i,\vbfx_\cL,\vs} = P_{\bfu,\bfx,\bfx_{[L]},\bfs}}}> 2^{-n\eps/2}/K} \notag \\
\le& 2^{n\eps/6}\prob{\frac{1}{M}\card{\curbrkt{i\in[M]\colon \exists\cL\in\binom{[i-1]}{L},\;\tau_{\vu,\vbfx_i,\vbfx_\cL,\vs} = P_{\bfu,\bfx,\bfx_{[L]},\bfs}}}> 2^{-\frac{2}{3}n\eps}} \notag \\
\le& (L+1)2^{-2^{\frac{n\eps}{4L}}+{n\eps/6}}. \label{eqn:to_be_rep}
\end{align}
The last inequality is by \Cref{eqn:bound_unperm}.

\paragraph{Proof of \Cref{eqn:property_2}}
Under the condition $ I(\bfx;\bfx_k,\bfs|\bfu) - \sqrbrkt{R-I(\bfx_k;\bfs|\bfu)}^+\ge\eps $,  \Cref{eqn:property_2} follows by setting $L=1$ and $ \eps_1 = \sqrbrkt{R -I(\bfx_k;\bfs|\bfu)}^++\eps/4 $. 
Specifically, by repeating the proof of \Cref{eqn:property_22}, we have  the following bound similar to \Cref{eqn:to_be_rep}.
\begin{align}
\prob{\frac{1}{M}\card{\curbrkt{i\in[M]\colon\exists j\in[M]\setminus\curbrkt{i},\;\tau_{\vu,\vbfx_i,\vbfx_j,\vs} = P_{\bfu,\bfx,\bfx_k,\bfs}}}> 2^{-n\eps/2}} \le& 2\cdot2^{-\frac{1}{2}2^{n\eps/4}+n\eps/6}. \label{eqn:bound_property_2} 
\end{align}

\paragraph{Finishing up the proof of \Cref{lem:cw_select}}
Finally, \Cref{lem:cw_select} follows by taking a union bound over $ \vx\in\cX^n $ such that $ \tau_{\vu,\vx} = P_{\bfu,\bfx} $, $ \vs\in\cS^n $ and joint types $ P_{\bfu,\bfx,\bfx_{[L]},\bfs} $. 
There are in total exponentially many of them and the concentration bounds in 
\Cref{eqn:property_3_proved},
\Cref{eqn:bound_property_1},
\Cref{eqn:three_prime},
\Cref{eqn:to_be_rep}
are all doubly exponentially small. 
We thus have shown that with probability doubly exponentially close to 1, a random code consisting of codewords $ \vbfx_1,\cdots,\vbfx_M\in\cX^n $ each of type $ \tau_{\vu,\vbfx_i} = P_{\bfu,\bfx} $ ($i\in[M]$) simultaneously satisfies \Cref{eqn:property_1}, \Cref{eqn:property_2}, \Cref{eqn:property_3}, \Cref{eqn:property_22} and \Cref{eqn:property_33}  for all $ \vx,\vs,P_{\bfx,\bfx_{[L]},\bfs} $. 
\end{proof}

\section{Unambiguity of decoding rules (proof of \Cref{lem:unambiguity})}
\label[app]{app:unambiguity_pf}
The proof is by contradiction. 
Suppose there does exist a joint distribution $ P_{\bfu,\bfx_{[L+1]},\bfs_{[L+1]},\bfy} $ satisfying the conditions in \Cref{eqn:unambiguity_cond}.
Observe that for any $ i\in[L+1] $, 
\begin{align}
2\eta \ge& D\paren{P_{\bfu,\bfx_i,\bfs_i,\bfy}\middle\| P_\bfu P_{\bfx_i|\bfu}P_{\bfs_i}W_{\bfy|\bfx,\bfs}} + I\paren{\bfx_i,\bfy;\bfx_{[L+1]\setminus\curbrkt{i}}\middle|\bfu,\bfs_i} = \kl{P_{\bfu,\bfx_i,\bfx_{[L+1]\setminus\curbrkt{i}},\bfs_i,\bfy}}{P_\bfu P_{\bfx_i|\bfu}P_{\bfx_{[L+1]\setminus\curbrkt{i}},\bfs_i|\bfu}W_{\bfy|\bfx,\bfs}} \notag \\
\ge& \kl{P_{\bfu,\bfx_i,\bfx_{[L+1]\setminus\curbrkt{i}},\bfy}}{P_\bfu P_{\bfx_i|\bfu}V_{\bfx_{[L+1]\setminus\curbrkt{i}},\bfy|\bfu,\bfx_i} } . \notag
\end{align}
where 
\begin{align}
V_{\bfx_{[L+1]\setminus\curbrkt{i}},\bfy|\bfu,\bfx_i} \coloneqq& \sqrbrkt{ P_{\bfx_{[L+1]\setminus\curbrkt{i}},\bfs_i|\bfu}W_{\bfy|\bfx,\bfs}}_{\bfx_{[L+1]\setminus\curbrkt{i}},\bfy|\bfu,\bfx_i}. \notag
\end{align}
By Pinsker's inequality (\Cref{lem:pinsker}), 
\begin{align}
2\sqrt{\ln2}\sqrt{\eta}\ge \normone{P_{\bfu,\bfx_i,\bfx_{[L+1]\setminus\curbrkt{i}},\bfy} - P_\bfu P_{\bfx_i|\bfu} V_{\bfx_{[L+1]\setminus\curbrkt{i}},\bfy|\bfu,\bfx_i}  }. \label{eqn:bound_i}
\end{align}
The same bound as \Cref{eqn:bound_i} with $i$ replaced by $i'\in[L+1]\setminus\curbrkt{i} $ still holds. 
Adding up both sides of these two bounds and applying triangle inequality on the RHS, we obtain
\begin{align}
&4\sqrt{\ln2}\sqrt{\eta} \ge \normone{P_\bfu P_{\bfx_i|\bfu} V_{\bfx_{[L+1]\setminus\curbrkt{i}},\bfy|\bfu,\bfx_i}  - P_\bfu P_{\bfx_{i'}|\bfu} V_{\bfx_{[L+1]\setminus\curbrkt{i'}},\bfy|\bfu,\bfx_{i'}} } \notag \\
\ge& p_u^*\cardU \normone{P_{\bfx_i|\bfu = u^*} V_{\bfx_{[L+1]\setminus\curbrkt{i}},\bfy|\bfu = u^*,\bfx_i} - P_{\bfx_{i'}|\bfu = u^*} V_{\bfx_{[L+1]\setminus\curbrkt{i'}},\bfy|\bfu = u^*,\bfx_{i'}}} \label{eqn:def_star} \\
=& p_u^*\cardU \sum_{x_{[L+1]}\in\cX^{L+1}}\sum_{y\in\cY} \left| P_{\bfx|\bfu}(x_i|u^*)V_{\bfx_{[L+1]\setminus\curbrkt{i}},\bfy|\bfu,\bfx_i}\paren{ x_{[L+1]\setminus\curbrkt{i}},y\condon u^*,x_i } \right. \notag \\
 & - \left.  P_{\bfx|\bfu}(x_{i'}|u^*) V_{\bfx_{[L+1]\setminus\curbrkt{i'}},\bfy|\bfu,\bfx_{i'}} \paren{ x_{[L+1]\setminus\curbrkt{i'}}, y\condon u^*,x_{i'} } \right| \label{eqn:same_marginal} \\
=& p_u^*\cardU \frac{1}{(L+1)!}\sum_{\pi\in S_{L+1}} \sum_{x_{[L+1]}\in\cX^{L+1}}\sum_{y\in\cY} \left| P_{\bfx|\bfu}(x_i|u^*)V_{\bfx_{{[L+1]} \setminus \curbrkt{\pi(i)}},\bfy|\bfu,\bfx_{\pi(i)}} \paren{x_{\pi\paren{[L+1]} \setminus\curbrkt{i}},y\condon u^*,x_i} \right. \notag \\
&  - \left. P_{\bfx|\bfu}(x_{i'}|u^*)V_{\bfx_{[L+1]\setminus\curbrkt{\pi(i')}}, \bfy|\bfu,\bfx_{\pi(i')}}\paren{ x_{\pi\paren{[L+1]}\setminus\curbrkt{i'}},y\condon u^*,x_{i'} } \right| \label{eqn:perm} \\
=& p_u^*\cardU \frac{1}{(L+1)!}\sum_{\pi\in S_{L+1}} \sum_{x_{[L+1]}\in\cX^{L+1}}\sum_{y\in\cY}
 \left| \sum_{s\in\cS}\left( P_{\bfx|\bfu}(x_i|u^*) P_{\bfx_{[L+1]\setminus\curbrkt{\pi(i)}}, \bfs_{\pi(i)}| \bfu}\paren{x_{\pi\paren{[L+1]}\setminus\curbrkt{i}}, s\condon u^*} W_{\bfy|\bfx,\bfs}(y|x_i,s)  \right.  \right. \notag \\
 & -\left. \left. P_{\bfx|\bfu}(x_{i'}|u^*) P_{\bfx_{[L+1]\setminus\curbrkt{\pi(i')}}, \bfs_{\pi(i')}| \bfu}\paren{x_{\pi\paren{[L+1]}\setminus\curbrkt{i'}}, s\condon u^*} W_{\bfy|\bfx,\bfs}(y|x_{i'},s) \right) \right| \notag \\
\ge&  p_u^*\cardU \sum_{x_{[L+1]}\in\cX^{L+1}}\sum_{y\in\cY}
 \left| \sum_{s\in\cS}\left( P_{\bfx|\bfu}(x_i|u^*) Q_{\bfx_{[L]},\bfs |\bfu }\paren{x_{{[L+1]}\setminus\curbrkt{i}}, s\condon u^*} W_{\bfy|\bfx,\bfs}(y|x_i,s)  \right.  \right. \notag \\
 & -\left. \left. P_{\bfx|\bfu}(x_{i'}|u^*) Q_{\bfx_{[L]},\bfs |\bfu }\paren{x_{{[L+1]}\setminus\curbrkt{i'}}, s\condon u^*} W_{\bfy|\bfx,\bfs}(y|x_{i'},s) \right) \right| \label{eqn:def_q}
\end{align}
where in \Cref{eqn:def_star} we use the following notation
\begin{align}
p_u^* \coloneqq& \min_{u\in\cU}P_\bfu(u), \;
u^* \coloneqq \argmin{u\in\cU} \normone{  P_{\bfx|\bfu = u} V_{\bfx_{[L+1]\setminus\curbrkt{i}},\bfy|\bfu = u,\bfx_i} -  P_{\bfx|\bfu = u} V_{\bfx_{[L+1]\setminus\curbrkt{i'}},\bfy|\bfu = u,\bfx_{i'}} }. \notag 
\end{align}
In \Cref{eqn:def_q} we define
\begin{align}
Q_{\bfx_{[L]},\bfs |\bfu } \paren{x_{[L]},s\condon u} \coloneqq& \frac{1}{(L+1)!} \sum_{\pi'\in S_L} \sum_{j\in[L+1]}  P_{\bfx_{[L+1] \setminus \curbrkt{j}},\bfs_{j} |\bfu } \paren{ x_{\pi'\paren{[L]}}, s\condon u } , \notag
\end{align}
which is due to the following identity:
\begin{align}
Q_{\bfx_{[L]}, \bfs|\bfu} \paren{ x_{{[L+1]}\setminus\curbrkt{i}}, s\condon u^* } =& \frac{1}{(L+1)!} \sum_{ \pi\in S_{L+1} } P_{\bfx_{[L+1]\setminus\curbrkt{\pi(i)}}, \bfs_{\pi(i)}|\bfu} \paren{ x_{\pi\paren{[L+1]}\setminus\curbrkt{i}}, s\condon u^* } \notag \\
=& \frac{1}{(L+1)!} \sum_{\pi'\in S_L}\sum_{j\in[L+1]}  P_{\bfx_{[L+1] \setminus \curbrkt{j}},\bfs_{j} |\bfu } \paren{ x_{\pi'\paren{[L+1]\setminus\curbrkt{i}}}, s\condon u^* }. \notag 
\end{align}
\Cref{eqn:same_marginal} follows since $ P_{\bfu,\bfx_i} = P_{\bfu,\bfx} $ for all $ i\in[L+1] $. 
\Cref{eqn:perm} follows since \Cref{eqn:same_marginal} is invariant under any permutation $ \pi\in S_{L+1} $. 
\Cref{eqn:def_q} follows from triangle inequality. 

We observe that $ Q_{\bfx_{[L]},\bfs |\bfu } $ is symmetric in $ \bfx_{[L]} $. 
Indeed, for any $ \sigma\in S_L $, 
\begin{align}
Q_{\bfx_{[L]},\bfs |\bfu } \paren{x_{\sigma\paren{[L]}},s\condon u}=& \frac{1}{(L+1)!} \sum_{\pi'\in S_L} \sum_{j\in[L+1]}  P_{\bfx_{[L+1] \setminus \curbrkt{j}},\bfs_{j} |\bfu } \paren{ x_{\sigma\paren{\pi'\paren{[L]}}}, s\condon u } \notag \\
=& \frac{1}{(L+1)!} \sum_{\pi'\in S_L} \sum_{j\in[L+1]}  P_{\bfx_{[L+1] \setminus \curbrkt{j}},\bfs_{j} |\bfu } \paren{ x_{\paren{\sigma\pi'}\paren{[L]}}, s\condon u } \notag \\
=& Q_{\bfx_{[L]},\bfs |\bfu } \paren{x_{[L]},s\condon u}. \notag
\end{align}

We claim that \Cref{eqn:def_q} must be strictly positive. 
Otherwise,
\begin{align}
&&  &\sum_{s\in\cS} P_{\bfx|\bfu}(x_i|u^*) Q_{\bfx_{[L]},\bfs |\bfu }\paren{x_{{[L+1]}\setminus\curbrkt{i}}, s\condon u^*} W_{\bfy|\bfx,\bfs}(y|x_i,s)  \notag \\
&& =&\sum_{s\in\cS} P_{\bfx|\bfu}(x_{i'}|u^*) Q_{\bfx_{[L]},\bfs |\bfu }\paren{x_{{[L+1]}\setminus\curbrkt{i'}}, s\condon u^*} W_{\bfy|\bfx,\bfs}(y|x_{i'},s)   \label{eqn:contradiction} \\
&\implies & &\sum_{s\in\cS} P_{\bfx|\bfu}(x_i|u^*) Q_{\bfx_{[L]},\bfs |\bfu }\paren{x_{{[L+1]}\setminus\curbrkt{i}}, s\condon u^*}   =\sum_{s\in\cS} P_{\bfx|\bfu}(x_{i'}|u^*) Q_{\bfx_{[L]},\bfs |\bfu }\paren{x_{{[L+1]}\setminus\curbrkt{i'}}, s\condon u^*}  \notag \\
&\implies & & P_{\bfx|\bfu}(x_i|u^*) Q_{\bfx_{[L]} |\bfu }\paren{x_{{[L+1]}\setminus\curbrkt{i}} \condon u^*}   = P_{\bfx|\bfu}(x_{i'}|u^*) Q_{\bfx_{[L]} |\bfu }\paren{x_{{[L+1]}\setminus\curbrkt{i'}} \condon u^*} 
\label{eqn:prod_distr}
\end{align}
In fact, $Q$ satisfying \Cref{eqn:prod_distr} must be a product distribution $ Q_{\bfx_{[L]}|\bfu = u^* } = P_{\bfx|\bfu = u^* }^{\tl} $ and is obviously symmetric. 
This is follows from  \Cref{lem:prod_distr} which is stated at the end of this section.
A proof can be found in \cite{hughes-1997-list-avc}. 
Substituting $Q$ back to \Cref{eqn:contradiction}, we get
\begin{align}
& & &\sum_{s\in\cS} P_{\bfx|\bfu}(x_i|u^*)  Q_{\bfx_{[L]} |\bfu }\paren{x_{{[L+1]}\setminus\curbrkt{i}}\condon u^*} Q_{\bfs|\bfx_{[L]},\bfu }\paren{s\condon x_{{[L+1]}\setminus\curbrkt{i}}, u^*} W_{\bfy|\bfx,\bfs}(y|x_i,s)  \notag \\
&& =&\sum_{s\in\cS} P_{\bfx|\bfu}(x_{i'}|u^*) Q_{\bfx_{[L]} |\bfu }\paren{x_{{[L+1]}\setminus\curbrkt{i'}} \condon u^*} Q_{\bfs|\bfx_{[L]},\bfu }\paren{s\condon x_{{[L+1]}\setminus\curbrkt{i'}},  u^*} W_{\bfy|\bfx,\bfs}(y|x_{i'},s) \notag \\
&\implies & & \sum_{s\in\cS} P_{\bfx_{i}|\bfu }(x_i|u^*) P_{\bfx|\bfu }^{\tl}\paren{x_{[L+1]\setminus\curbrkt{i}}\condon u^*}  Q_{\bfs |\bfx_{[L]},\bfu }\paren{s\condon x_{[L+1]\setminus\curbrkt{i}}, u^*} W_{\bfy|\bfx,\bfs }(y|x_i,s) \notag \\
&&=& \sum_{s\in\cS} P_{\bfx_{i'}|\bfu }(x_{i'}|u^*) P_{\bfx|\bfu }^{\tl}\paren{x_{[L+1]\setminus\curbrkt{i'}}\condon u^*} Q_{\bfs |\bfx_{[L]},\bfu }\paren{s\condon x_{[L+1]\setminus\curbrkt{i'}}, u^*} W_{\bfy|\bfx,\bfs }(y|x_{i'},s) \notag \\
&\implies & & \sum_{s\in\cS} P_{\bfx|\bfu }^{\otimes(L+1)}(x_{[L+1]}|u^*) Q_{\bfs |\bfx_{[L]},\bfu }\paren{s\condon x_{[L+1]\setminus\curbrkt{i}}, u^*} W_{\bfy|\bfx,\bfs }(y|x_i,s) \notag \\
&&=& \sum_{s\in\cS} P_{\bfx|\bfu }^{\otimes(L+1)}(x_{[L+1]}|u^*) Q_{\bfs |\bfx_{[L]},\bfu }\paren{s\condon x_{[L+1]\setminus\curbrkt{i'}}, u^*} W_{\bfy|\bfx,\bfs }(y|x_{i'},s) \notag \\
&\implies & & \sum_{s\in\cS}  Q_{\bfs |\bfx_{[L]},\bfu }\paren{s\condon x_{[L+1]\setminus\curbrkt{i}}, u^*} W_{\bfy|\bfx,\bfs }(y|x_i,s) = \sum_{s\in\cS}  Q_{\bfs |\bfx_{[L]},\bfu }\paren{s\condon x_{[L+1]\setminus\curbrkt{i'}}, u^*} W_{\bfy|\bfx,\bfs }(y|x_{i'},s). \label{eqn:symm_id}
\end{align}
Note that
$Q_{\bfs |\bfx_{[L]},\bfu } = {Q_{\bfs ,\bfx_{[L]},\bfu }}/{Q_{\bfx_{[L]},\bfu }}$, 
and both $ Q_{\bfs ,\bfx_{[L]},\bfu }, Q_{\bfx_{[L]},\bfu }$ are symmetric in $ \bfx_{[L]} $.
Therefore $ Q_{\bfs |\bfx_{[L]},\bfu } $ is also symmetric in $ \bfx_{[L]} $.
Combining this observation with \Cref{eqn:symm_id}, we know that $Q_{\bfs |\bfx_{[L]},\bfu }\in\cU_{\obli,\symml{L}} $.

\begin{lemma}\label{lem:prod_distr}
Let $ L\in\bZ_{\ge2} $. 
If $ P_{\bfx}\in\Delta(\cX) $ and $ Q_{\bfx_{[L]}}\in\Delta(\cX^L) $ satisfy
\begin{align}
P_{\bfx}(x_i) Q_{\bfx_{[L]}  }\paren{x_{{[L+1]}\setminus\curbrkt{i}} }   = P_{\bfx}(x_{i'}|u) Q_{\bfx_{[L]}  }\paren{x_{{[L+1]}\setminus\curbrkt{i'}} } \notag 
\end{align}
for all $ i\ne i'\in[L] $ and $ x_{[L]}\in\cX^L $, then
$ Q_{\bfx_{[L]}} = P_{\bfx}^\tl $. 
\end{lemma}

\section{Strong converse for fading DMCs with approximate constant-composition codes and list-decoding (proof of \Cref{thm:strong_conv_fading})}
\label[app]{app:strong_conv_fading_pf}
Let $ \cC $ be a code of rate $ R = C(W_{\bfy|\bfx,\bfu})+\delta $ and let $L\in\bZ_{\ge1} $ be the list-size. 
Let $ P_\bfu = \tau_\vu $. 
Suppose that for some $ P_{\bfx|\bfu} $, for all $ \vx\in\cC $, $ d\paren{\tau_{\vu,\vx},P_\bfu P_{\bfx|\bfu}}\le\lambda $ where $ 0<\lambda\ll\delta $ is a constant. 
Let $ \psi\colon \cY^n\to\binom{[M]}{\le L} $ be the optimal list-decoder of $\cC$ used over $ W_{\bfy|\bfx,\bfu} $.

Let $\eps>0$ be a sufficiently small constant to be determined later. 
Define the $ \eps $-typical set as
\begin{align}
 \cA_{\vbfy|\vu}^\eps(P_{\bfy|\bfu}) \coloneqq& \curbrkt{\vy\in\cY^n\colon\forall u\in\cU,\forall y\in\cY,\; \frac{\tau_{\vy|\vu}(y|u)}{P_{\bfy|\bfu}(y|u)}\in[1-\eps,1+\eps] }, \notag 
\end{align}
where $ P_{\bfy|\bfu} = \sqrbrkt{P_{\bfx|\bfu}W_{\bfy|\bfx,\bfu}}_{\bfy|\bfu} $. 
Note that by the asymptotic equipartition property (\Cref{lem:aep})
\begin{align}
\card{\cA_{\vbfy|\vu}^\eps(P_{\bfy|\bfu})}\le2^{n(H(\bfy|\bfu) + f_2(\lambda,\eps))}, \notag
\end{align}
for some $ f_2(\lambda,\eps)>0 $.

We now lower bound the average error probability. 
\begin{align}
1-P_{\e,\avg}(\cC) =& \frac{1}{M}\sum_{i\in[M]}\prob{\psi(\vbfy)\ni i\condon \bfm = i,\vbfu = \vu} \notag \\
=& \frac{1}{M}\sum_{i\in[M]}\sum_{\vy\in\cY^n}W^\tn_{\bfy|\bfx,\bfu}(\vy|\vx_i,\vu)\indicator{\psi(\vy)\ni i} \notag \\
\le& \frac{1}{M}\sum_{i\in[M]}\sum_{\vy\in\cA_{\vbfy|\vu}^\eps(P_{\bfy|\bfu})}W^\tn_{\bfy|\bfx,\bfu}(\vy|\vx_i,\vu)\indicator{\psi(\vy)\ni i} + \frac{1}{M}\sum_{i\in[M]}\sum_{\vy\notin\cA_{\vbfy|\vu}^\eps(P_{\bfy|\bfu})} W^\tn_{\bfy|\bfx,\bfu}(\vy|\vx_i,\vu). \label{eqn:strong_two_terms}
\end{align}
We claim that $ \vx_i\in\cC $ satisfies $ d\paren{\tau_{\vx_i|u}, P_{\bfx|\bfu = u}}\le\lambda' $ for all $ u\in\cU $, where $ \lambda' = \frac{\lambda}{\cardU p_u^*} $.
This is guaranteed by approximate constant-composition of $\cC$.
Indeed for any $ \vx\in\cC $,
\begin{align}
\lambda\ge& d\paren{\tau_{\vu,\vx}, P_\bfu P_{\bfx|\bfu}} \notag \\
=& \sum_{(u,x)\in\cU\times\cX}\abs{ \tau_{\vu,\vx}(u,x) - P_\bfu P_{\bfx|\bfu}(u,x) } \notag \\
=& \sum_{(u,x)\in\cU\times\cX} \abs{ \tau_\vu(u) \tau_{\vx|\vu}(x|u) - P_\bfu(u) P_{\bfx|\bfu}(x|u) } \notag \\
=& \sum_{(u,x)\in\cU\times\cX} P_\bfu(u) \abs{  \tau_{\vx|\vu}(x|u) -  P_{\bfx|\bfu}(x|u) } \notag \\
=& \sum_{u\in\cU} P_\bfu(u)\sum_{x\in\cX}\abs{ \tau_{\vx|u}(x) - P_{\bfx|\bfu = u}(x) } \notag \\
=& \sum_{u\in\cU} P_\bfu(u) d\paren{\tau_{\vx|u}, P_{\bfx|\bfu = u}} \notag \\
\ge& \cardU p_u^* d\paren{\tau_{\vx|u}, P_{\bfx|\bfu = u}}, \label{eqn:def_pu_min}
\end{align}
where \Cref{eqn:def_pu_min} holds for any $u\in\cU$ and $ p_u^* $ is defined as $ p_u^*\coloneqq\min_{u\in\cU} P_\bfu(u)>0 $ (since $ P_\bfu = \tau_\vu $ has no zero atoms).

Hence for any $ \vy\in\cA_{\vbfy|\vu}^\eps(P_{\bfy|\bfu}) $, 
\begin{align}
W^\tn_{\bfy|\bfx,\bfu}(\vy|\vx_i,\vu)\le& 2^{-n(H(\bfy|\bfx,\bfu) - f_1(\lambda,\eps))}, \notag 
\end{align}
for some $ f_1(\lambda,\eps)>0 $. 
Thus the first term in \Cref{eqn:strong_two_terms} is at most
\begin{align}
\frac{1}{M}2^{-n(H(\bfy|\bfx,\bfu) - f_1(\lambda,\eps))}\sum_{\vy\in \cA_{\vbfy|\vu}^\eps(P_{\bfy|\bfu})}\sum_{i\in[M]}\indicator{\psi(\vy)\ni i} 
\le& \frac{1}{M}2^{-n(H(\bfy|\bfx,\bfu) - f_1(\lambda,\eps))}\card{\cA_{\vbfy|\vu}^\eps(P_{\bfy|\bfu})}L \notag \\ 
\le& L^{-1}2^{-nR}2^{-n(H(\bfy|\bfx,\bfu) - f_1(\lambda,\eps))}2^{n(H(\bfy|\bfu) + f_2(\lambda,\eps))}L \notag \\
\le& 2^{n(-C +I(\bfx;\bfy|\bfu) -\delta + f_1(\lambda,\eps )+ f_2(\lambda,\eps))} \notag \\
\le& 2^{-n(\delta - f_1(\lambda,\eps) - f_2(\lambda,\eps))}. 
\label{eqn:first_term}
\end{align}
Let $ \eps $ be sufficiently small so that $ \delta > f_1(\lambda,\eps) + f_2(\lambda,\eps) $.
Then the  bound in \Cref{eqn:first_term} is exponentially decaying. 

As for the second term in \Cref{eqn:strong_two_terms}, by the large deviation principle, we have that for any $ \vx_i $ satisfying $ d\paren{\tau_{\vx_i|u}, P_{\bfx|\bfu = u}}\le\lambda' = \frac{\lambda}{\cardU p_u^*} $ for all $ u\in\cU $ (which is guaranteed by approximate constant-composition of $\cC$), 
\begin{align}
\prob{\vbfy\notin \cA_{\vbfy|\vu}^\eps(P_{\bfy|\bfu})\condon\bfm = i,\vbfu = \vu} \le& 2^{-nf_3(\lambda,\eps)}, \label{eqn:second_term} 
\end{align}
for some $ f_3(\lambda,\eps)>0 $. 

Putting \Cref{eqn:first_term} and \Cref{eqn:second_term} back to \Cref{eqn:strong_two_terms}, we obtain
\begin{align}
P_{\e,\avg}(\cC) \ge& 1 - 2^{-n(\delta - f_1(\lambda,\eps) - f_2(\lambda,\eps))} -  2^{-nf_3(\lambda,\eps)}\xrightarrow{n\to\infty}1, \notag 
\end{align}
which finishes the proof of \Cref{thm:strong_conv_fading}.

\end{document}